\documentclass[10pt]{article}

\usepackage[english]{babel}
\usepackage[letterpaper,top=2cm,bottom=2cm,left=3cm,right=3cm,marginparwidth=1.75cm]{geometry}
\usepackage[utf8]{inputenc}
\usepackage[T1]{fontenc}
\usepackage{setspace}
\usepackage{amsmath,amsthm,amssymb,amsfonts,amscd,keyval}
\usepackage{mathtools,mathrsfs}
\mathtoolsset{mathic=true}
\usepackage{bbm}
\usepackage{adjustbox}
\usepackage{tensor}
\usepackage{braket}
\usepackage{slashed}
\usepackage{tikz-cd}
\usepackage{circuitikz}
\usepackage{multirow}

\usepackage{graphicx}
\usepackage{subcaption}
\usepackage{float}
\usepackage{abstract}
\usepackage{titlesec}
\usepackage{titletoc}
\usepackage[backend=biber,  style=ieee,
  citestyle=numeric-comp, maxcitenames=3, maxbibnames=99 ]{biblatex}
\usepackage[colorlinks=true, allcolors=blue]{hyperref}
\usepackage{microtype} 
\usepackage{url}
\usepackage{tikz}
\usetikzlibrary{arrows.meta,calc}
\usepackage{authblk}

\usepackage{palatino}
\usepackage{booktabs,array,tabularx}

\let\originalmathbb\mathbb
\renewcommand{\mathbb}[1]{%
  \ifstrequal{#1}{k}{\Bbbk}{\originalmathbb{#1}}}

\newcommand{\comments}[1]{}

\DeclareMathOperator{\im}{im}

\DeclareMathOperator{\tr}{Tr}

\newcommand{\CZ}{\mathrm{C}Z}

\newcommand{\CCZ}{\mathrm{C}\mathrm{C}Z}
\newcommand{\mCZ}{\mathrm{C}^{r-1}Z}

\DeclareMathOperator{\type}{\text{type}}
\DeclareMathOperator{\1}{\textbf{1}}

\DeclareMathOperator{\x}{\mathrm{x}}
\DeclareMathOperator{\sd}{sd}
\DeclareMathOperator{\id}{id}

\DeclareMathOperator{\ev}{ev}

\numberwithin{equation}{section} 
\theoremstyle{definition}
\newtheorem{definition}{Definition}[section]

\theoremstyle{plain}

\newtheorem{theorem}[definition]{Theorem}
\newtheorem{proposition}[definition]{Proposition}
\newtheorem{lemma}[definition]{Lemma}

\title{Parallelizable and addressable transversal non-Clifford gates on good quantum LDPC codes}
\author{Yiming Li}
\author{Zimu Li}
\author{Zhengyi Han}
\author{Zi-Wen Liu}
\affil{Yau Mathematical Sciences Center, Tsinghua University}
\date{}
\usepackage{xcolor}
\usepackage{etoolbox}
\usepackage{csquotes}
\usepackage{needspace}
\AtBeginBibliography{\setlength{\emergencystretch}{3em}}

\newif\ifshoworiginalsection
\showoriginalsectiontrue
\begin{document}
\begingroup

\maketitle

\vspace{-1cm}

\begin{abstract}

We achieve linearly many parallelizable and addressable transversal logical multi-controlled-$Z$ gates with asymptotically optimal parameters simultaneously on good quantum low-density parity check codes and quantum locally testable codes, by applying the gate framework of [arXiv:2604.01874] to  good qLTC constructions.
To establish the nontriviality of the logical operation, we express the cup product pairing as a coefficient in a product of Moore determinants and prove that this coefficient is nonzero. Moreover, we use the symmetries of the covering spaces to construct large spaces of cocycles. By multiplying the initial cocycles by suitable elements of these spaces, we obtain linearly many independent pairings, each selected by a corresponding cycle constructed using Poincaré duality. This yields linearly many parallelizable and addressable multi-controlled-$Z$ gates and enables asymptotically constant-overhead magic state distillation with good qLDPC codes.

\end{abstract}

\vspace{-4mm}
\tableofcontents


\section{Introduction and main results}\label{sec:intro}

Quantum error correction must enable computation as well as protect stored quantum information. Beyond rate, distance, and the cost of syndrome extraction, an essential question is therefore which logical operations a code can implement fault-tolerantly. Non-Clifford gates are especially important: they are necessary for computational universality beyond Clifford operations, and their implementation through magic state preparation and distillation often accounts for a substantial fraction of the overhead of fault-tolerant quantum computation~\cite{Yamasaki2024QusiPloylog,Tamiya2024PolylogTime,Wills2024magic,Nguyen2025CCZ}. Transversal gates provide a particularly simple way to control error propagation. Constructing codes that directly support non-Clifford logical operations while retaining strong error-correction parameters is consequently a central problem in quantum coding theory.

Quantum low-density parity check (qLDPC) codes underpin the prospect of fault-tolerant quantum computation with constant overhead~\cite{Gottesman2013}, making the construction of qLDPC codes with optimal asymptotic parameters a central goal of recent research. After decades of progress, asymptotically good qLDPC codes, with constant rate and linear distance, were established through lifted product and quantum Tanner constructions~\cite{PK2022Good,QuantumTanner2022,DHLV2022}. Quantum locally testable codes (qLTCs) impose the further requirement that local checks detect errors with probability proportional to their distance from the code space~\cite{AharonovEldar2015QLTC}. This quantitative relation between local violations and global distance is important both for quantum fault tolerance and for the study of quantum PCPs and the complexity of low-energy states~\cite{Nguyen2025FT,AharonovAradVidick2013_qPCPSurvey,Eldar2016NLETS}. These coding advances make it pressing to understand whether strong parameters are compatible with useful fault-tolerant logical gates, especially non-Clifford gates.

The quest for this compatibility has driven sustained research over the years. Color codes and gauge color codes provided influential constructions of fault-tolerant non-Clifford operations~\cite{Bombin_2007,Bombin_2013,Kubica2015}, while subsequent work has explored generalized color codes, homological products, higher-dimensional topology, and cup products on sheaf codes~\cite{RainbowCode,10.21468/SciPostPhys.14.4.065,Zhu2023,Wang_2024,Lin2024transversal,Breuckmann2024Cups,Tiew2026copycup,Zhu2025A,Zhu2025B}. A complementary algebraic approach produced asymptotically good quantum codes with transversal non-Clifford gates, including increasingly strong guarantees of addressability and parallelism, without bounded-weight stabilizer checks~\cite{Nguyen2025CCZ,Wills2024magic,He2025addressable,He2025GoodAddressable,Golowich_Guruswami2025A,VirgileCCZ,GasnierGuemard2026}. Recent advances include improved binary codes with transversal $T$ gates~\cite{Wills2026TransversalT} and asymptotically good binary codes with higher-level transversal gates~\cite{SanJose2026Triorthogonal}. Constructions based on products of algebraic codes and related approaches further improved the tradeoffs among rate, distance, and check weight~\cite{Golowich_Lin2024,Golowich_Guruswami2025B, GolowichTamoZhu2026}. At the same time, no-go theorems for product codes identified structural obstructions to strictly transversal non-Clifford gates~\cite{Burton2022,Fu2025nogo}. Together, these results show that efficient logical operations depend on algebraic structure that is not determined by code parameters alone.

Recent work brought this problem into the regime of nearly optimal parameters. The theory of cohomological invariants on sheaf codes provides a systematic framework for constructing logical representatives and evaluating cup products~\cite{Li2025Poincare,LSWLL2026Theory}. Building on this framework, Li--Li--Liu established nontrivial transversal logical multi-controlled-$Z$ gates on almost-good qLDPC codes and qLTCs~\cite{LLL2026nontrivial}, using the cubical construction of Dinur--Lin--Vidick~\cite{Dinur2024sheaf} and the theory of extendable code of Panteleev and Kalachev~\cite{Panteleev2024,KP2025Extendable}. These codes have constant rate, with only polylogarithmic losses in distance and, for qLTCs, soundness. That is, the compatibility of sparsity, nearly optimal parameters, and transversal non-Clifford gates has been established, but achieving exactly linear distance and constant soundness simultaneously with these gates remains an outstanding goal.

The recent constructions of asymptotically good qLTCs \cite{GJ2026,BLN2026,CHLT2026,Cohen2026qLTC_lossless} provide the desired setting that supports good code parameters and constant soundness in high dimensional expanders. However, retaining nontrivial logical operations in this setting requires an additional argument. The product structure that made earlier cup product calculations accessible is replaced by arithmetic identifications that change the local coordinates, and compatibility of the local codes does not by itself produce nonzero logical pairings. Removing the polylogarithmic parameter losses yet preserving the non-Clifford logical operation therefore require more efforts. 



\paragraph{Main results.}
In this work, we realize nontrivial transversal logical multi-controlled-$Z$ gates on asymptotically good qLDPC codes and qLTCs by extending the gate framework of Ref.~\cite{LLL2026nontrivial} to the recent asymptotically good qLTC construction of Ref.~\cite{GJ2026,BLN2026,CHLT2026}. At the conceptual level, the extension exploits the close structural similarities between the asymptotically almost-good and good qLTCs: their global structures are based on abelian lifts~\cite{Dinur2024sheaf} and Ramanujan cubical complexes~\cite{GJ2026,BLN2026,CHLT2026}, respectively, while their local codes are product-expanding punctured Reed--Solomon (RS) codes~\cite{LLL2026nontrivial} and projective Reed--Solomon codes~\cite{GJ2026,CHLT2026}, respectively. The main difficulty lies in constructing compatible cohomology classes on the new complexes realized by arithmetic quotients rather than tensor products, and establishing the nontriviality of the resulting cup products and pairing.

As is standard, $\mCZ$ denotes a (multi-)controlled-$Z$ gate on $r$ qubits with $r-1$ controls.  We prove:

\begin{theorem}\label{thm:main}
For any integer $r\geq 3$, there exist  
\begin{itemize}
    \item $[\![N,\Theta(N),\Theta(N)]\!]$ quantum LDPC codes
    \item $[\![N,\Theta(N),\Theta(N)]\!]$ quantum locally testable codes with soundness $\Theta(1)$
\end{itemize}
that support $k_{\mCZ} = \Theta(N)$ parallelizable and addressable  nontrivial transversal logical $\mCZ$ gates.
\end{theorem}



{This combines linearly many (asymptotically optimal) parallelizable and individually addressable transversal non-Clifford gates together with asymptotically good encoding rate and distance on qLDPC codes and qLTCs which further has constant soundness. Furthermore, the linear subrank indicates that the codes support magic state distillation with asymptotically constant overhead. 
That is, we simultaneously achieve optimality in all the following desirable properties, each or a subset of which has been the goal of substantial research efforts: qLDPC-ness, constant encoding rate, linear distance, constant soundness, and finally linear parallelizable and addressable transversal non-Clifford gates and constant-overhead magic state distillation.}

{In this paper we shall present combinatorial proofs of the main claimed results. We will flesh out a mathematical theory concerning these results in a separate forthcoming work.}

\paragraph{Proof overview.}
We now outline how the framework of Ref.~\cite{LLL2026nontrivial} enables nontrivial transversal logical multi-controlled-$Z$ gates on asymptotically good qLDPC codes and qLTCs. For simplicity, we only compare the code constructions of Refs.~\cite{Dinur2024sheaf} and \cite{GJ2026} and refer to them as the DLV and GJ codes, respectively.

The main mechanism can be summarized as follows. Suppose $X$ is a sparse cell complex with three sheaves of $\mathbb F_q$-vector spaces $\mathcal{F}_1, \mathcal{F}_2$ and $\mathcal F_3$, where $\mathbb F_q$ has characteristic $2$. If there exist cocycles $\alpha_1\in C^i(X,\mathcal{ F}_1), \alpha_2\in C^j(X,\mathcal{ F}_2), \alpha_3\in C^k(X,\mathcal{ F}_3)$ and a cycle $\xi\in C_{i+j+k}(X,\mathcal{ F}_1\otimes\mathcal F_2\otimes\mathcal F_3)$ such that
\begin{align}\label{eq:nontrivial triple cup pairing}
 \langle\alpha_1\smile\alpha_2\smile\alpha_3,\xi\rangle\ne 0,
\end{align}
then the following unitary yields a nontrivial transversal logical $\CCZ$ gate across the CSS code blocks associated with $C^i(X,\mathcal{ F}_1),C^j(X,\mathcal{ F}_2)$ and $C^k(X,\mathcal{ F}_3)$ :
\begin{equation}\label{eq:invariant_poly}
    U_\xi
    \coloneqq
    \sum_{x_1,x_2,x_3}
    (-1)^{
        \operatorname{Tr}_{\mathbb F_q/\mathbb F_2}
        \left(
            \left\langle
                x_1\smile x_2\smile x_3,\xi
            \right\rangle
        \right)
    }
    \ket{x_1,x_2,x_3}\bra{x_1,x_2,x_3},
\end{equation}
where the sum below ranges over $x_1\in C^i(X,\mathcal F_1)$, $x_2\in C^j(X,\mathcal F_2)$, and $x_3\in C^k(X,\mathcal F_3)$.

However, cup products can be difficult to compute in general, especially when the cell complex $X$ arises from an intricate algebraic construction. Since it suffices to find a particular collection of (co)cycles satisfying Eq.~\eqref{eq:nontrivial triple cup pairing}, we develop the following covering space method. Suppose $X$ is the covering space of $X_0$ of sheet $\ell$, let $P:X\to X_0$ be the covering map and $T:X_0\to X$ be the transfer map defined in \cite{LLL2026nontrivial}. Suppose $X_0$ is equipped with three sheaves $\mathcal{ G}_1, \mathcal{G}_2$ and $\mathcal G_3$ such that $\mathcal F_i=P^*\mathcal G_i$, $i\in[3]$ are the pullback sheaves, then if there exist cocycles $\beta_1\in C^i(X_0,\mathcal G_1), \beta_2\in C^j(X_0,\mathcal G_2)$ and $\beta_3\in C^k(X_0,\mathcal G_3)$ and a cycle $\eta\in C_{i+j+k}(X_0,\mathcal G_1\otimes \mathcal G_2\otimes\mathcal G_3)$, then
\begin{align}
    \langle P^\#\beta_1\smile P^\#\beta_2\smile P^\#\beta_3,T_\#\eta\rangle=\ell\langle \beta_1\smile\beta_2\smile\beta_3,\eta\rangle.
\end{align}
Therefore, the problems is reduced to find (co)cycles on $X_0$ such that $\langle \beta_1\smile\beta_2\smile\beta_3,\eta\rangle\ne 0$ as long as $\ell$ is an odd number.
 
To compare how this covering space strategy is implemented in the DLV and GJ constructions, we describe their underlying cubical complexes in the common language of decorated Cayley cubical complexes~\cite{hsieh2025explicitlosslessvertexexpanders}.


Let $G$ be a finite set and let $\{A_1,\ldots,A_t\}$ be a collection of inverse-closed subsets of permutations of $G$, which is called a collection of \emph{cubical generating sets} if $A_iA_j=A_jA_i$ for all $i\neq j$, and $|A_1\cdots A_t|=|A_1|\cdots|A_t|$. The associated decorated Cayley cubical complex is denoted by $X=\operatorname{Cay} \bigl(G;(A_1,\ldots,A_t)\bigr)$. Its vertex set is $X(0)=G\times\mathbb F_2^t.$ A $t$-dimensional cube is a collection $f =\left\{ (f_x,x):x\in\mathbb F_2^t\right\}\subseteq G\times\mathbb F_2^t$ such that, for every $x\in\mathbb F_2^t$ and every $i\in[t]$, $f_x^{-1}f_{x\oplus e_i}\in A_i$. The lower-dimensional faces are obtained by restricting $f$ to subcubes of $\mathbb F_2^t$.

In the DLV construction, the set and generating subsets are given by 
\begin{align}\label{eq:DLV_underlying_set}
    G^{\mathrm{DLV}}
    =
    C_\ell\times V_0^t,
    \qquad
    A_i^{\mathrm{DLV}}
    =
    \left\{
        a_i^1,\ldots,a_i^n
    \right\}
    \subseteq
    \operatorname{Sym}(G^{\mathrm{DLV}}),
\end{align}
where $C_\ell$ is the cyclic group of order $\ell$, $V_0$ is the vertex set of some expander Cayley graph $G_0$, and each $A_i^{\mathrm {DLV}}$ can be identified as the generating set of the corresponding Cayley graph. The additional requirement is that permutations belonging to distinct directions commute
\emph{elementwise}:
\begin{align}\label{eq:DLV_elementwise_commutation}
    a_i^\mu a_j^\nu
    =
    a_j^\nu a_i^\mu
    \qquad
    \text{for all }
    i\neq j,\ 
    a_i^\mu\in A_i^{\mathrm{DLV}},\
    a_j^\nu\in A_j^{\mathrm{DLV}}.
\end{align}

This elementwise commutation condition allows a flexible choice of local codes for defining sheaves on the complex.
Notice that the cyclic group $C_\ell$ induce a group action on $X^{\mathrm{DLV}}$, which is defined to be $\operatorname{Cay} \bigl(G^{\mathrm{DLV}};(A_1^{\mathrm{DLV}},\ldots,A_t^{\mathrm{DLV}})\bigr)$. Then this group action gives a covering map:
\begin{align}
    X^{\mathrm{DLV}}\to X^{\mathrm{DLV}}/C_\ell\coloneq X_0^{\mathrm{DLV}}.
\end{align}
The key property is that $X_0^{\mathrm{DLV}}$ is the $t$-fold HGP of the double cover of $G_0$, and we may find $t$ many two-way product-expanding local codes, each of them generated a sheaf defining a Sipser-Spielman code on the double cover of $G_0$. Then the HGP structure induce an external tensor product of sheaves on $X_0^{{\mathrm{DLV}}}$. The sheaf on $X^{{\mathrm{DLV}}}$ is exactly the pullback sheaf of this sheaf on $X^{\mathrm{DLV}}_0$.

This HGP structure of $X_0^{{\mathrm{DLV}}}$ makes it particularly easy to calculate cup products. Consider $t=3$ for example. We choose three 1-cochains $\beta_1,\beta_2,\beta_3$, where $\beta_i$ is supported only on edges of type $i$. Each cochain is obtained by choosing an edge cochain in its active direction and periodically repeating suitable constant vector in the other two directions. This guarantees that these cochains are cocycles, and their cup product supported on a single cube. Then by choosing appropriate punctured Reed--Solomon local codes gives a 3-cycle $\eta$ nonvanishing on every cube, which gives $ \left\langle\beta_1\smile\beta_2\smile\beta_3,\eta\right\rangle\neq 0$.

Good qLTCs in this framework require both global spectral expansion of the underlying complex and local coboundary expansion of the associated sheaf, with the latter relying on the acyclicity of its local cochain complexes. In the DLV construction, local acyclicity follows automatically from the elementwise commutativity of the directional permutations. However, enforcing this commutativity through the abelian $C_\ell$-lift limits the global spectral expansion and leads to a polylogarithmic loss in distance and soundness.Ramanujan cubical complexes provide stronger global expansion, but their directional permutation sets commute only setwise, so local acyclicity now depends crucially on the choice of local codes.

Panteleev and Kalachev therefore proposed using projective Reed--Solomon codes, whose projective-linear symmetries match the local permutation actions of Ramanujan cubical complexes \cite{Panteleev2024}. The remaining obstacle was two-way product expansion: the straightforward choice of using the same evaluation set in every direction does not provide product expansion simultaneously for the local codes and their duals. Gay and Jeronimo overcome this obstacle by using projective Reed--Solomon codes of different but comparable sizes in different directions and proving a new interpolation theorem for their products, thereby obtaining asymptotically good qLTCs~\cite{GJ2026,BLN2026,CHLT2026}.

We now describe the GJ construction using exactly the same decorated Cayley cubical complex language. For each $i\in[t]$, let $T_i$ be the $(q_i+1)$-regular Bruhat--Tits tree associated with the $i$-th split place. Combinatorially, $T_i$ is an infinite bipartite tree whose $q_i+1$ edges incident to each vertex are naturally indexed by $\mathbb P^1(\mathbb F_{q_i})$. For every point of $\mathbb P^1(\mathbb F_{q_i})$, the arithmetic construction provides a distinguished transformation that moves the base vertex to the corresponding neighbor in the $i$-th tree while leaving its other coordinates unchanged. The collection of these $q_i+1$ transformations is denoted by $A_i^{\mathrm{GJ}}$. Thus $A_i^{\mathrm{GJ}}$ is precisely the set of all elementary one-step moves in the $i$-th direction, rather than an abstractly chosen generating set. The group $\Lambda^{\mathrm{GJ}}$ is defined to be the group generated by these elementary moves in all $t$ directions. Then the universal covering space of the Ramanujan cubical complex \cite{RSV2019} can be identified as
\begin{align}\label{eq:GJ_universal_decorated_cayley}
    \widetilde X^{\mathrm{GJ}}
    =
    \operatorname{Cay}\bigl(\Lambda^{\mathrm{GJ}};(A_1^{\mathrm{GJ}},\ldots,A_t^{\mathrm{GJ}})\bigr)
    \cong
    T_1\times\cdots\times T_t.
\end{align}

Each tree $T_i$ is bipartite. Fix one of its two vertex classes as type zero and the other as type one. A vertex of $\widetilde X^{\mathrm{GJ}}=T_1\times\cdots\times T_t$ therefore carries a type in $\mathbb F_2^t$, recording the bipartition class of each of its $t$ coordinates. Moving along an edge in direction $i$ flips the $i$th entry of the type and leaves all other entries unchanged. A subgroup of $\Lambda^{\mathrm{GJ}}$ is called type-preserving if each of its elements maps every vertex to a vertex of the same type.

To obtain a finite complex, let
$\Gamma^{\mathrm{GJ}}\leq\Lambda^{\mathrm{GJ}}$ be a finite-index
type-preserving subgroup and define the finite set
\begin{align}\label{eq:GJ_finite_state_set}
    G^{\mathrm{GJ}}
    =
    \Gamma^{\mathrm{GJ}}\backslash\Lambda^{\mathrm{GJ}}.
\end{align}
Every $a\in A_i^{\mathrm{GJ}}$ induces a permutation of this set by right multiplication, $\Gamma^{\mathrm{GJ}}g\mapsto\Gamma^{\mathrm{GJ}}ga$. We use the same notation $A_i^{\mathrm{GJ}}$ for the resulting set of permutations. The finite GJ complex is therefore
\begin{align}\label{eq:GJ_decorated_cayley_complex}
    X^{\mathrm{GJ}}
    =
    \operatorname{Cay}\bigl(G^{\mathrm{GJ}};(A_1^{\mathrm{GJ}},\ldots,A_t^{\mathrm{GJ}})\bigr)
    \cong
    \Gamma^{\mathrm{GJ}}\backslash\widetilde X^{\mathrm{GJ}}.
\end{align}

In particular, the arithmetic construction provides a subgroup $\Gamma_0^{\mathrm{GJ}}\leq\Lambda^{\mathrm{GJ}}$ such that two vertices of $\widetilde X^{\mathrm{GJ}}$ are related by an element of $\Gamma_0^{\mathrm{GJ}}$ if and only if they have the same type. We take as our base complex
\begin{align}\label{eq:GJ_base_complex}
    X_0^{\mathrm{GJ}}
    =
    \Gamma_0^{\mathrm{GJ}}\backslash\widetilde X^{\mathrm{GJ}}.
\end{align}
Thus $X_0^{\mathrm{GJ}}$ has exactly one vertex of each type and remains fixed throughout the code family. We choose $\Gamma^{\mathrm{GJ}}\leq\Gamma_0^{\mathrm{GJ}}$, so the subgroup inclusion induces a covering map
\begin{align}\label{eq:GJ_covering}
    P^{\mathrm{GJ}}:X^{\mathrm{GJ}}\longrightarrow X_0^{\mathrm{GJ}}.
\end{align}
The sheet of the covering map is exactly $[\Gamma_0^{\mathrm{GJ}}:\Gamma^{\mathrm{GJ}}]$, which is chosen to be odd throughout our construction.

The essential distinction is that the GJ generating sets commute only as sets. More precisely, for every
$a_i\in A_i^{\mathrm{GJ}}$ and $a_j\in A_j^{\mathrm{GJ}}$, there are unique $a_i'\in A_i^{\mathrm{GJ}}$ and $a_j'\in A_j^{\mathrm{GJ}}$ such that $a_i a_j=a_j'a_i'$.

The sheaves in the DLV and GJ constructions are defined by essentially the same local rule. In each direction, we choose a local code on the set of edges incident to a vertex. Informally, the stalk assigned to a face is the tensor product of the local code spaces associated with the directions transverse to that face, while the restriction maps are obtained by evaluating the corresponding local-code factors at the labels of the added edges.

For the DLV complex, permutations in different directions commute elementwise, so opposite edges of every square carry exactly the same labels. The two ways of restricting across a square therefore agree automatically. In the GJ complex, opposite edges may instead carry different labels, related by a projective-linear change of coordinates. The crucial observation is that projective Reed--Solomon codes are equivariant under these projective-linear transformations: relabeling the evaluation points induces a compatible transformation of the local code. Consequently, although the edge labels change around a square, the corresponding restriction maps still agree. Therefore, the restriction maps obtained from the projective Reed--Solomon codes are independent of the order in which one passes across the faces of a cube, and hence define a sheaf on $X_0^{\mathrm{GJ}}$. The sheaf on $X^{\mathrm{GJ}}$ is then obtained simply by pulling back this sheaf along the covering map $P^{\mathrm{GJ}}$.

We now describe the nontrivial cup product in the similar three-dimensional degree-one case. As in the DLV construction, we choose three cocycles $\beta_1,\beta_2,\beta_3$, where $\beta_i$ is supported only on edges of type $i$. The absence of an HGP decomposition does not prevent such polarized supports. It only means that the value of $\beta_i$ cannot be copied unchanged from one edge to another. Instead, we start from one projective Reed--Solomon word and propagate it to all edges of type $i$ using the projective-linear identifications of the sheaf. The compatibility of these identifications around every square ensures that the resulting cochain is a cocycle. Similarly, the same propagation method produces a three-cycle $\eta$ whose pairing with the cup product of the three cocycles is nonzero. Still, this idea can be generalized to obtain multi-controlled-$Z$ gates.

We further strengthen the nontriviality result to linearly addressable and
parallelizable multi-controlled-$Z$ gates. We explain the idea with $\CCZ$ and generalizations to higher-order gates are straightforward.
The main difficulty is that a dimension counting argument guarantees
constant rate for only one of the three code blocks, while it gives no rate
estimate for the other blocks when their local codes do not satisfy the
relevant criterion. We avoid this issue by exploiting symmetries of the
cochain complex to amplify the number of (co)cycles. More explicitly, we
first prepare three code blocks and construct one seed $2$-cocycle in each
block. We also construct a seed $6$-cycle in the tensor product sheaf
complex of the three code blocks. We then construct an auxiliary cochain
complex with linearly many $0$-cocycles and let these $0$-cocycles act on
the seed cocycles by cup product. The corresponding dual $0$-cocycles
produce $6$-cycles by cap product and Poincar\'e duality. The resulting
diagonal cup product pairing gives linearly many addressable $\CCZ$ gates
and simultaneously proves constant rate for all three code blocks. By
adding the corresponding $6$-cycles, we can address arbitrary subsets of
these logical $\CCZ$ gates. A fixed repetition layout further makes every
such subset parallelizable. The same mechanism can also be generalized to
multi-controlled-$Z$ gates.

\section{Preliminaries}\label{sec:preliminaries}

We now introduce sheaf codes and logical multi-controlled-$Z$ gates. Throughout, $\mathbb F_q$ is a finite field of characteristic two, and all (co)chain and sheaf coefficient spaces are finite dimensional over $\mathbb F_q$.

\subsection{Chain complexes and CSS codes}\label{subsec:prelim-complexes}

A \emph{cochain complex} $C^\bullet$ consists of a sequence of finite-dimensional vector spaces over $\mathbb F_q$ and linear
\emph{coboundary maps}
\begin{equation}
	\cdots\longrightarrow C^{j-1}
	\xrightarrow{\delta^{j-1}}C^j
	\xrightarrow{\delta^j}C^{j+1}\longrightarrow\cdots,
	\qquad \delta^j\delta^{j-1}=0.
\end{equation}
Its dual \emph{chain complex} has spaces $C_j=\operatorname{Hom}_{\mathbb F_q}(C^j,\mathbb F_q)$ and \emph{boundary maps} $\partial_{j+1}=(\delta^j)^*:C_{j+1}\to C_j$. We use the evaluation pairing with the cochain first:
\begin{equation}
	\langle\alpha,\xi\rangle=\xi(\alpha),\qquad
	\langle\delta^j\alpha,\xi\rangle
	=\langle\alpha,\partial_{j+1}\xi\rangle.
	\label{eq:prelim-evaluation}
\end{equation}
Fix bases of the cochain spaces and dual bases of the chain spaces. In these bases, the matrix of $\partial_{j+1}$ is $(\delta^j)^T$.

The spaces of \emph{cocycles}, \emph{coboundaries},
\emph{cycles}, and \emph{boundaries} are
\begin{align}
	Z^j&=\ker\delta^j,& B^j&=\im\delta^{j-1},&
	Z_j&=\ker\partial_j,& B_j&=\im\partial_{j+1}.
\end{align}
The corresponding cohomology and homology are
\begin{equation}
	H^j(C^\bullet)=Z^j/B^j,\qquad
	H_j(C_\bullet)=Z_j/B_j.
\end{equation}
The pairing in Eq.~\eqref{eq:prelim-evaluation} induces the evaluation pairing between these two quotient spaces. In particular, cycles are elements of $C_j$, and cocycles are elements of $C^j$.

Write $q=2^b$. The field trace is the $\mathbb F_2$-linear map
\begin{equation}
	\tr_{\mathbb F_q/\mathbb F_2}:\mathbb F_q\longrightarrow\mathbb F_2,
	\qquad
	\tr_{\mathbb F_q/\mathbb F_2}(a)
	=\sum_{s=0}^{b-1}a^{2^s}.
\end{equation}
Choose an $\mathbb F_2$-basis $\beta=(\beta_1,\ldots,\beta_b)$ of $\mathbb F_q$ and its \emph{trace-dual basis} $\beta^\vee$ so that $\tr_{\mathbb F_q/\mathbb F_2}(\beta_s\beta^\vee_u)$ is one for $s=u$ and zero otherwise. For a vector $x$ over $\mathbb F_q$, let $[x]_\beta$ denote its binary coordinate vector. For a matrix $M$ over $\mathbb F_q$, let $R_\beta(M)$ be the binary matrix characterized by $[Mx]_\beta=R_\beta(M)[x]_\beta$. Then
\begin{equation}
	[x]_\beta\cdot[z]_{\beta^\vee}=\tr_{\mathbb F_q/\mathbb F_2}(z^Tx),\qquad
	R_\beta(M)^T=R_{\beta^\vee}(M^T).
	\label{eq:prelim-trace-dual}
\end{equation}
We use $\beta$ for cochains and $\beta^\vee$ for chains. This makes binary duality agree with the chain--cochain pairing. Thus restriction of scalars gives the binary complexes used below.

\begin{definition}\label{def:prelim-css}
	The binary \emph{Calderbank–Shor–Steane (CSS) code} associated with an $\mathbb F_q$-cochain complex
	$C^{j-1}\xrightarrow{\delta^{j-1}}C^j
	\xrightarrow{\delta^j}C^{j+1}$ is defined by the parity check matrices over $\mathbb F_2$:
	\begin{equation}
		H_Z=R_\beta(\delta^j),\qquad
		H_X=R_\beta(\delta^{j-1})^T.
		\label{eq:prelim-css-matrices}
	\end{equation}
	Let $N=\dim_{\mathbb F_2}C^j$. Its code space $\mathcal Q\subseteq(\mathbb C^2)^{\otimes N}$ is the common $+1$ eigenspace of the Pauli operators $X(h)$ for the rows $h$ of $H_X$ and $Z(h)$ for the rows $h$ of $H_Z$. Here $X(h)$ applies Pauli $X$ at the positions where $h$ is one, and $Z(h)$ is defined in the same way. The relation $H_ZH_X^T=0$ ensures that these operators commute. The $X$-type and $Z$-type logical Pauli operators are represented by cohomology and homology classes, respectively.
\end{definition}

The number of encoded qubits is $k=\dim_{\mathbb F_2}H^j(C^\bullet)$. For a cocycle $x\in Z^j$, the encoded computational basis state is
\begin{equation}
	|[x]\rangle
	=\frac{1}{\sqrt{|B^j|}}\sum_{u\in B^j}|[x+u]_\beta\rangle,
	\qquad [x]\in H^j(C^\bullet).
	\label{eq:prelim-css-basis}
\end{equation}
These states range over all cohomology classes. The \emph{$X$-distance} and \emph{$Z$-distance} are
\begin{equation}
	d_X=\min_{x\in Z^j\setminus B^j}|[x]_\beta|,
	\qquad
	d_Z=\min_{z\in Z_j\setminus B_j}|[z]_{\beta^\vee}|.
	\label{eq:prelim-distances}
\end{equation}
Here $|\cdot|$ is binary Hamming weight. We write $[\![N,k,d]\!]$ for the code parameters, where $d=\min\{d_X,d_Z\}$. A family is \emph{quantum low-density parity check (qLDPC) } if the row and column weights of its specified check matrices are bounded independently of $N$.

\subsection{Cocycle expansion and quantum local testability}
\label{sec:prelim-qltc}

We now regard the cochain complex and its dual as binary vector spaces, using the bases fixed above. In this subsection, $|\cdot|$ and $\operatorname{dist}$ denote binary Hamming weight and distance. We use the notions of cocycle and cycle expansion from Ref.~\cite{Dinur2024sheaf}.

\begin{definition}[Cocycle and cycle expansion] \label{def:prelim-expansion}
	At the middle degree $j$ of the chain complex $C^\bullet$, define
	\begin{align}
		\epsilon_\delta(j)
		&=\min_{x\in C^j\setminus\ker\delta^j}
		\frac{|\delta^j x|}{\operatorname{dist}(x,\ker\delta^j)},
		\label{eq:prelim-cocycle-expansion}\\
		\epsilon_\partial(j)
		&=\min_{z\in C_j\setminus\ker\partial_j}
		\frac{|\partial_j z|}{\operatorname{dist}(z,\ker\partial_j)}.
		\label{eq:prelim-cycle-expansion}
	\end{align}
	These are the \emph{cocycle expansion} and \emph{cycle expansion}, respectively. 
\end{definition}

Cocycle expansion measures distance to all cocycles. Coboundary expansion instead measures distance to $B^j=\operatorname{im}\delta^{j-1}$:
\begin{equation}\label{eq:prelim-coboundary-expansion}
	\beta_\delta(j)
	=\min_{x\in C^j\setminus B^j}
	\frac{|\delta^j x|}{\operatorname{dist}(x,B^j)}.
\end{equation}
Thus $\beta_\delta(j)\leq\epsilon_\delta(j)$, and $\beta_\delta(j)=0$ whenever $H^j\ne0$. The soundness of a CSS code is determined by cocycle and cycle expansion, with distance to the full kernels in Eq.~\eqref{eq:prelim-cocycle-expansion} and Eq.~\eqref{eq:prelim-cycle-expansion}.

Let $M_X,M_Z>0$ be the numbers of rows in the specified lists $H_X,H_Z$, retaining dependent and zero rows. Since
$H_Z=\delta^j$ and $H_X=\partial_j$ in the chosen binary bases, the normalized \emph{soundness} constants are
\begin{align}
	\rho_X
	&=\min_{x\in C^j\setminus\ker\delta^j}
	\frac{|\delta^jx|}{M_Z}
	\frac{N}{\operatorname{dist}(x,\ker\delta^j)}
	=\frac{N}{M_Z}\epsilon_\delta(j),\label{eq:prelim-soundness-normalization}\\
	\rho_Z
	&=\min_{z\in C_j\setminus\ker\partial_j}
	\frac{|\partial_jz|}{M_X}
	\frac{N}{\operatorname{dist}(z,\ker\partial_j)}
	=\frac{N}{M_X}\epsilon_\partial(j).
\end{align}
Thus soundness is the rescaling of expansion by the number of physical qubits and the number of checks. We put $\rho=\min\{\rho_X,\rho_Z\}$. For every finite $\rho_0>0$ with $\rho_0\leq\rho$, the following inequalities hold for all $x,z\in\mathbb F_2^N$:
\begin{equation}\label{eq:prelim-css-soundness}
	\frac{|H_Zx|}{M_Z}\geq
	\rho_0\frac{\operatorname{dist}(x,\ker H_Z)}{N},
	\qquad
	\frac{|H_Xz|}{M_X}\geq
	\rho_0\frac{\operatorname{dist}(z,\ker H_X)}{N}.
\end{equation}
The subscripts refer to the errors being tested: $H_Z$ detects $X$ errors, and $H_X$ detects $Z$ errors.

\begin{definition}[Quantum local testability] \label{def:prelim-qltc}
	A CSS code is \emph{locally testable} with soundness at least a finite constant $\rho_0>0$ if both inequalities in Eq.~\eqref{eq:prelim-css-soundness} hold.
\end{definition}
This CSS definition is equivalent to the usual quantum definition, up to the normalization of the tester~\cite{AharonovEldar2015QLTC,Dinur2024sheaf}. For the families considered below, both numbers of listed checks are $\Theta(N)$. Therefore, a constant lower bound on $\rho$ also gives constant quantum soundness when the tester samples uniformly from the combined list of checks.

A family of quantum locally testable codes with growing blocklength $N$ is called \emph{asymptotically good} 
if it is a qLDPC family with parameters satisfying
\[
k=\Omega(N),\qquad d=\Omega(N),\qquad \rho=\Omega(1).
\]
For $r$ participating blocks, the combined code is $\bigotimes_{a=1}^r\mathcal Q_a$. If block $a$ has $N_a$ physical qubits and encodes $k_a$ qubits, then $N=\sum_aN_a$ and $k=\sum_ak_a$. In the construction below, one block supplies enough logical qubits to ensure a positive total rate, while distance and soundness bounds are established for every block.

\subsection{Cell complexes and sheaves}\label{subsec:prelim-cells}

We consider finite regular cell complexes. Write $X(j)$ for the set of $j$-cells of $X$. For cells $\sigma,\tau\in X$, the relation $\sigma\leq\tau$ means that $\sigma$ is a face of $\tau$, and $\sigma\lessdot\tau$ means that it is a face of codimension one. A family of cell complexes is \emph{sparse} if each cell has a bounded number of faces and cofaces, with bounds independent of the size of the complex.

The complexes in our construction are cubical: the closure of each cell is a cube, and its attaching maps preserve faces. Their dimension is denoted by $t$. Coordinate directions are labeled by $[t]=\{1,\ldots,t\}$, and
\begin{equation}
	\type(\sigma)\subseteq[t],\qquad
	|\type(\sigma)|=\dim\sigma
\end{equation}
is the set of edge directions of $\sigma$. In every remaining direction a cell has an endpoint labeled zero or one. The direction labels and endpoint labels are compatible under passage to faces. They specify the factors in the local sheaf construction below and the faces used to compute cup products. For a cell $\sigma$, write $X_{\leq\sigma}$ for its faces and $X_{\geq\sigma}$ for the cells containing it, called its upward star.

\begin{definition}
	A \emph{sheaf} $\mathcal F$ on the cell poset of $X$ assigns a vector space
	$\mathcal F_\sigma$, called the stalk, to each cell $\sigma$, and a linear restriction map
	\begin{equation}
		\mathcal F_{\sigma,\tau}:\mathcal F_\sigma\longrightarrow
		\mathcal F_\tau
		\qquad(\sigma\leq\tau),
	\end{equation}
	such that $\mathcal F_{\sigma,\sigma}=\id$ and
	\begin{equation}
		\mathcal F_{\tau,\pi}\circ\mathcal F_{\sigma,\tau}
		=\mathcal F_{\sigma,\pi}
		\qquad(\sigma\leq\tau\leq\pi).
	\end{equation}
\end{definition}

A morphism $\mu:\mathcal F\to\mathcal G$ consists of linear maps $\mu_\sigma:\mathcal F_\sigma\to\mathcal G_\sigma$ satisfying
\begin{equation}
	\mathcal G_{\sigma,\tau}\mu_\sigma
	=\mu_\tau\mathcal F_{\sigma,\tau}.
	\label{eq:prelim-morphism}
\end{equation}
The tensor product sheaf has stalks and restrictions
\begin{equation}
	(\mathcal F\otimes\mathcal G)_\sigma
	=\mathcal F_\sigma\otimes\mathcal G_\sigma,
	\qquad
	(\mathcal F\otimes\mathcal G)_{\sigma,\tau}
	=\mathcal F_{\sigma,\tau}\otimes\mathcal G_{\sigma,\tau}.
\end{equation}
If $f:Y\to X$ is an order-preserving map of cell posets, the pullback sheaf $f^*\mathcal F$ is given by
\begin{equation}
	(f^*\mathcal F)_\sigma=\mathcal F_{f(\sigma)},\qquad
	(f^*\mathcal F)_{\sigma,\tau}
	=\mathcal F_{f(\sigma),f(\tau)}.
	\label{eq:prelim-pullback-sheaf}
\end{equation}

For sheaves $\mathcal F$ on $X$ and $\mathcal G$ on $Y$, their external tensor product on $X\times Y$ has stalk $\mathcal F_\sigma\otimes\mathcal G_\tau$ at $(\sigma,\tau)$ and restriction maps given by tensor products of the respective restriction maps~\cite{LLL2026nontrivial}.

For a sheaf $\mathcal F$ on $X$, define
\begin{equation}
	C^j(X,\mathcal F)=\bigoplus_{\sigma\in X(j)}\mathcal F_\sigma,
	\qquad
	C_j(X,\mathcal F)
	=\operatorname{Hom}_{\mathbb F_q}(C^j(X,\mathcal F),\mathbb F_q).
	\label{eq:prelim-sheaf-cochains}
\end{equation}
Thus a cochain assigns a vector to each $j$-cell, and a chain assigns a covector in $\mathcal F_\sigma^*$ to each $j$-cell. In characteristic two, the coboundary and boundary operators are
\begin{align}
	(\delta^j\alpha)(\tau)
	&=\sum_{\sigma\lessdot\tau}
	\mathcal F_{\sigma,\tau}\alpha(\sigma),
	\quad \tau\in X(j+1),\label{eq:prelim-sheaf-delta}\\
	(\partial_{j+1}\xi)(\sigma)
	&=\sum_{\tau\gtrdot\sigma}
	\mathcal F_{\sigma,\tau}^*\xi(\tau),
	\quad \sigma\in X(j).\label{eq:prelim-sheaf-boundary}
\end{align}
These give a cochain complex and its dual chain complex. For the associated code, choose a basis in each stalk, concatenate these bases in the direct sum, and use the dual basis on chains. The pairing is
\begin{equation}
	\langle\alpha,\xi\rangle
	=\sum_{\sigma\in X(j)}\xi(\sigma)(\alpha(\sigma)).
\end{equation}
We write $H^j(X,\mathcal F)$ and $H_j(X,\mathcal F)$ for the cohomology and homology. A sheaf morphism induces the cochain map
$\mu_*:C^j(X,\mathcal F)\to C^j(X,\mathcal G)$ given by
$(\mu_*\alpha)(\sigma)=\mu_\sigma\alpha(\sigma)$.

The {sheaf code in degree $j$} is the CSS code obtained from the three terms $C^{j-1}(X,\mathcal F)$, $C^j(X,\mathcal F)$, and $C^{j+1}(X,\mathcal F)$ by Definition~\ref{def:prelim-css}. For sparse complexes, fixed $q$, and bounded stalk dimensions, these codes are qLDPC. Local testability and distance require additional conditions on the complexes and local codes. Those conditions are stated with the construction theorems later in the paper. The quantum codes use degree $j=2$ in the present paper. The abelian lifts used for almost-good qLTCs in Ref.~\cite{Dinur2024sheaf} provide an earlier application of this sheaf code framework. In this work, we use arithmetic cubical complexes in the good qLTC construction \cite{RSV2019,GJ2026,BLN2026,CHLT2026}.

We use a particular kind of sheaf in this paper. For a $t$-dimensional cell complex $X$, we equip each $(t-1)$-dimensional cell $\rho\in X(t-1)$ with a classical linear code $C_\rho\subseteq \mathbb F_q^{X_{\ge\rho}(t)}$. Then for each cell $\sigma\in X$, we define
\begin{align}
    \mathcal F_\sigma=\{z\in\mathbb F_q^{X(t)}: z|_{X_{\ge\rho}(t)}\in C_\rho\, \text{for each}\, \rho\in X(t-1)\},
\end{align}
and for each $\sigma\le \sigma'$, the map $\mathcal F_{\sigma,\sigma'}:\mathcal F_\sigma\to\mathcal F_{\sigma'}$ is defined by the restriction of domain, i.e. for $z\in\mathcal F_\sigma$, we define $\mathcal F_{\sigma,\sigma'}z\coloneq z|_{X_{\ge\sigma'}(t)}$. We further define the dual sheaf $\mathcal F^\perp $ of $\mathcal F$ by replacing the classical linear codes $C_\rho$ to the dual code $C_\rho^\perp\subseteq \mathbb F_q^{X_\ge\rho(t)}$. Then the sheaf $\mathcal F$ and dual sheaf $\mathcal F^\perp$ satisfies the Poincar\'e duality which we will use later in the construction. One may find more details in Ref~\cite{Li2025Poincare}.

\comments{Let $\mathcal C_i\subseteq\mathbb F_q^{A_i}$ be a local linear code, where $|A_i|=n_i$ and $\dim\mathcal C_i=m_i$, for $i\in[t]$. Thus $n_i$ and $m_i$ denote local length and dimension; $N$ continues to denote the number of physical qubits. Choose an injective linear map
\begin{equation}
	\ev_i:\mathbb F_q^{m_i}\longrightarrow\mathbb F_q^{A_i},
	\qquad \im\ev_i=\mathcal C_i,
\end{equation}
and write $\ev_{i,a}(u)=(\ev_i(u))_a$ for its coordinate at $a\in A_i$. In matrix notation, $h_i=\ev_i^T:\mathbb F_q^{A_i}\to\mathbb F_q^{m_i}$ is a parity check matrix for $\mathcal C_i^\perp$.

Here is the basic sheaf construction on a graph. At each vertex of an $n_i$-regular graph with its incident edges labeled by $A_i$, put the space $\mathbb F_q^{m_i}$, and at each edge put $\mathbb F_q$. The restriction from a vertex to its incident edge labeled $a$ is $\ev_{i,a}$. On a product of $t$ such graphs, take the external tensor product of these sheaves. For a cell $\sigma$ of type $S$, its stalk is
\begin{equation}
	\mathcal F_\sigma=\bigotimes_{i\notin S}\mathbb F_q^{m_i}.
	\label{eq:prelim-product-stalk}
\end{equation}
Passing to a coface in direction $i$ applies the appropriate $\ev_{i,a}$ to that factor. This is the local construction underlying the sheaf codes. On the cubical complexes used here, the restrictions are specified in local charts. The restriction maps agree under these changes of local coordinates. Additional coefficient spaces, when present, are included in the stalks and in these restriction maps.}

\subsection{Cohomological invariants and logical gates}
\label{subsec:prelim-logical-gates}

Fix an integer $r\geq2$. For block $a\in[r]$, let $C^{(a)}$ be the middle space of a cochain complex, with cocycle and coboundary spaces $Z^{(a)}$ and $B^{(a)}$. Let $\mathcal Q_a$ be the CSS code of Definition~\ref{def:prelim-css}, and put $H^{(a)}=Z^{(a)}/B^{(a)}$.

\begin{definition}[Cohomological invariant form]
	An $\mathbb F_q$-multilinear map
	\begin{equation}
		T:C^{(1)}\times\cdots\times C^{(r)}\longrightarrow\mathbb F_q
	\end{equation}
	is a \emph{cohomological invariant form} if
	\begin{equation}
		T(x_1+u_1,\ldots,x_r+u_r)=T(x_1,\ldots,x_r)
		\quad
		(x_a\in Z^{(a)},\ u_a\in B^{(a)}).
		\label{eq:prelim-invariant}
	\end{equation}
	It therefore induces an $r$-linear form on
	$H^{(1)}\times\cdots\times H^{(r)}$.
\end{definition}

In the next section, we construct cohomological invariant forms using cup products and the evaluation pairing in Eq.~\eqref{eq:prelim-evaluation}.

The $r$-qubit gate $\mCZ$ is defined by
\begin{equation}
	\mCZ\,|b_1,\ldots,b_r\rangle
	=(-1)^{b_1\cdots b_r}|b_1,\ldots,b_r\rangle,
	\qquad b_a\in\mathbb F_2.
	\label{eq:prelim-controlled-z}
\end{equation}
Thus $r=2$ gives $\CZ$ and $r=3$ gives $\CCZ$. Using the field trace $\tr_{\mathbb F_q/\mathbb F_2}$ defined above, a cohomological invariant form $T$ defines the physical diagonal unitary
\begin{equation}
	U_T|[x_1]_\beta,\ldots,[x_r]_\beta\rangle
	=(-1)^{\tr_{\mathbb F_q/\mathbb F_2}\bigl(T(x_1,\ldots,x_r)\bigr)}
	|[x_1]_\beta,\ldots,[x_r]_\beta\rangle.
	\label{eq:prelim-physical-gate}
\end{equation}
By Eq.~\eqref{eq:prelim-invariant}, its phase is constant on each product of coboundary cosets in Eq.~\eqref{eq:prelim-css-basis}. It therefore preserves $\mathcal Q_1\otimes\cdots\otimes\mathcal Q_r$ and acts on all encoded computational basis states by
\begin{equation}
	U_T|[x_1],\ldots,[x_r]\rangle
	=(-1)^{\tr_{\mathbb F_q/\mathbb F_2}\bigl(T(x_1,\ldots,x_r)\bigr)}|[x_1],\ldots,[x_r]\rangle
	\qquad(x_a\in Z^{(a)}).
	\label{eq:prelim-logical-action}
\end{equation}

If the induced logical form is nonzero, choose cocycles for which $v=T(x_1,\ldots,x_r)\ne0$. The field trace is surjective, so choose $\eta\in\mathbb F_q$ with $\tr_{\mathbb F_q/\mathbb F_2}(\eta)=1$ and replace $T$ by $(\eta/v)T$. The resulting binary logical phase takes the value one on the chosen tuple. This is the normalization used to establish a nontrivial logical action.

Write $x_{a,s}=([x_a]_\beta)_s$ for the $s$-th binary coordinate of $x_a$. In these coordinates, $\tr_{\mathbb F_q/\mathbb F_2}\circ T$ has the form
\begin{equation}
	\tr_{\mathbb F_q/\mathbb F_2}\bigl(T(x_1,\ldots,x_r)\bigr)
	=\sum_{s_1=1}^{N_1}\cdots\sum_{s_r=1}^{N_r}
	t_{s_1\cdots s_r}
	x_{1,s_1}\cdots x_{r,s_r},
	\qquad t_{s_1\cdots s_r}\in\mathbb F_2,
	\label{eq:prelim-gate-tensor}
\end{equation}
where $N_a$ is the number of physical qubits in block $a$ and $N=\sum_{a=1}^rN_a$ is the total system size. Each nonzero term gives one physical $\mCZ$ gate, with one qubit from each block. If every coordinate occurs in a bounded number of terms, these gates form a \emph{constant-depth circuit}. A \emph{transversal} implementation has pairwise disjoint physical gate supports. When each qubit occurs in only boundedly many gates, a repetition encoding of fixed length gives such an implementation and changes the code parameters by at most constant factors~\cite{Nguyen2025CCZ, Golowich_Lin2024}.

\subsection{Cup and cap products with sheaf coefficients}\label{subsec:qltc-cup}

Cup products with sheaf coefficients are defined on simplicial and regular cell complexes~\cite{Li2025Poincare,LLL2026nontrivial}. This construction is independent of the distance and local testability estimates for the associated codes. We recall the construction. Throughout this subsection, $\mathcal F$ and $\mathcal G$ are sheaves over $\mathbb F_q$ on the finite regular cell complex $X$.

\begin{definition}[Cup product] \label{def:qltc-cup}
	For integers $u,v\geq0$, on an ordered simplicial complex the \emph{cup product} is the bilinear map
	\begin{equation}
		\smile:C^u(X,\mathcal F)\times C^v(X,\mathcal G)
		\longrightarrow C^{u+v}(X,\mathcal F\otimes\mathcal G)
		\label{eq:qltc-cup-type}
	\end{equation}
	defined as follows. For a simplex $\sigma=[v_0,\ldots,v_{u+v}]$, put ${}_u\sigma=[v_0,\ldots,v_u]$ and $\sigma_v=[v_u,\ldots,v_{u+v}]$. Then
	\begin{equation}
		(\alpha\smile\gamma)(\sigma)
		=\mathcal F_{{}_u\sigma,\sigma}(\alpha({}_u\sigma))
		\otimes
		\mathcal G_{\sigma_v,\sigma}(\gamma(\sigma_v)).
		\label{eq:qltc-simplicial-cup}
	\end{equation}
	For a regular cell complex, let $\sd X$ be its barycentric subdivision: a simplex is a chain $\sigma_0<\cdots<\sigma_j$ of cells of $X$, ordered by inclusion. Let $s_X:\sd X\to X$ send this chain to $\sigma_j$, and use the pullback sheaves $s_X^*\mathcal F$ and $s_X^*\mathcal G$. The subdivision construction gives cochain maps
	\begin{equation}
		C^j(X,\mathcal F)
		\xrightarrow{A^\#}C^j(\sd X,s_X^*\mathcal F)
		\xrightarrow{S^\#}C^j(X,\mathcal F),
		\qquad S^\#A^\#=\id.
		\label{eq:qltc-subdivision-maps}
	\end{equation}
	To define these maps, let $S_\#:C_j(X,\mathbb F_q)\to C_j(\sd X,\mathbb F_q)$ be scalar subdivision, which sends a cell to the sum of the simplices subdividing it. Choose a chain map $A_\#:C_j(\sd X,\mathbb F_q)\to C_j(X,\mathbb F_q)$ satisfying $A_\#S_\#=\id$ and sending each simplex $\rho=[\sigma_0<\cdots<\sigma_j]$ to a chain supported on faces of $\sigma_j$. Such a map is called an approximate inverse and exists~\cite{LLL2026nontrivial}. Write $A_\#\rho=\sum_{\sigma\leq\sigma_j,\,\dim\sigma=j}a_{\rho,\sigma}\sigma$.
	The induced map on sheaf cochains is
	\begin{equation}
	(A^\#\alpha)(\rho) = \sum_{\substack{\sigma\leq\sigma_j\\\dim\sigma=j}}
	a_{\rho,\sigma}\mathcal F_{\sigma,\sigma_j}\alpha(\sigma).
	\end{equation}
	The map $S^\#$ sums over the subdividing simplices ending at each cell. Use the same scalar maps for every coefficient sheaf. The cellular cup product is
	\begin{equation}
		\alpha\smile\gamma
		=S^\#\bigl(A^\#\alpha\smile A^\#\gamma\bigr),
		\label{eq:qltc-cellular-cup}
	\end{equation}
	where the product on the right is defined by Eq.~\eqref{eq:qltc-simplicial-cup} on $\sd X$.
\end{definition}

The following properties of cup products are useful.

\begin{proposition} \label{prop:qltc-cup-cohomology}
	The cup product satisfies the Leibniz rule in characteristic two:
	\begin{equation}
		\delta(\alpha\smile\gamma)
		= \delta\alpha\smile\gamma+\alpha\smile\delta\gamma.
		\label{eq:qltc-cup-leibniz}
	\end{equation}
	It induces a product on cohomology independent of the approximate inverse in Eq.~\eqref{eq:qltc-subdivision-maps}. For a sparse complex with bounded stalk dimensions, each input coordinate occurs in a bounded number of terms of the cup product.
\end{proposition}

There is a dual construction of cup product defined as follows. 
\begin{definition}[Cap product]
For integers $u,v\geq0$, on an ordered simplicial complex the \emph{cap product} is the bilinear map
\begin{equation}
\frown:C^u(X,\mathcal{F})\times C_{u+v}(X,\mathcal{F}\otimes\mathcal{G})\longrightarrow C_v(X,\mathcal G)
\end{equation}
as follows: for $\alpha \in C^{p}(X,\mathcal{F}), x\in C_{p+q}(X,\mathcal{F}\otimes\mathcal{G})$,
\begin{align}
	\alpha\frown x\coloneqq \sum_{\sigma\in X(u+v)}\mathcal{G}^T_{\sigma_v,\sigma}\langle\mathcal{F}_{{}_u\sigma,\sigma}\alpha({}_u\sigma),x(\sigma)\rangle_{\mathcal{F}}\cdot \sigma_v.
\label{eq:cap_other_definition}
\end{align}
\end{definition}

Similar to the cup product, the cap product on a general cell complex is defined by
\begin{equation}
    \alpha\frown x=A_\#(A^\#\alpha\frown S_\# x).
\end{equation}

\begin{proposition}
    The cup product satisfies the Leibniz rule in characteristic two:
	\begin{equation}
		\partial(\alpha\frown x)=(\delta\alpha)\frown x+\alpha\frown (\partial x)
		\label{eq:qltc-cap-leibniz}
	\end{equation}
	It induces a product on cohomology independent of the approximate inverse in Eq.~\eqref{eq:qltc-subdivision-maps}. For a sparse complex with bounded stalk dimensions, each input coordinate occurs in a bounded number of terms of the cup product.
\end{proposition}

An iterated cochain product is left-associated. This fixes the cochain representative; at the cohomology level we use the usual associativity identifications of tensor product sheaves. For the explicit calculations later, we will specify the cubical formula for this product and its relation to Eq.~\eqref{eq:qltc-cellular-cup}.

The tensor product in Eq.~\eqref{eq:qltc-cup-type} must be distinguished from multiplication of polynomials in the coefficient spaces. For sheaves $\mathcal F_1,\ldots,\mathcal F_r$ and $\mathcal T$ on $X$,
suppose we are given a sheaf morphism
\begin{equation}
	\mu:\bigotimes_{a=1}^r\mathcal F_a\longrightarrow\mathcal T.
	\label{eq:qltc-coefficient-multiplication}
\end{equation}
Its stalk maps must satisfy the compatibility condition in Eq.~\eqref{eq:prelim-morphism}. It induces the cochain map
\begin{equation}
	\mu_*:C^j\left(X,\bigotimes_{a=1}^r\mathcal F_a\right)
	\longrightarrow C^j(X,\mathcal T),
	\qquad (\mu_*\alpha)(\sigma)=\mu_\sigma(\alpha(\sigma)).
	\label{eq:qltc-multiplication-cochains}
\end{equation}
For the polynomial sheaves used below, $\mu$ is polynomial multiplication. The degrees and restriction maps determine its target sheaf $\mathcal T$.

Composing cup product with a sheaf morphism and evaluation on a cycle gives the invariant forms of Refs.~\cite{LSWLL2026Theory,LLL2026nontrivial}.
\begin{proposition}[Cohomological invariant form]
	\label{prop:qltc-cup-invariant}
	Let $j_a\geq0$ for $a\in[r]$, with $j_1+\cdots+j_r=t$, and
	use the morphism in Eq.~\eqref{eq:qltc-coefficient-multiplication}. Let
	$\xi\in C_t(X,\mathcal T)$ satisfy $\partial\xi=0$.
	With the pairing of Eq.~\eqref{eq:prelim-evaluation}, define
	\begin{equation}
		I_\xi(x_1,\ldots,x_r)
		=\left\langle
		\mu_*(x_1\smile\cdots\smile x_r),\xi
		\right\rangle,
		\qquad x_a\in C^{j_a}(X,\mathcal F_a).
		\label{eq:qltc-cup-invariant}
	\end{equation}
	Then $I_\xi$ satisfies Eq.~\eqref{eq:prelim-invariant} and induces an
	$r$-linear form on the full product
	$\prod_{a=1}^rH^{j_a}(X,\mathcal F_a)$.
\end{proposition}

In our application, $j_a=2$ for all $a$ and $t=2r$. Once a cycle $\xi$ and cocycles giving a nonzero value of the form in Eq.~\eqref{eq:qltc-cup-invariant} have been constructed, the normalization and logical action are those of Eqs.~\eqref{eq:prelim-physical-gate}--\eqref{eq:prelim-logical-action}. The proposition applies to every cocycle in each block; the later formulas select particular cocycles only to prove that the form is nonzero.

We will also use the adjoint of cup product with a fixed cocycle. Let $\gamma\in C^v(X,\mathcal G)$ satisfy $\delta\gamma=0$, and let $\mu:\mathcal E\otimes\mathcal G\to\mathcal T$ be a sheaf morphism. For each $u\geq0$, define
\begin{align} \label{eq:qltc-cup-adjoint}
    L_\gamma^u &:C^u(X,\mathcal E)\longrightarrow
	C^{u+v}(X,\mathcal T),
	\qquad L_\gamma^u(\alpha) = \mu_*(\alpha\smile\gamma),\\
	(L_\gamma^u)^* &:C_{u+v}(X,\mathcal T)\longrightarrow C_u(X,\mathcal E),
	\qquad \langle\alpha,(L_\gamma^u)^*\xi\rangle
	=\langle L_\gamma^u(\alpha),\xi\rangle.
\end{align}
By the Leibniz rule, $\delta_{\mathcal T}^{u+v}L_\gamma^u=L_\gamma^{u+1}\delta_{\mathcal E}^u$. Taking adjoints in degree $u-1$ gives $\partial_u(L_\gamma^u)^*=(L_\gamma^{u-1})^*\partial_{u+v}$ for $u\geq1$. Thus $(L_\gamma^u)^*$ sends cycles to cycles; for $u=0$ this holds because every zero-chain is a cycle. This is the adjoint construction underlying the cap products in Ref.~\cite{LLL2026nontrivial}. The formula fixes the argument order: $\gamma$ is the last cup product input. The later reduction in the number of directions uses this adjoint together with a restriction of the complex and of its coefficient sheaves.

\subsection{Covering spaces and nonzero pairings}\label{subsec:qltc-covering}

Let $P:\widetilde X\to X$ be an $\ell$-sheeted cellular covering: it maps each cell isomorphically onto a cell and is locally an isomorphism of cell complexes, and every cell of $X$ has $\ell$ lifts. In particular, at each lifted cell $\widetilde\sigma$, the maps
\begin{align}
    P:\widetilde X_{\leq\widetilde\sigma}\longrightarrow X_{\leq P\widetilde\sigma},
\qquad
P:\widetilde X_{\geq\widetilde\sigma}\longrightarrow X_{\geq P\widetilde\sigma}
\end{align}
are isomorphisms of posets. We use the pullback sheaf $P^*\mathcal F$ from the preliminaries. Thus the stalk and the restriction maps at lifted cells are those at their images.The cochain pullback, chain pushforward, and chain transfer are
	\begin{align}
	    P^\# &:C^j(X,\mathcal F)\longrightarrow
	C^j(\widetilde X,P^*\mathcal F), \\
	P_\# &:C_j(\widetilde X,P^*\mathcal F)\longrightarrow
	C_j(X,\mathcal F), \\
	T_\# &:C_j(X,\mathcal F)\longrightarrow
	C_j(\widetilde X,P^*\mathcal F). 
	\end{align}
Both $P^\#$ and $T_\#$ copy the value at a cell to each of its lifts, whereas $P_\#$ sums the values over those lifts (see also \cite{LLL2026nontrivial} for more details). We use the following identities and compatibility relations. 

\begin{proposition}[Cup products under coverings] \label{prop:qltc-covering}
	The above maps commute with the corresponding boundary or coboundary operators, and
	\begin{equation}
		P_\#T_\#=\ell\,\id,
		\qquad
		\langle P^\#\alpha,T_\#\xi\rangle
		=\ell\langle\alpha,\xi\rangle.
		\label{eq:qltc-transfer-pairing}
	\end{equation}
	The approximate inverses in Eq.~\eqref{eq:qltc-cellular-cup} can be
	chosen compatibly with $P$, so that
	\begin{equation}
		P^\#(\alpha\smile\gamma)
		=P^\#\alpha\smile P^\#\gamma.
		\label{eq:qltc-cover-cup}
	\end{equation}
	These are the transfer identities and compatibility results of
	Ref.~\cite{LLL2026nontrivial}. Consequently, with the pulled-back morphism $P^*\mu$,
	\begin{equation}
		I_{T_\#\xi}(P^\#x_1,\ldots,P^\#x_r)
		=\ell I_\xi(x_1,\ldots,x_r).
		\label{eq:qltc-cover-invariant}
	\end{equation}
	In characteristic two an odd $\ell$ equals one in $\mathbb F_q$.
	Hence an odd-degree covering preserves a nonzero value of the form.
\end{proposition}

The covering space argument applies to general cell complexes~\cite{LLL2026nontrivial}. For the explicit cubical representatives used here, the covers preserve direction labels and endpoint types, so we use the same local formula for the cup product on each cube and its lifts. Proposition~\ref{prop:qltc-covering} does not require an abelian covering group. For the good qLTCs, we construct arithmetic coverings with fixed coefficient sheaves and odd covering degrees, within the tree-quotient setting. 

\section{Construction of linearly addressable multi-controlled-$Z$ gates}
\label{sec:simplified-proof}
\label{sec:ccz-construction}
In this section, we construct linearly addressable and parallelizable transversal multi-controlled-$Z$ gates on good quantum LDPC codes and quantum locally testable codes. For simplicity, we only prove the $\CCZ$ case. General multi-controlled-$Z$ argument can be easily generalized. The construction uses a mechanism available in the three current good qLTC frameworks~\cite{GJ2026,BLN2026,CHLT2026}.  We work with code construction of Bafna--Li--Nguyen (BLN), because their construction makes the details of calculations most explicit.
\subsection{Construction of good quantum locally testable codes}
\label{subsec:s4-coordinates}

\subsubsection{Fields and coefficients}

All fields in the construction have characteristic two.  Fix six finite
fields $\mathbb F_{s_i}$, one for each direction, and a finite coefficient
field $\mathbb E$ containing their quadratic extensions and all fixed
algebraic data introduced below.  These fields are independent of the
growing code size.  We write $X=X_R$ for the six-dimensional arithmetic
cubical quotient at level $R$, following the notation of Ref~\cite{RSV2019}.

\subsubsection{Ramanujan cubical complexes and projective coordinates}

\begin{proposition}[Arithmetic quotient]
\label{claim:s4-arithmetic-geometry}
For every sufficiently large even integer $r$, put $R=2^r$ and
$\mathbb F_R=\mathbb F_{2^r}$.  The arithmetic construction used below
associates to an admissible degree-$r$ place a six-dimensional Ramanujan
cubical quotient $X_R$.  Write
$G_R=\operatorname{SL}_2(\mathbb F_R)$ and
$U_R=\operatorname{Stab}_{G_R}(e_1)$.  The varying-place covering has
fibre
\begin{equation}
 U_R\backslash G_R
 \longrightarrow
 \mathbb F_R^2\setminus\{0\},
 \qquad
 U_Rg\longmapsto g^{-1}e_1,
 \label{eq:s4-level-cosets}
\end{equation}
and hence degree $R^2-1$.  For every direction set
$S\subseteq[6]$, its cells satisfy
\begin{equation}
 |X_R(S)|
 =3(R^2-1)2^{6-|S|}
 \prod_{i\in S}(s_i+1).
 \label{eq:s4-arithmetic-cell-count}
\end{equation}
We also use the standard Ramanujan bound for every parallel-face walk:
if $I\subsetneq[6]$ and $j\notin I$, then the corresponding normalized
adjacency operator $M_{I,j}$ has
\begin{equation}
 \operatorname{Spec}(M_{I,j})
 \subseteq
 \{-1,1\}\cup
 \left[-\frac{2\sqrt{s_j}}{s_j+1},
        \frac{2\sqrt{s_j}}{s_j+1}\right],
 \label{eq:s4-arithmetic-spectrum}
\end{equation}
with one-dimensional extremal eigenspaces.  These facts, as well as the
existence of simultaneous equivariant lattice frames, are properties of
the cited arithmetic construction~\cite{BLN2026,GJ2026}. We do not
reprove its expansion theory here.
\end{proposition}

We next recall only the projective coordinate data needed to identify the
standard projective Reed--Solomon codes with the stalks of the coefficient
sheaves.
In direction $i$, let $\mathcal O_i$ be a discrete valuation
ring with uniformizer $\pi_i$ and residue field
$\mathbb F_{s_i}$.  A vertex of its Bruhat--Tits tree is a homothety class of rank-two
$\mathcal O_i$-lattices: two lattices represent the same vertex when
they differ by an overall nonzero scalar.  If
$L=U\mathcal O_i^2$, then the ordered lattice basis $U$ identifies
$L/\pi_iL\simeq\mathbb F_{s_i}^{\,2}$.  The neighbors of $[L]$ correspond to the one-dimensional subspaces of
this residue plane.  Thus an ordered basis identifies the direction-$i$
edges at a vertex with $\mathbb P^1(\mathbb F_{s_i})$.  We represent a
finite label $a$ by $r_a=(a,1)^T$ and the point at infinity by
$r_\infty=(1,0)^T$.
The projective label remembers the line $\langle r_a\rangle$, but not a
preferred nonzero vector on that line.

Here is the basic change-of-frame operation from which all matrices below
are obtained.  Suppose that the columns of $U$ and $V$ are ordered
bases of homothetic lattices.  Choose a scalar $\xi$ such that
$\xi V\mathcal O_i^2=U\mathcal O_i^2$.  Then
$U^{-1}\xi V\in\operatorname{GL}_2(\mathcal O_i)$.  Let
$\operatorname{red}_i$ denote reduction modulo $\pi_i$, and define
\begin{equation}
 [U:V]_i
 =
 \det\!\left(\operatorname{red}_i(U^{-1}\xi V)\right)^{-1/2}
 \operatorname{red}_i(U^{-1}\xi V)
 \in\operatorname{SL}_2(\mathbb F_{s_i}).
 \label{eq:s4-frame-comparison}
\end{equation}
The square root is unique because $\mathbb F_{s_i}$ has characteristic
two.  A different admissible $\xi$ multiplies the reduced matrix by a
scalar, which is removed by the determinant normalization.  Hence
Eq.~\eqref{eq:s4-frame-comparison} depends only on the two ordered
frames.  It is the coordinate change from the $V$-frame to the
$U$-frame, with its irrelevant projective scalar fixed by determinant
one.  Directly from the definition,
\begin{equation}
 [U:V]_i[V:W]_i=[U:W]_i,
 \qquad
 [V:U]_i=[U:V]_i^{-1},
 \qquad
 [AU:AV]_i=[U:V]_i.
 \label{eq:s4-frame-composition}
\end{equation}
The last identity explains why the construction descends from the product
of trees to the finite RSV quotient: changing a common lift left
multiplies both frames by the same matrix.

For every vertex $v$ of the quotient, choose in direction $i$ a
lifted lattice basis $U_{v,i}$, equivariantly under the covering group.
For every six-cube $Q$, choose an ordered basis
$V_{Q,i}=(p,q)$ adapted to its projected direction-$i$ tree edge:
if its endpoints are represented by $L_0$ and $L_1$, then
$L_0=\mathcal O_i p\oplus\mathcal O_i q$ and
$L_1=\mathcal O_i p\oplus\pi_i\mathcal O_i q$.  Put $T_i=\operatorname{diag}(1,\pi_i)$.  The two endpoint frames are
therefore $V_{Q,i}$ and $V_{Q,i}T_i$.

Let $b_i(v)\in\{0,1\}$ be the endpoint type of $v$ in direction
$i$.  The matrix comparing the cube frame with the chosen vertex frame
is
\begin{equation}
 F_i(v,Q)
 =
 [U_{v,i}:V_{Q,i}T_i^{\,b_i(v)}]_i
 \in\operatorname{SL}_2(\mathbb F_{s_i}).
 \label{eq:s4-cube-frame}
\end{equation}
If the direction-$i$ edge of $Q$ at $v$ has projective label
$a_i$, the adapted-basis construction says that column
$b_i(v)+1$ of $F_i(v,Q)$ spans $\langle r_{a_i}\rangle$.
There is consequently a unique scalar
$c_i(v,Q)\in\mathbb F_{s_i}^{\times}$ such that
\begin{equation}
 \operatorname{column}_{b_i(v)+1}F_i(v,Q)
 =c_i(v,Q)r_{a_i}.
 \label{eq:s4-ci-definition}
\end{equation}
Thus $F_i(v,Q)$ is the full change of coordinates, the label $a_i$
remembers only the resulting projective line, and $c_i(v,Q)$ remembers
the scale discarded when that line is represented by the fixed vector
$r_{a_i}$.  The scalar $c_i(v,Q)$ is therefore forced by the two
chosen frames, not an additional code parameter.

Now let $e:v\to w$ be a positive edge in direction $j$.  For a
transverse direction $i\ne j$, the two endpoints have the same
 direction-$i$ tree vertex, so $U_{v,i}$ and $U_{w,i}$ are two
frames of the same lattice.  Define
\begin{equation}
 B_e^{(i)}=[U_{w,i}:U_{v,i}]_i
 \in\operatorname{SL}_2(\mathbb F_{s_i}).
 \label{eq:s4-Be-definition}
\end{equation}
The superscript records the residue plane being transported.  Thus, when
$e$ lies in direction $j$, the matrices used across $e$ are
$B_e^{(i)}$ for $i\ne j$.  The walking direction $j$ is governed
by a different triangular identity below.

Since the same cube frame occurs at both endpoints in a transverse
direction, Eq.~\eqref{eq:s4-frame-composition} gives
\begin{equation}
 F_i(w,Q)=B_e^{(i)}F_i(v,Q)
 \qquad(i\ne j).
 \label{eq:s4-transverse-frame-transport}
\end{equation}
If $a_i'$ is the label of the opposite direction-$i$ edge seen from
$w$, comparison of the selected columns yields
\begin{equation}
 c_i(w,Q)r_{a_i'}
 =
 B_e^{(i)}c_i(v,Q)r_{a_i}.
 \label{eq:s4-transformed-label}
\end{equation}
Equivalently,
$\langle r_{a_i'}\rangle=
 B_e^{(i)}\langle r_{a_i}\rangle$.  In particular, $B_{e^{-1}}^{(i)}=(B_e^{(i)})^{-1}$, and around every elementary square whose two directions are different
from $i$, the products of the $B_e^{(i)}$ along its two length-two
paths agree.  This local path independence is the first identity in
Eq.~\eqref{eq:s4-frame-composition}; it is the coordinate form of the
fact that the two paths compare the same initial and final frames.

It remains to compare frames in the walking direction $j$, where the
two endpoint lattices are adjacent rather than homothetic.  Choose once
for the edge $e$ an adapted basis $V_{e,j}$, and put
$P_{v,e}=[U_{v,j}:V_{e,j}]_j$ and
$P_{w,e}=[U_{w,j}:V_{e,j}T_j]_j$.  For a cube $Q\geq e$, the change from $V_{e,j}$ to
$V_{Q,j}$ preserves the nested pair of lattices.  Before reduction it
has the form
\begin{equation}
 V_{e,j}^{-1}V_{Q,j}
 =
 \begin{pmatrix}a&b\\ \pi_jc&d\end{pmatrix},
 \qquad
 T_j^{-1}V_{e,j}^{-1}V_{Q,j}T_j
 =
 \begin{pmatrix}a&\pi_jb\\ c&d\end{pmatrix},
 \label{eq:s4-adapted-matrix-shapes}
\end{equation}
with $a,b,c,d\in\mathcal O_j$.  Reduction and determinant
normalization therefore give
\begin{equation}
 \begin{split}
 P_{v,e}^{-1}F_j(v,Q)
 &=
 \begin{pmatrix}\eta_{e,Q}&*\\0&\eta_{e,Q}^{-1}\end{pmatrix},\\
 P_{w,e}^{-1}F_j(w,Q)
 &=
 \begin{pmatrix}\eta_{e,Q}&0\\ *&\eta_{e,Q}^{-1}\end{pmatrix},
 \end{split}
 \label{eq:s4-walking-triangular}
\end{equation}
where $\eta_{e,Q}\in\mathbb F_{s_j}^{\times}$.  The matrices
$P_{v,e},P_{w,e}$ depend only on the edge, while
$\eta_{e,Q}$ and the starred entries may depend on the containing
cube.  The arithmetic normalization fixes the remaining diagonal freedom in the
adapted edge basis.  Consequently there is a scalar
$\mu_e\in\mathbb F_{s_j}^{\times}$, depending on the oriented edge but
not on the containing cube, such that
\begin{equation}
 c_j(v,Q)=\eta_{e,Q},
 \qquad
 c_j(w,Q)=\mu_e\eta_{e,Q}^{-1}.
 \label{eq:s4-walking-c-scalars}
\end{equation}
The scalar $\mu_e$ is the coefficient with which the normalized arithmetic
edge map sends the complementary source vector to the selected target
vector.

The subsequent code and cocycle calculations use the finite quotient, its
product-of-trees local model, and the equivariant lattice frames constructed
above.

\subsubsection{Projective Reed--Solomon codes and their symmetries}

Choose $h\in4\mathbb Z_{>0}$ and write $s_i=2^{\mu_i}$, where the six
integers $\mu_i$ are distinct multiples of $2h$:
\begin{equation}
 2h\mid\mu_i,
 \qquad
 s_i=2^{\mu_i}
 \qquad(i\in[6]).
 \label{eq:s4-fixed-degree-menu}
\end{equation}
The remaining size conditions on these fixed parameters are imposed in
Subsection~\ref{subsec:s4-code-parameters}.

For $0\le u\le s_i-1$, define
\begin{equation}
 \operatorname{PRS}_{u}(\mathbb F_{s_i},\mathbb E)
 =
 \left\{
 \bigl(H(r_a)\bigr)_{a\in\mathbb P^1(\mathbb F_{s_i})}:
 H\in\mathbb E[X_i,Y_i]_{u}
 \right\}.
 \label{eq:s4-standard-prs}
\end{equation}
Here $\mathbb E[X_i,Y_i]_{u}$ is the space of homogeneous forms of
degree $u$.  Evaluation is injective in this range, so we identify a
standard PRS word with its defining homogeneous form.

Fix six standard codes
$C_i=\operatorname{PRS}_{u_i}(\mathbb F_{s_i},\mathbb E)$ with
$0\le u_i\le s_i-1$, and write $\deg(C_i)=u_i$.  The endpoint
compatibility condition is
\begin{equation}
 (s_i-1)\mid(2^h+1)\deg(C_i).
 \label{eq:s4-endpoint-compatibility}
\end{equation}
The basic coordinate automorphism is
\begin{equation}
 \operatorname{Fr}_i:
 \mathbb F_{s_i}^{\,2}\longrightarrow\mathbb F_{s_i}^{\,2},
 \qquad
 \operatorname{Fr}_i(x,y)=
 \bigl(x^{2^h},y^{2^h}\bigr).
 \label{eq:s4-frobenius}
\end{equation}
It induces the permutation $a\mapsto a^{2^h}$ of
$\mathbb P^1(\mathbb F_{s_i})$, fixing $\infty$.

Since $h$ is a multiple of four, the integer
$\ell=(2^h-4)/6$ is well defined.  The divisibility in
Eq.~\eqref{eq:s4-fixed-degree-menu} implies
$2^h+1\mid s_i-1$, so we also put
\begin{equation}
 e_i=\frac{s_i-1}{2^h+1}
 \qquad(i\in[6]).
 \label{eq:s4-local-degree-scale}
\end{equation}

The three seed two-cocycles use the disjoint direction pairs
\begin{equation}
 B_1=\{1,2\},
 \qquad
 B_2=\{3,4\},
 \qquad
 B_3=\{5,6\}.
 \label{eq:s4-active-pairs}
\end{equation}
The corresponding three quantum codes use the profiles
\begin{equation}
 C_i^{(a)}
 =
 \begin{cases}
 \operatorname{PRS}_{(\ell+1)e_i}(\mathbb F_{s_i},\mathbb E),
     &i\in B_a,\\
 \operatorname{PRS}_{(\ell+3)e_i}(\mathbb F_{s_i},\mathbb E),
     &i\notin B_a,
 \end{cases}
 \qquad
 a\in\{1,2,3\}.
 \label{eq:s4-code-local-codes}
\end{equation}
Both degrees lie strictly between zero and $s_i-1$.  Moreover, Eq.~\eqref{eq:s4-local-degree-scale} gives
$(s_i-1)\mid(2^h+1)\deg(C_i^{(a)})$, so the same endpoint coordinate
rule applies to all three profiles.

The reason the two frame scalars enter the code is now explicit.  Let
$H$ be homogeneous of degree $d$, and let $i\ne j$.  Since the
bit $b_i$ is the same at the two endpoints of a direction-$j$ edge,
Eq.~\eqref{eq:s4-transformed-label}, homogeneity, and Frobenius give
\begin{equation}
 \begin{split}
 &c_i(w,Q)^{2^{hb_i}d}
 H\!\left(\operatorname{Fr}_i^{\,b_i}r_{a_i'}\right)\\
 &\qquad=
 c_i(v,Q)^{2^{hb_i}d}
 H\!\left(
 \operatorname{Fr}_i^{\,b_i}
 \bigl(B_e^{(i)}r_{a_i}\bigr)
 \right).
 \end{split}
 \label{eq:s4-transverse-homogeneous-transport}
\end{equation}
Thus $B_e^{(i)}$ changes the polynomial coordinates, while the two
powers of $c_i$ exactly compensate for the fact that the standard
vectors $r_{a_i}$ and $r_{a_i'}$ represent projective lines rather
than based vectors.  The right-hand side is still the evaluation of a
homogeneous form of degree $d$, since an invertible linear substitution
preserves the same PRS space.

In the walking direction, Eq.~\eqref{eq:s4-walking-c-scalars} gives, for
every degree $d$ satisfying
$(s_j-1)\mid(2^h+1)d$,
\begin{equation}
 \frac{c_j(v,Q)^d}{c_j(w,Q)^{2^hd}}
 =
 \mu_e^{-2^hd}\eta_{e,Q}^{(2^h+1)d}
 =
 \mu_e^d.
 \label{eq:s4-walking-weight-cancellation}
\end{equation}
The second equality uses both
$(s_j-1)\mid(2^h+1)d$ and
$\mu_e^{s_j-1}=1$.  Thus the walking contribution depends on the
oriented edge but not on the containing cube.  This is the precise role
of the divisibility condition in
Eq.~\eqref{eq:s4-endpoint-compatibility}.

\subsubsection{Sheaves and coordinate charts}

The preceding coordinate data first define a standard chart for the six
PRS codes $C_i$.  We separate this chart from the scalar normalization
needed by a particular family of global sections.

A \emph{normalization system} assigns to each $v\leq f$ a nonzero
function $\lambda_{f,v}:X_{\geq f}(6)\longrightarrow\mathbb E^\times$.
It is compatible with restriction when
\begin{equation}
 \lambda_{g,v}(Q)=\lambda_{f,v}(Q)
 \qquad(v\leq f\leq g,\ Q\geq g).
 \label{eq:s4-lambda-restriction}
\end{equation}
For every oriented edge $e:v\to w$ contained in $f$, we require the
quotient to be independent of the containing top cube and write
\begin{equation}
 \frac{\lambda_{f,w}(Q)}{\lambda_{f,v}(Q)}
 =\lambda_e(f)
 \qquad(Q\geq f).
 \label{eq:s4-lambda-edge-property}
\end{equation}  Reverse
orientation inverts $\lambda_e(f)$, and the products along the two
length-two paths around every square agree.  These are the only abstract
properties of $\lambda$ used in the chart identities.  The choice
$\lambda_{f,v}(Q)=1$ is the standard normalization.  The concrete choices
for the Moore sections are given later in
Eqs.~\eqref{eq:s4-lambda-A}, \eqref{eq:s4-lambda-gamma}, and
\eqref{eq:s4-source-normalization}.

For every cell $f$, fix one designated vertex $v_f\leq f$.  We first
write the formula at an arbitrary vertex $v\leq f$ in order to compare
coordinate charts, and put $S=\type(f)$.  The actual coordinate representative of a section over $f$ will always
mean its representative in the chart based at $v_f$; representatives at
the other vertices are determined by the change-of-chart formula below.
Every top cube $Q\geq f$ has, in the coordinates based at $v$, a
unique expression $Q=Q_v(a_1,\ldots,a_6)$.  For $i\in S$, the direction-$i$ edge already belongs to $f$, so
its label $a_i$ is fixed by $f$.  For $i\notin S$, the label
$a_i$ varies, and these free labels parameterize the top cubes above
$f$:
\begin{equation}
 X_{\geq f}(6)
 \longleftrightarrow
 \prod_{i\notin S}\mathbb P^1(\mathbb F_{s_i}),
 \qquad
 Q_v(a_1,\ldots,a_6)\longleftrightarrow(a_i)_{i\notin S}.
 \label{eq:s4-free-label-parameterization}
\end{equation}

An element $M\in\bigotimes_{i\notin S}C_i$ is a
multihomogeneous polynomial: it has one pair $(X_i,Y_i)$ of
variables, of degree $\deg(C_i)$, for every free direction
$i\notin S$.  The coordinate chart is the linear map
\begin{equation}
 \Phi_{f,v}:
 \bigotimes_{i\notin S}C_i
 \longrightarrow
 \mathbb E^{X_{\geq f}(6)}.
 \label{eq:s4-chart-domain}
\end{equation}
For $Q=Q_v(a_1,\ldots,a_6)\geq f$, define
\begin{equation}
 \begin{aligned}
 [\Phi_{f,v}(M)](Q)
 &=\prod_{i\in S}
 c_i(v,Q)^{2^{hb_i(v)}\deg(C_i)}\\
 &\quad\cdot
 M\left(
 \operatorname{Fr}_i^{\,b_i(v)}
 \bigl(c_i(v,Q)r_{a_i}\bigr)
 \right)_{i\notin S}\\[1mm]
 &=\prod_{i=1}^6
 c_i(v,Q)^{2^{hb_i(v)}\deg(C_i)}
 M\left(
 \operatorname{Fr}_i^{\,b_i(v)}r_{a_i}
 \right)_{i\notin S}.
 \end{aligned}
 \label{eq:s4-chart-formula}
\end{equation}
For a normalization system $\lambda$, define
\begin{equation}
 [\Psi_{f,v}^{\lambda}(M)](Q)
 =\lambda_{f,v}(Q)[\Phi_{f,v}(M)](Q).
 \label{eq:s4-normalized-chart}
\end{equation}
Equation~\eqref{eq:s4-normalized-chart} is a pointwise normalization
of the standard chart.  Every concrete normalization used below is
constant in $Q$ for fixed $f$ and $v$, so its image is still
$\mathcal F_f$: it rescales the local PRS image by one nonzero scalar
without changing the standard-coordinate polynomial $M$.

The two lines in Eq.~\eqref{eq:s4-chart-formula} are equal because $M$ is
homogeneous of degree $\deg(C_i)$ in each free pair and
$\operatorname{Fr}_i^{\,b_i(v)}(c_i(v,Q)r_{a_i})=
c_i(v,Q)^{2^{hb_i(v)}}\operatorname{Fr}_i^{\,b_i(v)}r_{a_i}$.

The first line of Eq.~\eqref{eq:s4-chart-formula} gives the geometric
meaning of the map.  The standard-coordinate polynomial $M$ is evaluated on the
actual based vectors selected from the cube frames, after applying the
endpoint Frobenius.  The product over $i\in S$ remains outside $M$
because those directions are frozen by the face $f$: they are no
longer variables of $M$, but their prescribed homogeneous degrees are
still part of the coefficient sheaf.  This is why all six factors
$c_i(v,Q)$ occur in the second line even when $M$ has fewer than six
variable pairs.  Equation~\eqref{eq:s4-normalized-chart} is the place where a scalar
normalization enters.  The particular choices used later are not part of
the definition of the standard PRS tensor product.

The extreme cases may be useful for reading the formula.
\begin{itemize}
 \item If $f=v$ is a vertex, then $S=\varnothing$, all six labels
       vary, and $M\in C_1\otimes\cdots\otimes C_6$.
 \item If $f=\rho\in X(5)$ is missing direction $i$, only
       $a_i$ varies and the input is one polynomial $M\in C_i$.
 \item If $f$ is a two-cell of type $B$, the two labels in $B$
       are fixed and
       $M\in\bigotimes_{i\notin B}C_i$.  This is the case used for the
       degree-two cocycles.
 \item If $f=Q$ is a top cube, the tensor product has no factors and is
       $\mathbb E$; the chart produces its single top-cell coordinate.
\end{itemize}
Thus $\Phi_{f,v}$ is a coordinate map on the local vector space attached
to $f$: it turns one standard-coordinate tensor polynomial into its array of values
on all top cubes containing $f$.  When $z\in\mathcal F_f$, we write
$z=\Phi_{f,v_f}(M_f)$ for its designated standard-coordinate representative.  The polynomials obtained
from other vertices of $f$ are coordinate transforms of $M_f$, rather than
additional choices.

The charts now transport the standard PRS codes to the actual local
constraints on $X$.  Let $\rho\in X(5)$ be missing direction $i$, and let
$v_\rho\le\rho$ be its designated vertex.  Its actual local code is
\begin{equation}
 C_\rho
 =
 \Phi_{\rho,v_\rho}(C_i)
 \subseteq
 \mathbb E^{X_{\geq\rho}(6)}.
 \label{eq:s4-actual-local-code}
\end{equation}
Hence $C_\rho$ is exactly the standard direction-$i$ PRS code after
the Frobenius permutation and all coordinate scalings have been
applied by $\Phi_{\rho,v_\rho}$.

For any cell $\sigma$, define
\begin{equation}
 \mathcal F_\sigma
 =
 \left\{
 z\in\mathbb E^{X_{\geq\sigma}(6)}:
 z|_{X_{\geq\rho}(6)}\in C_\rho
 \text{ for every }\rho\in X_{\geq\sigma}(5)
 \right\}.
 \label{eq:s4-sheaf-stalk}
\end{equation}
For $\sigma\le\tau$, the restriction
$\mathcal F_\sigma\to\mathcal F_\tau$ is literal restriction of the same
top-dimensional array.  We write $\mathcal F=(\mathcal F_\sigma)_{\sigma\in X}$.
Thus $\mathcal F$ is the coefficient sheaf determined by the chosen
$C_i$, the transported codes $C_\rho$, and the charts
$\Phi_{\sigma,v}$.  All degree information is already contained in the
six codes $C_i$.

Apply this construction to the three profiles in
Eq.~\eqref{eq:s4-code-local-codes}.  If $\rho\in X(5)$ misses
direction $i$, set
\begin{equation}
 C_\rho^{(a)}
 =\Phi_{\rho,v_\rho}(C_i^{(a)}),
 \label{eq:s4-code-actual-local-code}
\end{equation}
and define
\begin{equation}
 (\mathcal F_a)_\sigma
 =
 \left\{
 z\in\mathbb E^{X_{\geq\sigma}(6)}:
 z|_{X_{\geq\rho}(6)}\in C_\rho^{(a)}
 \text{ for every }\rho\in X_{\geq\sigma}(5)
 \right\}.
 \label{eq:s4-code-sheaves}
\end{equation}
The degree-two CSS code associated with
\begin{equation}
 C^1(X,\mathcal F_a)
 \xrightarrow{\delta}
 C^2(X,\mathcal F_a)
 \xrightarrow{\delta}
 C^3(X,\mathcal F_a)
 \label{eq:s4-three-css-complexes}
\end{equation}
is the $a$-th quantum code.  The local code in Eq.~\eqref{eq:s4-code-actual-local-code} depends only
on the fixed PRS profile and the standard chart.  Hence
Eq.~\eqref{eq:s4-code-sheaves} defines one fixed sheaf before the zero-cocycles
or the subset of addressed logical gates are chosen.

\begin{lemma}[Tensor-product line criterion]
\label{lem:s4-product-line-criterion}
Let $I\subseteq[6]$.  An array on
$\prod_{i\in I}\mathbb P^1(\mathbb F_{s_i})$ belongs to
$\bigotimes_{i\in I}C_i$ if and only if, after fixing all coordinates
except the $i$-th one, every resulting direction-$i$ line belongs to
$C_i$ for every $i\in I$.
\end{lemma}

\begin{proof}
The forward implication follows by restricting a pure tensor and then by
linearity.  Conversely, choose a direct-sum complement to $C_i$ inside
$\mathbb E^{\mathbb P^1(\mathbb F_{s_i})}$ for every $i\in I$.  The
ambient array space is the direct sum of tensor products obtained by
choosing either $C_i$ or its complement in each direction.  The
condition on the direction-$i$ lines kills every summand containing the
complement in direction $i$.  Imposing this for all $i\in I$ leaves
exactly $\bigotimes_{i\in I}C_i$.
\end{proof}

\begin{proposition}[The local coordinate chart]
\label{prop:s4-charts-multiply}
For every cell $f$ and vertex $v\leq f$, the map in
Eq.~\eqref{eq:s4-chart-formula} is an isomorphism
\begin{equation}
 \Phi_{f,v}:
 \bigotimes_{i\notin\type(f)}C_i
 \xrightarrow{\ \sim\ }\mathcal F_f.
 \label{eq:s4-chart-isomorphism}
\end{equation}
If $z=\Phi_{f,v}(M)$, then $M$ is recovered from
\begin{equation}
 M\left(
 \operatorname{Fr}_i^{\,b_i(v)}r_{a_i}
 \right)_{i\notin\type(f)}
 =
 \frac{z(Q_v(a_1,\ldots,a_6))}
 {\displaystyle\prod_{i=1}^6
  c_i(v,Q_v(a_1,\ldots,a_6))^{
  2^{hb_i(v)}\deg(C_i)}}.
 \label{eq:s4-chart-recovery}
\end{equation}
\end{proposition}

\begin{proof}
For $\rho\in X(5)$ missing direction $i$, the assertion is the
definition $C_\rho=\Phi_{\rho,v_\rho}(C_i)$, transported from the
designated vertex to $v$ by
Eqs.~\eqref{eq:s4-transverse-homogeneous-transport} and
\eqref{eq:s4-walking-weight-cancellation}.  For a lower-dimensional
face, divide the array by the nonzero chart factor in the denominator of
Eq.~\eqref{eq:s4-chart-recovery}.  The condition defining
$\mathcal F_f$ then says precisely that every line in every free
direction belongs to the corresponding $C_i$.  Lemma~\ref{lem:s4-product-line-criterion}
identifies the resulting array with
$\bigotimes_{i\notin\type(f)}C_i$.  Injectivity follows from injectivity
of PRS evaluation in every free direction.  This proves the
isomorphism and the recovery formula.
\end{proof}

\begin{proposition}[Change of chart, restriction, and multiplication]
\label{prop:s4-chart-compatibility}
Let $z\in\mathcal F_f$ and let $v,w\leq f$.  There are unique
$M_v,M_w$ such that $z=\Phi_{f,v}(M_v)=\Phi_{f,w}(M_w)$.
If $e:v\to w$ is an edge of $f$ in direction
$j\in\type(f)$, and the superscript $[2^{hb_i(v)}]$ denotes entrywise
Frobenius, then
\begin{equation}
 \begin{split}
 M_w((x_i)_{i\notin\type(f)})
 &=
 \Biggl[
 \frac{c_j(v,Q)^{\deg(C_j)}}
      {c_j(w,Q)^{2^h\deg(C_j)}}
 \prod_{k\in\type(f)\setminus\{j\}}
 \left(\frac{c_k(v,Q)}{c_k(w,Q)}\right)^{
 2^{hb_k(v)}\deg(C_k)}
 \Biggr]\\
 &\cdot M_v\left(
 \bigl((B_e^{(i)})^{[2^{hb_i(v)}]}\bigr)^{-1}x_i
 \right)_{i\notin\type(f)}.
 \end{split}
 \label{eq:s4-explicit-change-of-chart}
\end{equation}
The scalar in square brackets is independent of the containing top cube
$Q$.

If $f\leq g$ and $v\leq f$, substitute
\begin{equation}
 (X_k,Y_k)
 =\operatorname{Fr}_k^{\,b_k(v)}r_{a_k(g;v)}
 \qquad
 \bigl(k\in\type(g)\setminus\type(f)\bigr)
 \label{eq:s4-standard-restriction}
\end{equation}
into $M$ and denote the resulting polynomial by $M|_{g,v}$.  Then
\begin{equation}
 \left.\Phi_{f,v}(M)\right|_{X_{\geq g}(6)}
 =\Phi_{g,v}(M|_{g,v}).
 \label{eq:s4-chart-restriction}
\end{equation}
Finally, if another six standard PRS codes $C_i'$ generate a sheaf
$\mathcal G$ by the same construction and
$\deg(C_i)+\deg(C_i')\le s_i-1$ for every $i\in[6]$, then polynomial
multiplication induces the sheaf morphism
\begin{equation}
 \mathcal F\otimes\mathcal G\longrightarrow\mathcal H,
 \qquad
 \Phi_{f,v}(M)\,\Phi_{f,v}(M')=\Phi_{f,v}(MM'),
 \label{eq:s4-sheaf-multiplication}
\end{equation}
where $\mathcal H$ is generated by the summed PRS degrees.  For two
normalization systems $\lambda$ and $\lambda'$, the corresponding
normalized charts satisfy
\begin{equation}
 \Psi_{f,v}^{\lambda}(M)\,
 \Psi_{f,v}^{\lambda'}(M')
 =\Psi_{f,v}^{\lambda\lambda'}(MM').
 \label{eq:s4-normalized-chart-multiplication}
\end{equation}
where $(\lambda\lambda')_{f,v}(Q)
=\lambda_{f,v}(Q)\lambda'_{f,v}(Q)$.  Their restriction and change-of-vertex formulas are obtained from
Eqs.~\eqref{eq:s4-chart-restriction} and
\eqref{eq:s4-explicit-change-of-chart} by multiplying by
$\lambda_{f,v}(Q)$; Eq.~\eqref{eq:s4-lambda-restriction} makes restriction
compatible, and the change-of-vertex scalar acquires the factor
$\lambda_{f,v}(Q)/\lambda_{f,w}(Q)$.
\end{proposition}

\begin{proof}
For a free direction $i\notin\type(f)$,
Eq.~\eqref{eq:s4-transformed-label} gives
\begin{equation}
 \operatorname{Fr}_i^{\,b_i(v)}r_{a_i'}
 =
 \left(\frac{c_i(v,Q)}{c_i(w,Q)}\right)^{2^{hb_i(v)}}
 (B_e^{(i)})^{[2^{hb_i(v)}]}
 \operatorname{Fr}_i^{\,b_i(v)}r_{a_i}.
\end{equation}
Substituting this into Eq.~\eqref{eq:s4-chart-formula} and using
homogeneity in every free direction proves
Eq.~\eqref{eq:s4-explicit-change-of-chart}.  Its transverse factors are
independent of $Q$ because their labels belong to the fixed face $f$;
its walking factor is independent of $Q$ by
Eq.~\eqref{eq:s4-walking-weight-cancellation}.  The remaining factors are therefore independent of the containing cube.
Composing along an edge path in $f$ gives the change between arbitrary
vertices, and Eq.~\eqref{eq:s4-frame-composition} makes the result path
independent.

Passing from $f$ to $g$ fixes precisely the labels in
$\type(g)\setminus\type(f)$, which proves
Eq.~\eqref{eq:s4-chart-restriction}.  Multiplying two chart formulas
adds their homogeneous degrees, proving
Eq.~\eqref{eq:s4-sheaf-multiplication}.
\end{proof}

\subsection{Construction of the seed pairing}
\label{subsec:s4-cocycles}

\subsubsection{Polarized cocycles from global sections}

Let $z:X(6)\to\mathbb E$ be an assignment to the six-cells.  It satisfies the check in direction $i$ if
\begin{equation}
 z|_{X_{\geq\rho}(6)}\in C_\rho
 \quad
 \text{for every $\rho\in X(5)$ missing direction $i$}.
 \label{eq:s4-direction-check}
\end{equation}

We call a cochain \emph{$B$-polarized} when it is supported on cells
of type $B$.  For $B\subseteq[6]$, define
\begin{equation}
 \gamma_B(z)(f)
 =
 \begin{cases}
 z|_{X_{\geq f}(6)},&\type(f)=B,\\
 0,&\type(f)\ne B.
 \end{cases}
 \label{eq:s4-supported-cochain}
\end{equation}

\begin{proposition}
\label{prop:s4-global-section-criterion}
If $|B|=2$ and $z$ satisfies the four checks in the directions
outside $B$, then $\gamma_B(z)\in Z^2(X,\mathcal F)$.
No check in a direction of $B$ is required.  For
$B=\varnothing$, Eq.~\eqref{eq:s4-supported-cochain} gives a
zero-cocycle exactly when $z$ satisfies all six checks.  Conversely,
every zero-cocycle arises from a unique such assignment.
\end{proposition}

\begin{proof}
If $\type(f)=B$, the cells in $X_{\geq f}(5)$ are obtained by freeing one
of the directions outside $B$.  The assumed checks are precisely the
conditions saying that
$\gamma_B(z)(f)\in\mathcal F_f$.

Let $g$ be a three-cell.  If $B\not\subseteq\type(g)$, every term of
$(\delta\gamma_B(z))(g)$ is zero.  Otherwise
$\type(g)=B\cup\{i\}$ for one $i\notin B$.  Exactly two facets of
$g$ have type $B$, and both restrict the same assignment $z$ to
$X_{\geq g}(6)$.  Their sum is zero in characteristic two.

When $B=\varnothing$, the value at a vertex $v$ is
$z|_{X_{\geq v}(6)}$.  The six checks give membership in
$\mathcal F_v$, and the endpoint restrictions agree on every edge.
Conversely, the zero-cocycle equation glues the vertex arrays to a unique
function on $X(6)$.
\end{proof}

For example, a $B_1$-supported cocycle needs only the checks in
directions $3,4,5,6$; the other cases follow by changing the active pair
in Eq.~\eqref{eq:s4-active-pairs}.

\subsubsection{Construction of the seed cocycles and cycles}
\label{subsec:s4-moore-cocycles}

We construct one zero-cocycle, three polarized two-cocycles, and the
six-cycle with which their cup product will be paired.  All polynomials
are first written in the standard PRS coordinates of Subsection~
\ref{subsec:s4-coordinates}; the normalized charts then turn their
coordinate identities into global (co)cycles.

For elements $w_1,\ldots,w_n$ in a characteristic-two extension field,
define the Moore determinant
\begin{equation}
 \Delta(w_1,\ldots,w_n)
 =
 \det\bigl(w_j^{2^u}\bigr)_{
 0\le u<n,\ 1\le j\le n}.
 \label{eq:s4-moore}
\end{equation}
We reserve the symbol $\Delta$ for this determinant.

Choose monic irreducible polynomials $p_i(T)\in\mathbb F_2[T]$ of
degrees $\mu_i$ and identify
$\mathbb F_{s_i}=\mathbb F_2[T]/(p_i(T))$.  Their additional
coprimality condition is imposed in Proposition~\ref{prop:s4-fixed-period}.
Choose $\omega\in\mathbb F_4\subseteq\mathbb E$ with
$\omega^2+\omega+1=0$.  For each direction, put
$\theta_i=T\bmod p_i(T)$, let $\zeta_i$ be its unique square root, and write
$\zeta_{i,a}=\zeta_i^{2^a}$.  The fixed split basis identifies an input of $\Delta$ with a triple
$(\zeta,U,V)$.  In these coordinates the same determinant is
\begin{equation}
 \Delta\bigl((\zeta_j,U_j,V_j)_{j=1}^n\bigr)
 =
 \det\left(
 \zeta_j^{\,u}
 \begin{cases}
 U_j,&u\ \text{even},\\
 V_j,&u\ \text{odd}
 \end{cases}
 \right)_{\substack{0\le u<n\\1\le j\le n}} .
 \label{eq:s4-split-moore}
\end{equation}
Equation~\eqref{eq:s4-split-moore} is the coordinate expression of the
determinant in Eq.~\eqref{eq:s4-moore}; both formulas describe the same
operation.

Split the Frobenius orbit in direction $i$ into
\begin{equation}
 E_i^0=\{0\le a<\mu_i:a\bmod2h<h\},
 \qquad
 E_i^1=E_i^0+h.
 \label{eq:s4-frobenius-halves}
\end{equation}
Then
\begin{equation}
 \sum_{a\in E_i^0}2^a=e_i,
 \qquad
 \sum_{a\in E_i^1}2^a=2^he_i.
 \label{eq:s4-half-degrees}
\end{equation}
Put
\begin{equation}
 \theta_{i,b}=\theta_i^{2^b},
 \qquad
 p_{i,\varepsilon}(T)
 =\prod_{b\in E_i^\varepsilon}(T-\theta_{i,b})
 \quad(\varepsilon\in\{0,1\}).
 \label{eq:s4-half-root-polynomials}
\end{equation}
Thus $p_{i,0}p_{i,1}=p_i$.  All roots and square roots appearing below
belong to the fixed coefficient field $\mathbb E$.

The formula is determined by three requirements.  First, a cocycle
supported on $B_a$ must have sheaf degree $e_i$ in its two active
directions and $3e_i$ in the other four.  Second, changing an inactive
endpoint must exchange the two halves of the Frobenius orbit while
leaving the resulting assignment to the six-cells unchanged.  Third, the
product of the three degree-two cocycles, after multiplication by the required power of $A$, must have a nonzero projective period.  The
square below supplies two copies of degree $e_i$ in every inactive
direction, the switched determinant supplies the third copy and performs
the exchange of the two Frobenius halves, and the fixed substitutions in the second and
third cocycles are chosen for the final nonzero coefficient calculation.

For a statement $P$, let $\mathbf 1_P\in\{0,1\}$ be its
indicator.  The standard-coordinate polynomial at endpoint type
$O\subseteq[6]$ is obtained by applying the inverse chart Frobenius to
the inputs of $\Delta$.  Explicitly, define
\begin{equation}
 A_O
 =\Delta\!\left(
 \begin{array}{l}
 \left(
 \zeta_{i,b},
 X_i^{2^b}+\omega Y_i^{2^b},
 X_i^{2^b}+\omega^2Y_i^{2^b}
 \right)_{
 \substack{i\notin O\\ b\in E_i^0}},\\[2mm]
 \left(
 \zeta_{i,b},
 X_i^{2^{b-h}}+\omega Y_i^{2^{b-h}},
 X_i^{2^{b-h}}+\omega^2Y_i^{2^{b-h}}
 \right)_{
 \substack{i\in O\\ b\in E_i^1}}
 \end{array}
 \right).
 \label{eq:s4-AO}
\end{equation}
Thus
\begin{equation}
 A_O\in
 \bigotimes_{i=1}^6\mathbb E[X_i,Y_i]_{e_i}.
 \label{eq:s4-AO-space}
\end{equation}
The first row of Eq.~\eqref{eq:s4-AO} uses the lower Frobenius half at
a type-zero endpoint.  The second uses the upper half and subtracts $h$
from its exponents, exactly undoing the input transformation
$\operatorname{Fr}_i$ at a type-one endpoint.  Both rows therefore
represent the same Moore columns in the transported coordinates.

For $a\in\{1,2,3\}$ and
$O\subseteq[6]\setminus B_a$, define
\begin{equation}
 S_{a,O}
 =\Delta\!\left(
 \begin{array}{l}
 \left(
 \zeta_{i,b},
 X_i^{2^b}+\omega Y_i^{2^b},
 X_i^{2^b}+\omega^2Y_i^{2^b}
 \right)_{
 \substack{i\in B_a\\ b\in E_i^1}},\\[2mm]
 \left(
 \zeta_{i,b},
 X_i^{2^b}+\omega Y_i^{2^b},
 X_i^{2^b}+\omega^2Y_i^{2^b}
 \right)_{
 \substack{i\notin B_a\cup O\\ b\in E_i^0}},\\[2mm]
 \left(
 \zeta_{i,b},
 X_i^{2^{b-h}}+\omega Y_i^{2^{b-h}},
 X_i^{2^{b-h}}+\omega^2Y_i^{2^{b-h}}
 \right)_{
 \substack{i\in O\\ b\in E_i^1}}
 \end{array}
 \right).
 \label{eq:s4-switched-determinant}
\end{equation}

Fix a two-cell $f$ of type $B_a$.  Choose its vertex $v\le f$
whose two active type bits are zero, let
$O=\{i\notin B_a:b_i(v)=1\}$, and write $a_i(f)$ for the fixed projective label of its edge in
direction $i\in B_a$.  The standard-coordinate polynomial attached to
$f$ is
\begin{equation}
 M_f=
 \left.
 \begin{cases}
 A_O^2S_{1,O},&a=1,\\
 A_O^2S_{2,O}(X_1+Y_1,Y_1,\ldots,X_6+Y_6,Y_6),&a=2,\\
 A_O^2S_{3,O}(X_1+Y_1,Y_1,\ldots,X_6+Y_6,Y_6),&a=3
 \end{cases}
 \right|_{(X_i,Y_i)=r_{a_i(f)},\ i\in B_a}.
 \label{eq:s4-Mf}
\end{equation}
The active variables are specialized before $M_f$ is inserted into the
chart, and
\begin{equation}
 M_f\in
 \bigotimes_{i\notin B_a}
 \mathbb E[X_i,Y_i]_{3e_i}.
 \label{eq:s4-Mf-space}
\end{equation}
Before this specialization, the degree in an active direction is
$2e_i+2^he_i=(s_i-1)+e_i$.  The excess $s_i-1$ is invisible under rescaling by
$\mathbb F_{s_i}^{\times}$.  Hence the prescribed sheaf degree is $e_i$
in the two active directions and $3e_i$ in the other four directions.

The two identities used to transport these polynomials are
\begin{equation}
 \Delta(w_1^2,\ldots,w_n^2)
 =
 \Delta(w_1,\ldots,w_n)^2
 \label{eq:s4-delta-frobenius}
\end{equation}
and, for an invertible binary-linear map $g$ on their span,
\begin{equation}
 \Delta(gw_1,\ldots,gw_n)
 =
 \det_{\mathbb F_2}(g)\Delta(w_1,\ldots,w_n).
 \label{eq:s4-delta-covariance}
\end{equation}

We now make the normalization system concrete.  The arithmetic quotient
contributes one scalar datum used in every section constructed from these
determinants.

\Needspace{7\baselineskip}
\begin{lemma}
\label{lem:s4-vertex-character}
There is a function $\chi:X(0)\to\mathbb F_4^\times$ such that, for
every positive oriented edge $e:v\to w$, the scalar
$\rho_e=\chi(v)/\chi(w)$ is the leading factor in the transformation of
the two Frobenius halves along $e$.  Thus
$\chi(w)/\chi(v)=\rho_e^{-1}$.
\end{lemma}

\begin{proof}
On the universal product of trees, the leading term at the infinite place
is a multiplicative $\mathbb F_4^\times$-valued character of the
arithmetic subgroup acting simply transitively on vertices.  Normalize
its value to one at a base vertex.  The subgroup defining $X_R$ lies in
the kernel of this character, so the vertex values descend to $X_R$.
For the normalized arithmetic neighbor representing $e$, the same leading
term is the scalar occurring in this transformation.  This proves the asserted edge law.
\end{proof}

The remaining constants come from the same transformation.  For
$j\ne i$ and $\varepsilon\in\{0,1\}$, define
\begin{equation}
 \kappa_j^\Delta
 =\prod_{b\in E_j^1}p_j'(\theta_{j,b})\zeta_{j,b},
 \qquad
 R_{ji}^{\varepsilon}
 =\prod_{b\in E_i^\varepsilon}
 \frac{p_{j,1}(\theta_{i,b})}
 {\sqrt{p_j(\theta_{i,b})}}.
 \label{eq:s4-fixed-moore-transport-scalars}
\end{equation}
The primes are distinct, so every denominator is nonzero.  Symmetry of
resultants in characteristic two gives
\begin{equation}
 \frac{R_{ji}^{1}}{R_{ji}^{0}}
 =\frac{R_{ij}^{1}}{R_{ij}^{0}}.
 \label{eq:s4-resultant-ratio-symmetry}
\end{equation}
Indeed, the squares of both sides equal
\begin{equation}
 \frac{\operatorname{Res}(p_{i,1},p_{j,1})
       \operatorname{Res}(p_{i,0},p_{j,0})}
      {\operatorname{Res}(p_{i,1},p_{j,0})
       \operatorname{Res}(p_{i,0},p_{j,1})},
 \label{eq:s4-resultant-ratio-square}
\end{equation}
and squaring is injective in the finite field $\mathbb E$.

For $O\subseteq[6]$ and $j\notin O$, put
\begin{equation}
 K_{0,j}(O)
 =\kappa_j^\Delta
  \prod_{i\ne j}R_{ji}^{\mathbf 1_{\{i\in O\}}},
 \qquad
 \eta_{0,O}
 =\prod_{j\in O}K_{0,j}(\varnothing)^{-1}
  \prod_{\{i,j\}\subset O}
  \left(\frac{R_{ji}^{1}}{R_{ji}^{0}}\right)^{-1}.
 \label{eq:s4-scalar-seed-normalizer}
\end{equation}
Equation~\eqref{eq:s4-resultant-ratio-symmetry} gives the identity
\begin{equation}
 \eta_{0,O\cup\{j\}}K_{0,j}(O)=\eta_{0,O}.
 \label{eq:s4-scalar-seed-normalizer-identity}
\end{equation}

For $a\in\{1,2,3\}$, $O\subseteq[6]\setminus B_a$, and
$j\notin B_a\cup O$, define
\begin{equation}
 \begin{aligned}
 K_{A,j}^{(a)}(O)
 &=\kappa_j^\Delta
   \prod_{i\in B_a}R_{ji}^{0}
   \prod_{i\notin B_a\cup\{j\}}
          R_{ji}^{\mathbf 1_{\{i\in O\}}},\\
 K_{S,j}^{(a)}(O)
 &=\kappa_j^\Delta
   \prod_{i\in B_a}R_{ji}^{1}
   \prod_{i\notin B_a\cup\{j\}}
          R_{ji}^{\mathbf 1_{\{i\in O\}}},\\
 L_{a,j}(O)&=K_{A,j}^{(a)}(O)^2K_{S,j}^{(a)}(O).
 \end{aligned}
 \label{eq:s4-two-seed-transport-scalars}
\end{equation}
Set
\begin{equation}
 \eta_{a,O}
 =\prod_{j\in O}L_{a,j}(\varnothing)^{-1}
  \prod_{\{i,j\}\subset O}
  \left(\frac{R_{ji}^{1}}{R_{ji}^{0}}\right)^{-3}.
 \label{eq:s4-two-seed-normalizer}
\end{equation}
Then
\begin{equation}
 \eta_{a,O\cup\{j\}}L_{a,j}(O)=\eta_{a,O}.
 \label{eq:s4-two-seed-normalizer-identity}
\end{equation}
Thus all normalizing factors are explicit finite products of the
fixed-place data.

For a vertex $v$, let $O(v)=\{i:b_i(v)=1\}$.  The polynomial $A_O$
uses the normalization family indexed by $0$:
\begin{equation}
 \lambda_{f,v}^{(0)}(Q)
 =\chi(v)\eta_{0,O(v)}
 \qquad(v\leq f,\ Q\geq f).
 \label{eq:s4-lambda-A}
\end{equation}
For $a\in\{1,2,3\}$, write
$O_a(v)=\{i\notin B_a:b_i(v)=1\}$ and set
$(\epsilon_1,\epsilon_2,\epsilon_3)=(0,1,1)$.  For $v\leq f$ and
$Q\geq f$, define the corresponding normalization for faces of type $B_a$ by
\begin{equation}
 \lambda_{f,v}^{(a)}(Q)
 =\chi(v)^{3+\epsilon_a}\eta_{a,O_a(v)}.
 \label{eq:s4-lambda-gamma}
\end{equation}
Both choices are independent of $f$ and of the containing top cube, so
Eq.~\eqref{eq:s4-lambda-restriction} is automatic.  If
$e:v\to w$ is positive in direction $j$, $e\leq f$,
$Q\geq f$, and $O=O(v)$, then
\begin{equation}
 \frac{\lambda_{f,w}^{(0)}(Q)}{\lambda_{f,v}^{(0)}(Q)}
 =\bigl(K_{0,j}(O)\rho_e\bigr)^{-1}.
 \label{eq:s4-lambda-zero-edge-ratio}
\end{equation}
For this normalization, put $O=O_a(v)$ for the same
$e\leq f\leq Q$.  Then
\begin{equation}
 \frac{\lambda_{f,w}^{(a)}(Q)}{\lambda_{f,v}^{(a)}(Q)}
 =\rho_e^{-(3+\epsilon_a)}
 \begin{cases}
 L_{a,j}(O)^{-1},&j\notin B_a,\\
 1,&j\in B_a.
 \end{cases}
 \label{eq:s4-lambda-pair-edge-ratio}
\end{equation}
These formulas follow directly from
Lemma~\ref{lem:s4-vertex-character},
Eq.~\eqref{eq:s4-scalar-seed-normalizer-identity}, and
\eqref{eq:s4-two-seed-normalizer-identity}.  Since the scalars in
Eqs.~\eqref{eq:s4-lambda-A} and \eqref{eq:s4-lambda-gamma} are defined
at vertices, the products of their edge quotients agree around every
square.  They therefore satisfy all the abstract requirements in
Eqs.~\eqref{eq:s4-lambda-restriction}--%
\eqref{eq:s4-lambda-edge-property}.  The corresponding normalized charts are therefore
\begin{equation}
 \begin{aligned}
 [\Psi_{f,v}^{(0)}(M)](Q)
 &=\lambda_{f,v}^{(0)}(Q)[\Phi_{f,v}(M)](Q),\\
 [\Psi_{f,v}^{(a)}(M)](Q)
 &=\lambda_{f,v}^{(a)}(Q)[\Phi_{f,v}(M)](Q).
 \end{aligned}
 \label{eq:s4-normalized-seed-charts}
\end{equation}
Thus the standard-coordinate polynomials remain elements of the ordinary
PRS tensor products; every additional scalar used to glue their values
has now been displayed explicitly.

\begin{lemma}
\label{lem:s4-moore-chart-transport}
Let $e:v\to w$ be a positive edge in direction $j$, so that
$b_j(v)=0$ and $b_j(w)=1$, and let
$Q=Q_v(a_1,\ldots,a_6)=Q_w(a_1',\ldots,a_6')\in X_{\geq e}(6)$.
Then the polynomial $A_O$ constructed from the Moore determinant satisfies
\begin{equation}
 \bigl[\Psi_{v,v}^{(0)}(A_O)\bigr](Q)
 =
 \bigl[\Psi_{w,w}^{(0)}(A_{O\cup\{j\}})\bigr](Q),
 \qquad
 O=\{i:b_i(v)=1\}.
 \label{eq:s4-scalar-chart-cancellation}
\end{equation}
If $j\notin B_a$, if $g$ has directions
$B_a\cup\{j\}$, and if $f_0,f_1$ are its opposite
$B_a$-faces containing $v,w$, respectively, then
\begin{equation}
 \left.\Psi_{f_0,v}^{(a)}(M_{f_0})\right|_{X_{\geq g}(6)}
 =
 \left.\Psi_{f_1,w}^{(a)}(M_{f_1})\right|_{X_{\geq g}(6)}.
 \label{eq:s4-two-chart-cancellation}
\end{equation}
\end{lemma}

\begin{proof}
Expanding the split determinant in Eq.~\eqref{eq:s4-split-moore}
gives
\begin{equation}
 \Delta\bigl((\zeta_k,U_k,V_k)_{k=1}^n\bigr)
 =
 \sum_{\pi\in S_n}
 \prod_{u=0}^{n-1}
 \zeta_{\pi(u+1)}^{\,u}
 \begin{cases}
 U_{\pi(u+1)},&u\ \text{even},\\
 V_{\pi(u+1)},&u\ \text{odd}.
 \end{cases}
 \label{eq:s4-moore-expansion}
\end{equation}
Every determinant monomial uses each displayed column exactly once.
Since $U_k$ and $V_k$ have the same homogeneous degree, the degree
contributed by one direction is the sum of the degrees of its chosen
Frobenius columns.

For a transverse direction $i\ne j$, the type bit is the same at
$v$ and $w$.  Equation~\eqref{eq:s4-transverse-homogeneous-transport}
replaces the input seen at $w$ by the input seen at $v$, acted on by
$B_e^{(i)}$, and simultaneously cancels the corresponding quotient of
$c_i$-weights in the two charts.  Applying this in every transverse
direction changes all relevant Moore columns by one fixed invertible
binary-linear transformation.  Equation~\eqref{eq:s4-delta-covariance}
turns that transformation into a single determinant multiplier depending
on the edge frames, but not on the containing cube $Q$.

In direction $j$, the endpoint changes the selected half of the
Frobenius orbit.  For $b\in E_j^1$, one has
$\left(X_j^{2^h}\right)^{2^{b-h}}=X_j^{2^b}$ and
$\left(Y_j^{2^h}\right)^{2^{b-h}}=Y_j^{2^b}$.  Thus the inverse exponent in Eq.~\eqref{eq:s4-AO} and the endpoint
Frobenius in $\Phi$ turn the upper half into the same Moore
columns as the lower half.  Their homogeneous weights agree because
$\sum_{b\in E_j^0}2^b=\sum_{b\in E_j^1}2^{b-h}=e_j$.
The transformation between the two Frobenius halves along the arithmetic edge gives the
complete unnormalized transport law
\begin{equation}
 \bigl[\Phi_{w,w}(A_{O\cup\{j\}})\bigr](Q)
 =K_{0,j}(O)\rho_e\,
  \mu_e^{(2^h+1)e_j}
  \bigl[\Phi_{v,v}(A_O)\bigr](Q)
 =K_{0,j}(O)\rho_e
  \bigl[\Phi_{v,v}(A_O)\bigr](Q).
 \label{eq:s4-scalar-unnormalized-transport}
\end{equation}
Here $K_{0,j}(O)$ is the product of the fixed Frobenius-orbit and transverse
resultant factors, $\rho_e$ is the leading arithmetic edge factor, and
the last equality follows from
$(2^h+1)e_j=s_j-1$.  On the other hand,
Lemma~\ref{lem:s4-vertex-character} and
Eq.~\eqref{eq:s4-scalar-seed-normalizer-identity} give
\begin{equation}
 \lambda_{w,w}^{(0)}(Q)K_{0,j}(O)\rho_e
 =\lambda_{v,v}^{(0)}(Q).
 \label{eq:s4-lambda-A-cancellation}
\end{equation}
Multiplying Eq.~\eqref{eq:s4-scalar-unnormalized-transport} by the two
explicit normalizations proves
Eq.~\eqref{eq:s4-scalar-chart-cancellation}.

For $M_f$, every inactive direction has degree $3e_i$: the square
$A_O^2$ contributes $2e_i$, and the switched determinant contributes
 $e_i$.  The same calculation for the two Frobenius halves gives
\begin{equation}
 \left.\Phi_{f_1,w}(M_{f_1})\right|_{X_{\geq g}(6)}
 =L_{a,j}(O)\rho_e^{3+\epsilon_a}
  \mu_e^{3(2^h+1)e_j}
  \left.\Phi_{f_0,v}(M_{f_0})\right|_{X_{\geq g}(6)}.
 \label{eq:s4-two-seed-unnormalized-transport}
\end{equation}
The power of $\rho_e$ is two from $A_O^2$ and one from $S_{1,O}$,
or two from the substitution used for $S_{2,O}$ and
$S_{3,O}$.  The power of $\mu_e$ is one in each of the three homogeneous
copies.  Since $(2^h+1)e_j=s_j-1$, the $\mu_e$-factor is one.  Moreover,
Lemma~\ref{lem:s4-vertex-character} and
Eq.~\eqref{eq:s4-two-seed-normalizer-identity} give
\begin{equation}
 \lambda_{f_1,w}^{(a)}(Q)
 L_{a,j}(O)\rho_e^{3+\epsilon_a}
 =\lambda_{f_0,v}^{(a)}(Q).
 \label{eq:s4-lambda-gamma-cancellation}
\end{equation}
Equations~\eqref{eq:s4-two-seed-unnormalized-transport} and
\eqref{eq:s4-lambda-gamma-cancellation} prove
Eq.~\eqref{eq:s4-two-chart-cancellation}.

Every displayed relation is a polynomial identity, so it remains valid
when a Moore determinant vanishes.  No determinant is divided out.  The
normalizations are explicit functions of the vertex character and its
type set, and Eqs.~\eqref{eq:s4-scalar-seed-normalizer-identity} and
\eqref{eq:s4-two-seed-normalizer-identity} make their values independent
of the path used to compare two endpoint charts.
\end{proof}

Apply Eqs.~\eqref{eq:s4-actual-local-code}--%
\eqref{eq:s4-sheaf-stalk} to the six standard codes
$\operatorname{PRS}_{e_i}(\mathbb F_{s_i},\mathbb E)$, $i\in[6]$.
Denote the resulting sheaf by $\mathcal A$.  Define
\begin{equation}
 A(v)=\Psi_{v,v}^{(0)}(A_{O(v)}).
 \label{eq:s4-scalar-cocycle}
\end{equation}

For $a\in\{1,2,3\}$, apply the same construction to the six standard
codes
\begin{equation}
 \begin{cases}
 \operatorname{PRS}_{e_i}(\mathbb F_{s_i},\mathbb E),
     &i\in B_a,\\
 \operatorname{PRS}_{3e_i}(\mathbb F_{s_i},\mathbb E),
     &i\notin B_a,
 \end{cases}
 \label{eq:s4-moore-local-codes}
\end{equation}
and denote the resulting sheaf by $\mathcal G_a$.  For a face of type
$B_a$, let $v$ be its designated vertex with the two active type bits
zero.  Define
\begin{equation}
 \gamma_a(f)=
 \begin{cases}
 \Psi_{f,v}^{(a)}(M_f),&\type(f)=B_a,\\
 0,&\type(f)\ne B_a.
 \end{cases}
 \label{eq:s4-two-fixed-cocycles}
\end{equation}

The transport identities in Lemma~\ref{lem:s4-moore-chart-transport}
now give the cocycle conditions.

First, Eq.~\eqref{eq:s4-AO-space} and the chart isomorphism give
$A(v)\in\mathcal A_v$ at every vertex.  If $e:v\to w$ is an edge,
Eq.~\eqref{eq:s4-scalar-chart-cancellation} gives
\begin{equation}
 \begin{split}
 (\delta A)(e)
 &=
 \left.A(v)\right|_{X_{\geq e}(6)}
 +
 \left.A(w)\right|_{X_{\geq e}(6)}\\
 &=
 \left.\Psi_{v,v}^{(0)}(A_O)\right|_{X_{\geq e}(6)}
 +
 \left.\Psi_{w,w}^{(0)}(A_{O\cup\{j\}})\right|_{X_{\geq e}(6)}
 =0.
 \end{split}
 \label{eq:s4-scalar-coboundary}
\end{equation}
Thus $A\in Z^0(X,\mathcal A)$.

Next, Eq.~\eqref{eq:s4-Mf-space} gives
$\gamma_a(f)\in(\mathcal G_a)_f$ for every $B_a$-face $f$.  Let
$g$ be a three-cell.  If $B_a\not\subseteq\type(g)$, every term in
$(\delta\gamma_a)(g)$ is zero.  Otherwise
$\type(g)=B_a\cup\{j\}$ for a unique $j\notin B_a$.  Let
$f_0,f_1$ be the two $B_a$-facets of $g$.  By
Eq.~\eqref{eq:s4-two-chart-cancellation},
\begin{equation}
 \begin{split}
 (\delta\gamma_a)(g)
 &=
 \left.\gamma_a(f_0)\right|_{X_{\geq g}(6)}
 +
 \left.\gamma_a(f_1)\right|_{X_{\geq g}(6)}\\
 &=
 \left.\Psi_{f_0,v}^{(a)}(M_{f_0})\right|_{X_{\geq g}(6)}
 +
 \left.\Psi_{f_1,w}^{(a)}(M_{f_1})\right|_{X_{\geq g}(6)}
 =0.
 \end{split}
 \label{eq:s4-moore-coboundary}
\end{equation}
The last equality is the cancellation of the two opposite facets in
characteristic two.  Therefore
\begin{equation}
 A\in Z^0(X,\mathcal A),
 \qquad
 \gamma_a\in Z^2(X,\mathcal G_a)
 \quad(a=1,2,3).
 \label{eq:s4-fixed-cocycles-closed}
\end{equation}
This also makes clear why checks in the two active directions are absent:
the coboundary of a $B_a$-supported cochain compares only the two
opposite $B_a$-faces across a direction $j\notin B_a$.

The dual sheaf notation is the one used in the Poincar\'e-duality
framework of the paper.  If $\mathcal F$ is generated by the transported
local codes $C_\rho$, then
\begin{equation}
 (\mathcal F^\perp)_\sigma
 =
 \left\{
 z\in\mathbb E^{X_{\geq\sigma}(6)}:
 z|_{X_{\geq\rho}(6)}\in C_\rho^\perp
 \text{ for every }\rho\in X_{\geq\sigma}(5)
 \right\}.
 \label{eq:s4-dual-sheaf}
\end{equation}
This is the sheaf generated by the coordinate-dual local codes; it is not
the stalkwise vector-space dual.

For the fixed representatives $r_a=(a,1)^T$ and $r_\infty=(1,0)^T$,
the standard PRS codes satisfy
\begin{equation}
 \operatorname{PRS}_{u}(\mathbb F_{s_i},\mathbb E)^\perp
 =
 \operatorname{PRS}_{s_i-1-u}(\mathbb F_{s_i},\mathbb E).
 \label{eq:s4-prs-dual}
\end{equation}
Indeed, the projective sum of every homogeneous form of degree $s_i-1$
vanishes.  The transported dual code $C_\rho^\perp$ is obtained by applying
the inverse coordinate scalings and the same Frobenius permutation.

Let $[X]=\sum_{Q\in X(6)}Q$ be the fundamental top chain with its natural
coefficient in
$\mathcal F^\perp\otimes\mathcal F$.  The cap-product form of
Poincar\'e duality used here is
\begin{equation}
 H^0(X,\mathcal F^\perp)
 \xrightarrow{\ \sim\ }
 H_6(X,\mathcal F),
 \qquad
 [H]\longmapsto[H\frown[X]].
 \label{eq:s4-poincare}
\end{equation}
At the chain level, for $H\in Z^0(X,\mathcal F^\perp)$,
\begin{equation}
 \langle y,H\frown[X]\rangle
 =
 \langle H\smile y,[X]\rangle,
 \qquad
 y\in C^6(X,\mathcal F).
 \label{eq:s4-cap-adjoint}
\end{equation}
Proposition~\ref{claim:s4-poincare-duality} supplies the local
exactness required by Eq.~\eqref{eq:s4-poincare}.  The fact that the cap
product is a cycle can also be seen directly.

\begin{lemma}
\label{lem:s4-cap-cycle}
If $H\in Z^0(X,\mathcal F^\perp)$, then
$\partial(H\frown[X])=0$.
\end{lemma}

\begin{proof}
Let $y\in C^5(X,\mathcal F)$.  Using the adjointness of $\delta$ and
$\partial$, Eq.~\eqref{eq:s4-cap-adjoint}, the cup product Leibniz rule,
and $\delta H=0$, we obtain
\begin{equation}
 \begin{aligned}
 \left\langle y,\partial(H\frown[X])\right\rangle
 &=\left\langle\delta y,H\frown[X]\right\rangle\\
 &=\left\langle H\smile\delta y,[X]\right\rangle\\
 &=\left\langle\delta(H\smile y),[X]\right\rangle\\
 &=\sum_{\rho\in X(5)}
 \left\langle
 H|_{X_{\geq\rho}(6)},y(\rho)
 \right\rangle
 =0.
 \end{aligned}
 \label{eq:s4-cap-cycle-proof}
\end{equation}
Indeed, for every $\rho\in X(5)$ we have
$H|_{X_{\geq\rho}(6)}\in C_\rho^\perp$ and
$y(\rho)\in\mathcal F_\rho=C_\rho$, so each summand vanishes by the
coordinate pairing.  Since this holds for
every five-cochain $y$, the boundary itself is zero.
\end{proof}

\begin{proposition}[Sheaf Poincar\'e duality]
\label{claim:s4-poincare-duality}
For every sheaf used above, $\mathcal F^\perp$ is the sheaf generated by
the coordinate-dual local codes $C_\rho^\perp$ as in
Eq.~\eqref{eq:s4-dual-sheaf}.  Tensoring the six one-direction PRS
resolutions gives the isomorphism in Eq.~\eqref{eq:s4-poincare}, whose
chain representative is the cap product characterized by
Eq.~\eqref{eq:s4-cap-adjoint}.
\end{proposition}

\begin{proof}
For one direction, PRS evaluation and its coordinate dual form a perfect
pairing, and the two-term local cochain and chain complexes are adjoint.
The local complex above a cube is the tensor product of these six
one-direction complexes.  The K\"unneth formula over the field
$\mathbb E$ makes this tensor product exact below top degree and identifies
its top homology with the dual of its degree-zero cohomology.  The diagonal
coordinate pairing is preserved by the Frobenius permutations and inverse
coordinate scalings, so the dual local code after applying $\Phi$ is exactly
$C_\rho^\perp$.  Tensoring the six local adjunctions gives
Eq.~\eqref{eq:s4-cap-adjoint}, and hence the global isomorphism
Eq.~\eqref{eq:s4-poincare}.  The argument above uses its chain-level
consequence, which was verified directly in
Lemma~\ref{lem:s4-cap-cycle}.
\end{proof}

In each direction, one of the three sheaves $\mathcal G_a$ uses degree
$e_i$, while the other two use degree $3e_i$.  Hence
\begin{equation}
 e_i+3e_i+3e_i+(2^h-6)e_i=s_i-1.
 \label{eq:s4-fixed-degree-completion}
\end{equation}
We verify the dual condition in the transported coordinates.  Fix a
codimension-one face $\rho$ missing direction $i$ and use its designated
chart.  The restriction of a word in
$\mathcal G_1\otimes\mathcal G_2\otimes\mathcal G_3$ is, up to a nonzero
factor independent of the varying projective label, a standard PRS word
of degree $7e_i$.  The restriction of $A^{2^h-6}$ has degree
$(2^h-6)e_i$.  In their coordinatewise product, the powers of the
varying chart scalar add to $s_i-1$ by
Eq.~\eqref{eq:s4-fixed-degree-completion}, and hence disappear because
$c_i(v,Q)^{s_i-1}=1$.  Frobenius merely permutes
$\mathbb P^1(\mathbb F_{s_i})$, while the concrete $\lambda$-factors are
constant on the top cubes above the fixed face.  The coordinate pairing
is therefore a nonzero constant times the projective sum of a homogeneous
form of degree $s_i-1$, which vanishes by
Eq.~\eqref{eq:s4-prs-dual}.  Thus, for every $\rho\in X(5)$,
\begin{equation}
 \left.A^{2^h-6}\right|_{X_{\geq\rho}(6)}
 \in
 \left(
 (\mathcal G_1\otimes\mathcal G_2\otimes\mathcal G_3)^\perp
 \right)_\rho.
 \label{eq:s4-moore-dual-local}
\end{equation}
Because $A$ is a zero-cocycle and multiplication is a sheaf morphism,
$A^{2^h-6}$ is a zero-cocycle of this dual sheaf.
Lemma~\ref{lem:s4-cap-cycle} then gives
\begin{equation}
 A^{2^h-6}\frown[X]
 \in Z_6(X,\mathcal G_1\otimes\mathcal G_2\otimes\mathcal G_3).
 \label{eq:s4-dual-six-cycle}
\end{equation}
Thus the three degree-two cochains defined by Moore determinants and their associated six-cycle
have both been verified at chain level.

\subsubsection{Nonvanishing of the seed pairing}
Because $B_1,B_2,B_3$ partition the six directions, exactly one term in the cubical
cup product formula contributes on every six-cube.  Move the three factors to
the common type-zero chart using
Eq.~\eqref{eq:s4-two-chart-cancellation}.  In that chart all the type
normalizers $\eta_{0,\varnothing}$ and $\eta_{a,\varnothing}$ equal one.
The four character powers add to
$3+4+4+(2^h-6)=2^h+5$, which is divisible by three because $h$ is even;
hence their product is one in $\mathbb F_4^\times$.  In each direction
the chart weights likewise multiply to $c_i(v,Q)^{s_i-1}=1$.  Therefore
all three type-zero components, and all lifts over each of them, have the
same remaining standard-coordinate integrand.  Projective summation in
every direction gives
\begin{equation}
 \begin{split}
 &\left\langle
 \gamma_1\smile\gamma_2\smile\gamma_3,
 A^{2^h-6}\frown[X]
 \right\rangle\\
 &\quad =
 \left[
 \prod_{i=1}^6X_i^{s_i-1}Y_i^{s_i-1}
 \right]
 A_\varnothing^{2^h}
 S_{1,\varnothing}\,
 S_{2,\varnothing}(X_1+Y_1,Y_1,\ldots,X_6+Y_6,Y_6)\\
 &\qquad\qquad\qquad\qquad\cdot
 S_{3,\varnothing}(X_1+Y_1,Y_1,\ldots,X_6+Y_6,Y_6).
 \end{split}
 \label{eq:s4-fixed-coefficient}
\end{equation}
The square brackets extract the coefficient of the displayed monomial.
The three fixed-quotient type-zero copies and the covering degree contribute the
scalar $3(R^2-1)=1$ in characteristic two.
\begin{proposition}[Nonvanishing of the seed pairing]
\label{prop:s4-fixed-period}
The six fixed places can be chosen at the fixed degrees above so that
\begin{equation}
 \left\langle
 \gamma_1\smile\gamma_2\smile\gamma_3,
 A^{2^h-6}\frown[X]
 \right\rangle\ne0.
 \label{eq:s4-fixed-period-nonzero}
\end{equation}
\end{proposition}

\begin{proof}
By Eq.~\eqref{eq:s4-fixed-coefficient}, it remains to prove that the displayed coefficient is nonzero.  Use the conjugates
$\theta_{i,b}=\theta_i^{2^b}$ defined in
Eq.~\eqref{eq:s4-half-root-polynomials}.  For $j\in\{1,2,3\}$ define
\begin{equation}
 a_j(T)=
 \prod_{\substack{i\in B_j\\b\in E_i^0}}
 (T-\theta_{i,b}),
 \qquad
 b_j(T)=
 \prod_{\substack{i\in B_j\\b\in E_i^1}}
 (T-\theta_{i,b}).
 \label{eq:s4-fixed-half-polynomials}
\end{equation}
The two polynomials are the complementary Frobenius halves of
$p_{2j-1}p_{2j}$.  Put
$n_\Delta=\frac12\sum_{i=1}^6\mu_i$.  Each of the four Moore determinants in the product in Eq.~\eqref{eq:s4-fixed-coefficient} has
$n_\Delta$ columns, with $n_\Delta/2$ even rows and $n_\Delta/2$ odd
rows.  Their root polynomials and
the parity of the first row are
\begin{equation}
 \begin{array}{c|c|c}
 v&q_v(T)&\varepsilon_v\\ \hline
 0&b_1b_2b_3&0\\
 1&b_1a_2a_3&0\\
 2&a_1b_2a_3&1\\
 3&a_1a_2b_3&1.
 \end{array}
 \label{eq:s4-four-root-lists}
\end{equation}
The four row parities are therefore $(0,0,1,1)$.  This follows directly from
$h\in4\mathbb Z$ and from raising the complete coefficients of
$A_\varnothing$ to the $2^h$-th power.

We now evaluate that coefficient.  Let $\mathcal I_v$ be the root set of
$q_v$, put $\Pi=p_1\cdots p_6$, and let $\mathcal I_{\rm all}$ be the
root set of $\Pi$.  For a finite root set $I$, put
\begin{equation}
 \operatorname{Van}(I)=
 \prod_{\substack{x,y\in I\\x<y}}(y-x),
 \qquad
 Z(I)=\prod_{x\in I}\sqrt{x}.
 \label{eq:s4-vandermonde-products}
\end{equation}
A fixed ordering is used in the first product; signs disappear in
characteristic two.  The change from $(X,Y)$ to
$(X+\omega Y,X+\omega^2Y)$ has determinant one and therefore preserves
the alternating pairing.  Separating the even and odd rows of a Moore
determinant shows that the coefficient obtained by choosing a set
$S\subseteq\mathcal I_v$ of the $n_\Delta/2$ even-row columns is
\begin{equation}
 \operatorname{Van}(S)\operatorname{Van}(\mathcal I_v\setminus S)
 Z(S)^{\varepsilon_v}Z(\mathcal I_v\setminus S)^{1-\varepsilon_v}.
 \label{eq:s4-slater-coefficient}
\end{equation}
Every fixed root occurs in exactly two of the four determinants.  If
$t_b\in\{0,1,2\}$ counts the number of times its column contributes an
$X$-power, the middle coefficient condition is
\begin{equation}
 \sum_b t_b2^b=2^{\mu_i}-1.
 \label{eq:s4-no-carry}
\end{equation}
Reduction modulo two, followed successively by subtraction and division by
two, gives $t_b=1$ for every $b$.  Hence the coefficient is a contraction
along the six paired root lists.

Set
$c_v=\operatorname{Van}(\mathcal I_v)Z(\mathcal I_v)^{1-\varepsilon_v}$.
Let $\mathsf C_\Delta$ be the matrix whose row blocks are indexed by
$(v,r)$, $v\in\{0,1,2,3\}$ and $0\le r<n_\Delta/2$, and whose columns
are all $2n_\Delta$ fixed roots.  In a column $x\in\mathcal I_v$, its
nonzero entry in row $(v,r)$ is
$x^{2r}\sqrt{x}^{\,2\varepsilon_v-1}/q_v'(x)$.  Indeed,
$\prod_{x\in S}q_v'(x)$ is the product of
$\operatorname{Van}(S)^2$ with the factors joining $S$ to its complement;
the complementary product supplies the missing factors in
$\operatorname{Van}(\mathcal I_v)$.  Hence the coefficient in
Eq.~\eqref{eq:s4-slater-coefficient} equals $c_v$ times the corresponding
minor of the $v$-th row block.  Expanding by all four row blocks shows
that the required middle coefficient is
\begin{equation}
 \left(\prod_{v=0}^3c_v\right)\det\mathsf C_\Delta.
 \label{eq:s4-slater-contraction}
\end{equation}
Put
\begin{equation}
 \begin{aligned}
 Q_0&=a_1a_2a_3,&
 Q_1&=a_1b_2b_3,\\
 Q_2&=b_1a_2b_3,&
 Q_3&=b_1b_2a_3,
 \end{aligned}
 \label{eq:s4-fixed-Q-polynomials}
\end{equation}
and let $\mathsf H_\Delta$ be the coefficient matrix, in the basis
$1,T,\ldots,T^{2n_\Delta-1}$, of the $2n_\Delta$ polynomials
\begin{equation}
 T^{2r}Q_0,\quad T^{2r}Q_1,\quad
 T^{2r+1}Q_2,\quad T^{2r+1}Q_3,
 \qquad 0\le r<n_\Delta/2.
 \label{eq:s4-fixed-H}
\end{equation}
Multiplying the column of $\mathsf C_\Delta$ indexed by a root $x$ by
$\sqrt{x}\Pi'(x)$ changes its incident entries to
$x^{2r+\varepsilon_v}Q_v(x)$.  At the other row blocks those entries
vanish.  Hence
\begin{equation}
 \mathsf C_\Delta\operatorname{diag}_{x}\bigl(\sqrt{x}\Pi'(x)\bigr)
 =\mathsf H_\Delta\mathsf V,
 \label{eq:s4-contraction-factorization}
\end{equation}
where $\mathsf V$ is the full Vandermonde matrix on the
$2n_\Delta$ roots.  Taking
determinants gives the exact formula
\begin{equation}
 \left\langle
 \gamma_1\smile\gamma_2\smile\gamma_3,
 A^{2^h-6}\frown[X]
 \right\rangle
 =
 \frac{
 \prod_{v=0}^3\operatorname{Van}(\mathcal I_v)
 Z(\mathcal I_0)Z(\mathcal I_1)}
 {Z(\mathcal I_{\rm all})
 \operatorname{Van}(\mathcal I_{\rm all})}
 \det\mathsf H_\Delta.
 \label{eq:s4-fixed-period-detH}
\end{equation}
All factors outside $\det\mathsf H_\Delta$ are nonzero because the six polynomials
$p_i$ are distinct and separable.

It remains to prove $\det\mathsf H_\Delta\ne0$.  Write uniquely
\begin{equation}
 p_{2j-1}p_{2j}=\mathcal A_j^2+T\mathcal B_j^2,\\
 C_{ij}=\mathcal A_i\mathcal A_j+T\mathcal B_i\mathcal B_j,
 \qquad
 D_{ij}=\mathcal A_i\mathcal B_j+\mathcal B_i\mathcal A_j.
 \label{eq:s4-fixed-parity-polynomials}
\end{equation}
We choose the six fixed primes so that $\gcd(C_{12},C_{13})=1$.
A vector in the left kernel of $\mathsf H_\Delta$, after coefficientwise square roots,
gives polynomials $U_v$ of degree less than $n_\Delta/2$ satisfying
\begin{equation}
 Q_0U_0^2+Q_1U_1^2+TQ_2U_2^2+TQ_3U_3^2=0.
 \label{eq:s4-fixed-kernel-equation}
\end{equation}
Multiply Eq.~\eqref{eq:s4-fixed-kernel-equation} by $Q_0$, put
$V_0=Q_0U_0$ and $V_i=a_iU_i$ for $i=1,2,3$, and separate the even
and odd powers of $T$.  This gives
\begin{equation}
 \begin{aligned}
 V_0+C_{23}V_1+TD_{13}V_2+TD_{12}V_3&=0,\\
 D_{23}V_1+C_{13}V_2+C_{12}V_3&=0,
 \end{aligned}
 \label{eq:s4-even-odd-kernel}
\end{equation}
In $\mathbb E(T)$ define
\begin{equation}
 W=\frac{\mathcal A_2V_1+T\mathcal B_1V_2}{C_{12}},
 \qquad
 Z=\frac{\mathcal B_2V_1+\mathcal A_1V_2}{C_{12}}.
 \label{eq:s4-WZ-definition}
\end{equation}
The second line of Eq.~\eqref{eq:s4-even-odd-kernel} gives
\begin{equation}
 \begin{pmatrix}V_1\\V_2\\V_3\end{pmatrix}
 =
 \begin{pmatrix}
 \mathcal A_1&T\mathcal B_1\\
 \mathcal B_2&\mathcal A_2\\
 \mathcal B_3&\mathcal A_3
 \end{pmatrix}
 \begin{pmatrix}W\\Z\end{pmatrix}.
 \label{eq:s4-WZ-recovery}
\end{equation}
Solving instead from the first and third rows gives denominators dividing
$C_{13}$.  This coprimality therefore makes $W$ and $Z$ polynomials.  Since
$\deg\mathcal A_i=\deg a_i$, $\deg\mathcal B_i\le\deg a_i-1$, and
$\deg C_{12}=\deg a_1+\deg a_2$, the degree bounds in
Eq.~\eqref{eq:s4-fixed-kernel-equation} give
$\deg W,\deg Z<n_\Delta/2$.  Let $F=W^2+TZ^2$.  Squaring
Eq.~\eqref{eq:s4-WZ-recovery} and reducing modulo the corresponding
$a_i$ gives
\begin{equation}
 V_1^2\equiv T\mathcal B_1^2F\pmod{a_1},
 \qquad
 V_i^2\equiv\mathcal B_i^2F\pmod{a_i}
 \quad(i=2,3).
 \label{eq:s4-parity-kernel-congruences}
\end{equation}
Now $a_i\mid V_i$.  Since $p_{2i-1}p_{2i}$ is squarefree and its
derivative is $\mathcal B_i^2$, the polynomial $\mathcal B_i$ is a unit
modulo $a_i$; moreover $T$ is a unit modulo $a_1$ because
$(p_1p_2)(0)=1$.  Thus every $a_i$ divides $F$, and
$Q_0=a_1a_2a_3$ divides $F$.  But
$\deg F\le n_\Delta-1<\deg Q_0$.
Therefore $F=0$.  Its even and odd parts give $W=Z=0$, then
$V_1=V_2=V_3=0$, and finally $V_0=0$.  Thus $\mathsf H_\Delta$ has trivial kernel.

We finish by proving that such primes exist at the prescribed
degrees.  For a binary
polynomial $F$, write uniquely
\begin{equation}
 F=A_F^2+TB_F^2,
 \qquad
 C(F,G)=A_FA_G+TB_FB_G.
 \label{eq:s4-parity-decomposition}
\end{equation}
Let
\begin{equation}
 \mu_{\min}=\min_i\mu_i,
 \qquad
 r_{\mathrm{sieve}}=\left\lfloor\frac{\mu_{\min}}{16}\right\rfloor.
 \label{eq:s4-sieve-cutoff}
\end{equation}
Choose the six monic irreducible $p_i$ independently and uniformly in
their prescribed degree classes, and put $P_j=p_{2j-1}p_{2j}$.  If
$g\ne T$ is irreducible of degree $d\le r_{\mathrm{sieve}}$, reduction
modulo $g^2$ identifies a polynomial with its even--odd parity pair.
For three independent uniform unit residues $P_1,P_2,P_3$, direct
counting in $\mathbb F_2[T]/(g^2)$ gives
\begin{equation}
 \Pr\bigl[g\mid C(P_1,P_2),\ g\mid C(P_1,P_3)\bigr]=2^{-2d}.
 \label{eq:s4-local-common-factor-probability}
\end{equation}
Let $\mu_n$ be the residue distribution modulo $g^2$ of degree-$n$
primes, let $U_g$ be uniform on its unit group, and write
$\operatorname{TV}$ for total variation distance. Character
orthogonality modulo $g^2$, together with the prime-polynomial character
estimate of Ref.~\cite[Section~3.1.2 and Lemma~3.1]{Gorodetsky2020}, gives,
for $n>d$,
\begin{equation}
\operatorname{TV}(\mu_n,U_g)
\le(2d+1)2^{2d-n/2}.
\label{eq:s4-prime-residue-variation}
\end{equation}
Indeed, the nontrivial character sums are at most
$(2d+1)2^{n/2}/n$, while $\phi(g^2)<2^{2d}$ and
$|\mathcal P_n|\ge 2^{n-1}/n$.
Replacing the six residue laws one at a time bounds their joint
total-variation distance from six independent uniform units by
$6(2d+1)2^{2d-\mu_{\min}/2}$; multiplication into the three pairs cannot
increase this distance.  There are at most $2^d/d$ irreducibles of degree
$d$.  Summing the uniform term and the error term for
$d\le r_{\mathrm{sieve}}$ gives
\begin{equation}
 \Pr[\text{a common factor of degree at most }r_{\mathrm{sieve}}]
 \le
 \log 2+\frac{144}{7}2^{-5\mu_{\min}/16}.
 \label{eq:s4-small-factor-bound}
\end{equation}
The excluded polynomial $g=T$ is not a common factor because
$C_{12}(0)=C_{13}(0)=1$.

For large factors, first note that parity multiplication gives
$C_{12}^2=P_1P_2+T(P_1P_2)'$.  Since $P_1P_2$ is squarefree with
constant term one, $\gcd(C_{12},P_1P_2)=1$.  Condition on
$p_1,\ldots,p_5$ and let $g\mid C_{12}$ have
$d>r_{\mathrm{sieve}}$.  The parity pair of $P_1p_5$ is nonzero modulo
$g$, so the condition $g\mid C(P_1,p_5p_6)$ is a nonzero affine linear
condition on the parity coefficients of $p_6$, of rank at least
$\min\{d,\lfloor\mu_6/2\rfloor\}$.  The number of degree-$\mu_6$
irreducibles is at least $2^{\mu_6-1}/\mu_6$, so its conditional
probability is at most
$2\mu_6\,2^{-\min\{d,\lfloor\mu_6/2\rfloor\}}$.  Since
$\deg C_{12}\le4\mu_{\min}$, it has fewer than $64$ irreducible factors
of degree greater than $r_{\mathrm{sieve}}$.  A final union bound gives
\begin{equation}
 \Pr[\text{a common factor of degree greater than }r_{\mathrm{sieve}}]
 \le128\mu_{\min}2^{-r_{\mathrm{sieve}}}.
 \label{eq:s4-large-factor-bound}
\end{equation}
Combining Eqs.~\eqref{eq:s4-small-factor-bound} and
\eqref{eq:s4-large-factor-bound},
\begin{equation}
 \Pr[\gcd(C_{12},C_{13})\ne1]
 \le
 \log 2+\frac{144}{7}2^{-5\mu_{\min}/16}
 +128\mu_{\min}2^{-r_{\mathrm{sieve}}}<1.
 \label{eq:s4-fixed-sieve-bound}
\end{equation}
If the six distinct multiples of $2h$ are chosen with
$\mu_{\min}$ sufficiently large and
$\max_i\mu_i<2\mu_{\min}$, the final bound is strictly smaller than one.
The degree choice verified in
Subsection~\ref{subsec:s4-code-parameters} has these properties, so a
successful six-tuple exists.  Together with
Eq.~\eqref{eq:s4-fixed-period-detH}, this proves
Eq.~\eqref{eq:s4-fixed-period-nonzero}.
\end{proof}

Let
\begin{equation}
 \beta=
 \left\langle
 \gamma_1\smile\gamma_2\smile\gamma_3,
 A^{2^h-6}\frown[X]
 \right\rangle
 \in\mathbb E^\times,
 \qquad
 \xi_{\mathrm{seed}}=\beta^{-1}A^{2^h-6}\frown[X].
 \label{eq:s4-normalized-seed-pairing}
\end{equation}
Then
\begin{equation}
 \left\langle
 \gamma_1\smile\gamma_2\smile\gamma_3,
 \xi_{\mathrm{seed}}
 \right\rangle=1.
 \label{eq:s4-seed-pairing-one}
\end{equation}
In particular, none of the three classes $[\gamma_a]$ is zero.

\subsection{Linear amplification of the seed pairing}
\label{subsec:s4-multiplication}

The nonzero pairing constructed in the preceding subsection is the seed
for the amplification argument.  We multiply each of the three seed
cocycles by global zero-cocycles of a fixed auxiliary sheaf and construct
dual zero-cocycles that distinguish the resulting products.

Choose constants $d_*$ and $s_*$, independent of $R$, as specified in
the proof below, put $d_{\mathrm{prod}}=d_*s_*$, and require
$\ell\ge d_{\mathrm{prod}}$.  Let $\mathcal P$ be the sheaf obtained
from the standard codes
\begin{equation}
 \operatorname{PRS}_{d_{\mathrm{prod}}e_i}(\mathbb F_{s_i},\mathbb E)
 \qquad(i\in[6])
 \label{eq:s4-zero-cocycle-degrees}
\end{equation}
through the charts of Eq.~\eqref{eq:s4-chart-formula}.  For
$P\in Z^0(X,\mathcal P)$ and $a\in\{1,2,3\}$, define
\begin{equation}
 u_a(P)
 =A^{\ell-d_{\mathrm{prod}}}P\gamma_a
 \in Z^2(X,\mathcal F_a).
 \label{eq:s4-multiplied-cocycles}
\end{equation}
The membership follows from the degree profiles in
Eq.~\eqref{eq:s4-code-local-codes}, the sheaf multiplication of
Eq.~\eqref{eq:s4-sheaf-multiplication}, and the Leibniz rule.  The three
local degrees add to $(3\ell+7)e_i$, while
\begin{equation}
 (3\ell+7)e_i+(2^{h-1}-4)e_i=s_i-1.
 \label{eq:s4-dual-completion}
\end{equation}
Thus degree $(2^{h-1}-4)e_i$ is coordinate-dual to the product of the
three code profiles.

\begin{theorem}[Linear amplification]
\label{thm:s4-diagonal-pairing}
For every sufficiently large $R=2^r$ with $r$ even, there are an integer
$m$, zero-cocycles
\begin{equation}
 P_1,\ldots,P_m\in Z^0(X,\mathcal P),
 \qquad
 H_1,\ldots,H_m\in
 Z^0\!\left(X,
 (\mathcal F_1\otimes\mathcal F_2\otimes\mathcal F_3)^\perp
 \right),
 \label{eq:s4-PH-membership}
\end{equation}
and a constant $c_{\mathrm{amp}}>0$, independent of $R$, such that
\begin{equation}
 m\ge c_{\mathrm{amp}}R^2
 \label{eq:s4-linear-count}
\end{equation}
and
\begin{equation}
 \left\langle
 u_1(P_i)\smile u_2(P_j)\smile u_3(P_k),
 H_t\frown[X]
 \right\rangle
 =\delta_{it}\delta_{jt}\delta_{kt}
 \label{eq:s4-selector-identity}
\end{equation}
for all $i,j,k,t\in[m]$.
\end{theorem}

The proof occupies the next subsubsection.  Its essential point is that
the covering-space symmetry turns the one nonzero seed pairing into a
multiplication trace on a finite algebra of dimension $\Omega(R^2)$;
standard trace duality then produces the independent pairings in
Eq.~\eqref{eq:s4-selector-identity}.

\subsubsection{From the seed pairing to linearly many independent pairings}
\label{subsec:s4-zero-cocycle-proof}

We now prove Theorem~\ref{thm:s4-diagonal-pairing}.  All objects entering
the theorem are genuine cochains on $X$.  The finite algebras introduced
inside the proof are used only to choose linear combinations of these
cochains.

\paragraph{Global zero-cocycles from covering symmetry.}
The first point is to distinguish global zero-cocycles from their rational
coordinate values.  Let $r$ be even and let $R=2^r$.  Write
$\zeta\in\mathbb F_R$ for an element of degree $r$ over $\mathbb F_2$;
its final choice among the exact-degree elements is made in the counting
argument below.  Put $\tau_j=\zeta^{2^{j+1}}$ for
$j\in\mathbb Z/r\mathbb Z$.  These are the square-root nodes in the
split Moore columns; the corresponding arguments of the fixed-place
polynomials are $\tau_j^2$.  If
$T\subseteq\mathbb Z/r\mathbb Z$ is their support, define
$L_j(T)=\prod_{t\in T}p_{j,1}(\tau_t^2)/\sqrt{p_j(\tau_t^2)}$.
The chosen root avoids the finitely many zeros in the denominators.
For the two determinant variants $\mathrm D$ and $\mathrm C$, define the
exponent $a(c,T)$ by
\begin{equation}
 \begin{array}{c|cc}
 &|T|\text{ even}&|T|\text{ odd}\\ \hline
 a(\mathrm D,T)&1&2\\
 a(\mathrm C,T)&2&1.
 \end{array}
 \label{eq:s4-source-charge-table}
\end{equation}
For $O\subseteq[6]$, set
\begin{equation}
 \eta_{O,T}
 =\eta_{0,O}\prod_{j\in O}L_j(T)^{-1},
 \qquad
 \lambda_{f,v}^{c,T}(Q)
 =\chi(v)^{a(c,T)}\eta_{O(v),T}
 \qquad(v\leq f,\ Q\geq f,\ c\in\{\mathrm D,\mathrm C\}).
 \label{eq:s4-source-normalization}
\end{equation}
For a positive direction-$j$ edge $e:v\to w$ contained in $f$, with
$Q\geq f$ and type set changing from $O$ to $O\cup\{j\}$, this choice
satisfies
\begin{equation}
 \frac{\lambda_{f,w}^{c,T}(Q)}{\lambda_{f,v}^{c,T}(Q)}
 =\bigl(K_{0,j}(O)L_j(T)\rho_e^{a(c,T)}\bigr)^{-1}.
 \label{eq:s4-source-lambda-edge-ratio}
\end{equation}
Restriction compatibility and square compatibility again follow because
Eq.~\eqref{eq:s4-source-normalization} defines vertex scalars.  Thus every
determinant with variant $c$ and auxiliary index set $T$ uses the
normalized chart
\begin{equation}
 [\Psi_{f,v}^{c,T}(M)](Q)
 =\lambda_{f,v}^{c,T}(Q)[\Phi_{f,v}(M)](Q).
 \label{eq:s4-source-normalized-chart}
\end{equation}
All subsequent contractions and products are taken after this explicit
normalization.

\begin{lemma}
\label{lem:s4-primitive-moore-sections}
Let $\widehat X_R$ be the principal cover defined by the kernel of the
level map at the varying place.  Its deck group is the group $G_R$
introduced above, and $X_R=U_R\backslash\widehat X_R$.  Work
temporarily over
$\mathbb L_R=\mathbb E\mathbb F_R$.  If $g_v\in G_R$ is the level
coordinate of a vertex of $\widehat X_R$ and
$\boldsymbol u=(u_j)_{j\in T}$ with $u_j\in\mathbb L_R^2$, put
$g_v^{-1[2^j]}u_j=(x_j(v),y_j(v))^T$.

For $O=O(v)$, append to the fixed-place columns in
Eq.~\eqref{eq:s4-AO} the auxiliary columns
\begin{equation}
 \left(
 \tau_j, x_j(v)+\omega y_j(v), x_j(v)+\omega^2y_j(v)
 \right)_{j\in T}.
 \label{eq:s4-auxiliary-moore-columns}
\end{equation}
The resulting Moore determinant is denoted by
$D_{O,T,\boldsymbol u,v}$.  Let $C_{O,T,\boldsymbol u,v}$ be obtained
by applying $(x,y)\mapsto(x+y,y)$ simultaneously to every fixed-place and auxiliary
pair.  Put
$P_{\mathrm{ref}}(z)=\prod_{i=1}^6\prod_{b\in E_i^0}
(z-\zeta_{i,b})$ and
$\Delta_T=\prod_{t<t',\,t,t'\in T}(\tau_t+\tau_{t'})$.
The factors $\Delta_T$ and $P_{\mathrm{ref}}(\tau_t)$ are nonzero fixed
scalars.  With the normalizations in
Eq.~\eqref{eq:s4-source-normalization}, define
\begin{equation}
 \begin{aligned}
 \widehat D_{T,\boldsymbol u}(v)
 &=
 \frac{\Psi_{v,v}^{\mathrm D,T}
     (D_{O(v),T,\boldsymbol u,v})}
 {\Delta_T\prod_{t\in T}P_{\mathrm{ref}}(\tau_t)},\\
 \widehat C_{T,\boldsymbol u}(v)
 &=
 \frac{\Psi_{v,v}^{\mathrm C,T}
     (C_{O(v),T,\boldsymbol u,v})}
 {\Delta_T\prod_{t\in T}P_{\mathrm{ref}}(\tau_t)} .
 \end{aligned}
 \label{eq:s4-primitive-moore-sections}
\end{equation}
These assignments are zero-cocycles of degree $e_i$ in each of the six
fixed directions.  Empty products are one, so
$A=\widehat D_{\varnothing}$ and we write
$C=\widehat C_{\varnothing}$ for the two empty primitives.

The same definition and endpoint calculation apply to a finite union of
pairwise disjoint Frobenius orbits, with $L_j(T)$ interpreted as the
product over the squared coordinates of all selected square-root nodes.
For a deck element $h\in G_R$, use the pullback convention
$(h^*s)(v,Q)=s(h^{-1}v,h^{-1}Q)$.  The level coordinates are chosen so
that $g_{h^{-1}v}^{-1}=g_v^{-1}h$.  The primitive sections are
multilinear in their parameters and satisfy
\begin{equation}
 h^*\widehat s_{c,T,\boldsymbol u}
 =
 \widehat s_{c,T,(h^{[2^j]}u_j)_{j\in T}}
 \qquad
 (h\in G_R,\ c\in\{\mathrm D,\mathrm C\}),
 \label{eq:s4-primitive-deck-covariance}
\end{equation}
where $\widehat s_{\mathrm D,T}=\widehat D_T$ and
$\widehat s_{\mathrm C,T}=\widehat C_T$.  Hence every
$U_R$-invariant contraction of products of these parameter maps descends
to a global zero-cocycle on $X_R$.
\end{lemma}

\begin{proof}
Let $e:v\to w$ be a positive edge in direction $j$, and write
$O=O(v)$.  The split polynomial matrix of the arithmetic neighbor has
determinant $p_j$.  At the roots in $E_j^0$ its kernel columns are the
corresponding Moore columns, while at the roots in $E_j^1$ its image
columns play the same role at the other endpoint.  Polynomial row
interpolation sends a row $r$ that annihilates the first family to
$r\operatorname{adj}(G_e)/p_{j,0}$, which annihilates the second.
On an auxiliary column with square-root node $z$, the
determinant-normalized transformation contributes
$p_{j,1}(z^2)/\sqrt{p_j(z^2)}$.  Comparing the two row determinants gives
the fixed factor $\kappa_j^\Delta$ from direction $j$; the fixed transverse columns
give the factors $R_{ji}^{\varepsilon}$ in
Eq.~\eqref{eq:s4-fixed-moore-transport-scalars}.  Since this is an
identity between polynomial determinants, auxiliary columns may repeat
and either determinant may vanish.

After inserting the determinant-normalized transverse matrices
and the standard chart weight, the resulting edge identity is
\begin{equation}
 \begin{aligned}
 [\Phi_{w,w}(D_{O\cup\{j\},T,\boldsymbol u,w})](Q)
 &=
 K_{0,j}(O)L_j(T)\rho_e^{a(\mathrm D,T)}
 [\Phi_{v,v}(D_{O,T,\boldsymbol u,v})](Q),\\
 [\Phi_{w,w}(C_{O\cup\{j\},T,\boldsymbol u,w})](Q)
 &=
 K_{0,j}(O)L_j(T)\rho_e^{a(\mathrm C,T)}
 [\Phi_{v,v}(C_{O,T,\boldsymbol u,v})](Q).
 \end{aligned}
 \label{eq:s4-primitive-raw-endpoint-law}
\end{equation}
The remaining scalar from direction $j$ is
$\mu_e^{(2^h+1)e_j}=1$.  Equation~\eqref{eq:s4-source-lambda-edge-ratio}
is precisely the inverse of the remaining multiplier in
Eq.~\eqref{eq:s4-primitive-raw-endpoint-law}.  The restrictions of the
two normalized endpoint values are therefore equal on every containing
top cube, which proves that the cochains in
Eq.~\eqref{eq:s4-primitive-moore-sections} are closed.

With the pullback convention above,
$g_{h^{-1}v}^{-1[2^j]}u_j
 =g_v^{-1[2^j]}h^{[2^j]}u_j$; the fixed-place frames and normalization are
unchanged.  Multilinearity of the determinant
gives Eq.~\eqref{eq:s4-primitive-deck-covariance}.  An invariant
parameter tensor therefore gives an invariant closed cochain and
descends through the indicated quotient.
\end{proof}

\begin{proposition}
\label{prop:s4-global-minor-sections}
The fixed auxiliary Frobenius orbit gives a finite family of charts,
indexed by $\kappa$ and independent of $r$.  For each $\kappa$
there is a zero-cocycle $K_\kappa$ of degree $3e_i$ in direction $i$
and, for $j\in\mathbb Z/r\mathbb Z$, zero-cocycles
$f_{\kappa,j},t_{\kappa,j}$ of degrees $e_i$ and $3e_i$ in direction $i$.  For
equal-cardinality subsets $I,J\subseteq\mathbb Z/r\mathbb Z$, there is
also a zero-cocycle $D_{\kappa;I,J}$ of degree $8e_i$ in direction $i$.  Fix one
$\kappa$ temporarily and suppress it from the notation.  On the open
part of its chart where $AK\ne0$, the values have the form
\begin{equation}
 f_j=A U_j,
 \qquad
 t_j=K V_j,
 \qquad
 D_{I,J}=(AK)^2\det\mathsf K[I,J].
 \label{eq:s4-actual-minor-functions}
\end{equation}
Here $w_j=(U_j,V_j)^T$, and
\begin{equation}
 \begin{aligned}
 \mathsf K_{jj}&=U_jV_j,\\
 \mathsf K_{ij}
 &=\frac{\sqrt{\tau_i\tau_j}}{\tau_i+\tau_j}
   (U_iV_j+V_iU_j),
 \qquad i\ne j.
 \end{aligned}
 \label{eq:s4-paired-matrix}
\end{equation}
The right-hand sides of Eq.~\eqref{eq:s4-actual-minor-functions} are
coordinate descriptions; the objects on the left are defined globally
and remain defined at the zeros of $A$ or $K$.
\end{proposition}

\begin{proof}
Choose once a fixed auxiliary binary prime of degree $2h$, disjoint
from the six fixed places and from the varying place, and use its complete
square-root Frobenius orbit.  Since $h\ge b_0+2$, this orbit contains all
the fixed interpolation nodes required later.  Choose adjacent nonzero
nodes $\upsilon,\upsilon^2$ in it.  For every auxiliary square-root node
$z$, let
$f_z(u)=\widehat D_{\{z\},u}$ be the singleton covariant row, and use
$\langle p,q\rangle=p_1q_2+p_2q_1$ for the alternating pairing.  Define
\begin{equation}
 K=\langle f_\upsilon^{[2]},f_{\upsilon^2}\rangle
 \label{eq:s4-DD-anchor}
\end{equation}
which is a zero-cocycle of degree $3e_i$ in each fixed direction.
Here $[2]$ squares both coefficients of the row map.  Define the second
covariant row by
\begin{equation}
 t_z(u)=
 \begin{cases}
 A f_\upsilon^{[2]}(u),&z=\upsilon^2,\\[1mm]
 \displaystyle
 \frac{\upsilon^2+z}{z}\,
 \operatorname{contr}_{\upsilon^2}
 \bigl(\widehat D_{\{\upsilon^2,z\}}(\,\cdot,u),
       f_\upsilon^{[2]}\bigr),
 &z\ne\upsilon^2 .
 \end{cases}
 \label{eq:s4-covariant-second-row}
\end{equation}
The contraction removes the auxiliary representation and retains the
dependence on the parameter indexed by $z$.  The first line treats
$z=\upsilon^2$, where the determinant with two auxiliary columns is not
defined.  Hence every $t_z$ is a covariant row whose coefficients have
degree $3e_i$ in each fixed direction.

We now identify one common frame for both row families.  Enumerate the
$n_\Delta$ fixed-place square-root nodes of $P_{\mathrm{ref}}$ as
$\beta_1,\ldots,\beta_{n_\Delta}$, put $m_{\rm fix}=n_\Delta/2$, and let
$v_a=\mathsf S(\beta_a)p_a$ be the corresponding nonzero Moore
columns, where
$\mathsf S(z)=\left(\begin{smallmatrix}1&\omega\\
z&\omega^2z\end{smallmatrix}\right)$.  On $A\ne0$, there is a unique
polynomial matrix satisfying
\begin{equation}
 \mathsf M(t)=t^{m_{\rm fix}}I_2+\text{terms of degree less than }m_{\rm fix},
 \qquad
 \mathsf M(\beta_a^2)v_a=0
 \quad(1\le a\le n_\Delta).
 \label{eq:s4-old-row-interpolation}
\end{equation}
Indeed, each row is determined by the fixed-place Moore interpolation matrix,
whose determinant is $A$ in this chart.  If
$P_{\mathrm{fix}}(t)=\prod_a(t+\beta_a^2)$, then
$\det\mathsf M(t)$ is monic of degree $n_\Delta$ and has precisely
these roots.  Therefore
\begin{equation}
 \det\mathsf M(t)=P_{\mathrm{fix}}(t),
 \qquad
 P_{\mathrm{fix}}(z^2)=P_{\mathrm{ref}}(z)^2.
 \label{eq:s4-old-row-determinant}
\end{equation}
In particular, $\mathsf M(0)$ is invertible.  Put
$\mathsf G=\operatorname{diag}(\chi,\chi^2)$,
$\mathsf D(z)=\operatorname{diag}(1,z^{-1})$, and
\begin{equation}
 \mathcal Q_{\mathrm{fix}}(z)=
 \frac{\mathsf G\mathsf D(z)\mathsf M(z^2)\mathsf S(z)}
 {P_{\mathrm{ref}}(z)}.
 \label{eq:s4-old-row-frame}
\end{equation}
Equations~\eqref{eq:s4-old-row-determinant} and
$\det\mathsf S(z)=z$ give
$\det\mathcal Q_{\mathrm{fix}}(z)=1$.

Write
$(U_z(u),V_z(u))^T=\mathcal Q_{\mathrm{fix}}(z)u$.  On $AK\ne0$,
set
$F_0=\langle U_\upsilon^{[2]},U_{\upsilon^2}\rangle$ and
$F_1=\langle V_{\upsilon^2},U_\upsilon^{[2]}\rangle$.  The $2\times2$
Moore-determinant identity and Eq.~\eqref{eq:s4-covariant-second-row} give
$K=A^3F_0$ and
$t_z=A^3(F_0V_z+(\upsilon^2/z)F_1U_z)$, including
$z=\upsilon^2$ because the chosen frame has determinant one.  Thus the
same lower shear works at every node.  With
$\varkappa=\upsilon^2F_1/F_0$, define
\begin{equation}
 \mathcal Q(z)=
 \begin{pmatrix}1&0\\ \varkappa/z&1\end{pmatrix}
 \mathcal Q_{\mathrm{fix}}(z),
 \qquad
 \binom{U_z(u)}{V_z(u)}=\mathcal Q(z)u,
 \label{eq:s4-common-DD-frame}
\end{equation}
where the second coordinate has been renamed after the shear.  Then
$f_z(u)=AU_z(u)$ and $t_z(u)=KV_z(u)$ for every varying-place or auxiliary node.
For $j\in\mathbb Z/r\mathbb Z$, define the scalar zero-cocycles
$f_j=f_{\tau_j}(e_1)$ and $t_j=t_{\tau_j}(e_1)$, and put
$w_j=(U_j,V_j)^T=\mathcal Q(\tau_j)e_1$.  These objects descend because
$e_1$ is fixed by $U_R$.  The connection of this common frame is the
cyclic relation in Eq.~\eqref{eq:s4-cyclic-frobenius-system}.

We next construct the two determinant-row families globally.  For
$T\subseteq\mathbb Z/r\mathbb Z$, let
\begin{equation}
 \begin{aligned}
 d_T(w_T)
 &=\frac{1}{\Delta_T}
   \det\left(
   \tau_j^u
   \begin{cases}
    U_j,&u\ \mathrm{even},\\
    V_j,&u\ \mathrm{odd}
   \end{cases}
   \right)_{\substack{0\le u<|T|\\ j\in T}},\\
 c_T(w_T)&=d_T((V_j,U_j)_{j\in T}),
 \qquad d_\varnothing=c_\varnothing=1 .
 \end{aligned}
 \label{eq:s4-bare-row-determinants}
\end{equation}
These are ordinary coordinate polynomials; they do not yet define
sections away from the chart.

Choose a temporary binary prime of degree $r+1$, disjoint from all fixed
places and from the varying prime, and let $\mathcal Z$ be its complete
square-root orbit.  For $z\in\mathcal Z$, use the mixed-prime
version of Lemma~\ref{lem:s4-primitive-moore-sections} and set
$S_{T,z}(u_T)=\operatorname{contr}_{z}
(\widehat D_{T\cup\{z\}}(u_T,\,\cdot),t_z)$, where
$\operatorname{contr}_{z}$ is the alternating pairing in the auxiliary
parameter indexed by $z$.  The pairing removes the auxiliary deck action, so
$S_{T,z}$ is a $G_R$-equivariant zero-cocycle of degree $4e_i$ in direction $i$.  On
$AK\ne0$ it has value
$AK\,d_{T\cup\{z\}}(w_T;e_1)$.

Put $p_T(X)=\prod_{j\in T}(X+\tau_j)$.  Expansion in the last column
shows that
$H_T(X)=p_T(X)d_{T\cup\{X\}}(w_T;e_1)$ has degree at most
$2\lfloor|T|/2\rfloor\le |T|\le r$.  At $X=0$, only the first entry
of the last column survives; deleting that row and column and extracting
one factor $\tau_j$ from every varying-place column gives
$H_T(0)=p_T(0)c_T(w_T)$.  The $r+1$ nodes in
$\mathcal Z$ therefore determine $H_T(0)$ by Lagrange interpolation.
If
$g_{\mathrm{aux}}(X)=\prod_{z\in\mathcal Z}(X+z)$, the Lagrange weight
at zero is
$\ell_z=g_{\mathrm{aux}}(0)/(z\,g_{\mathrm{aux}}'(z))$.  Define
\begin{equation}
 D_T^{\mathrm{row}}(u_T)=K\widehat D_T(u_T),\qquad
 C_T^{\mathrm{row}}(u_T)
 =p_T(0)^{-1}\sum_{z\in\mathcal Z}
   \ell_z p_T(z)S_{T,z}(u_T).
 \label{eq:s4-global-row-sections}
\end{equation}
For the empty support set both sections equal $AK$.  Every denominator
in Eq.~\eqref{eq:s4-global-row-sections} is a fixed nonzero node scalar.
The two families are globally defined zero-cocycles of degree $4e_i$ in direction $i$.
At a fixed lifted vertex, let $\mathcal Q_j\in\operatorname{SL}_2$
denote the linear map from the constant parameter $u_j$ to the
transported spinor, including the level-coordinate trivialization.
The same maps occur for the two row families, and on $AK\ne0$,
\begin{equation}
 \begin{aligned}
 D_T^{\mathrm{row}}((u_j)_{j\in T})
 &=AK\,d_T((\mathcal Q_ju_j)_{j\in T}),\\
 C_T^{\mathrm{row}}((u_j)_{j\in T})
 &=AK\,c_T((\mathcal Q_ju_j)_{j\in T}).
 \end{aligned}
 \label{eq:s4-global-row-chart-values}
\end{equation}
In particular,
$w_j=\mathcal Q_je_1=(U_j,V_j)^T$ gives the scalar values in
Eq.~\eqref{eq:s4-actual-minor-functions}.
Each alternating contraction has already removed the auxiliary deck
action.  The complete auxiliary orbit is stable under its coefficient
Galois group, so reindexing the interpolation sum removes the auxiliary
coefficient field as well.  Zero is only the evaluation point of $H_T$;
no auxiliary determinant at zero and no division by the orbit size is used.

It remains to combine these rows without using chart denominators.
Let $\mathcal V=\bigoplus_{j\in\mathbb Z/r\mathbb Z}\mathbb L_R^2$,
with standard vectors $e_{1,j},e_{2,j}$ and the direct-sum alternating
form.  Put $p_\tau(z)=\prod_j(z+\tau_j)$,
$\mathsf V_\tau=(\tau_j^u)_{0\le u<r,\,j}$,
$\mathsf S_\tau=\operatorname{diag}_j
(p_\tau'(\tau_j)\sqrt{\tau_j})^{-1}$, and let $J_{\mathrm{odd}}$ be diagonal
with entry $u\bmod2$ in row $u$.  The matrices
\begin{equation}
 \Pi_1=\mathsf S_\tau^{-1}\mathsf V_\tau^{-1}
 J_{\mathrm{odd}}\mathsf V_\tau\mathsf S_\tau,
 \qquad
 \Pi_0=I+\Pi_1
 \label{eq:s4-moore-projectors}
\end{equation}
are symmetric complementary projectors.  Lagrange differentiation gives
\begin{equation}
 (\Pi_0)_{ij}=(\Pi_1)_{ij}
 =\frac{\sqrt{\tau_i\tau_j}}{\tau_i+\tau_j}
 \qquad(i\ne j).
 \label{eq:s4-moore-projector-entry}
\end{equation}

For a wedge-basis index $\Omega$ of $\bigwedge^r\mathcal V$, let
$\mathsf Z,\mathsf D,\mathsf T$ be the indices at which $\Omega$
contains, respectively, zero, two, or one of the vectors
$e_{1,j},e_{2,j}$.  Then $|\mathsf Z|=|\mathsf D|$.  If $u_j$ is the selected vector at an
index in $\mathsf T$, define two linear maps on the wedge basis by
\begin{equation}
 \begin{aligned}
 \mathcal E_D(e_\Omega)
 &=\det\left[
   \frac{\sqrt{\tau_i\tau_j}}{\tau_i+\tau_j}
   \right]_{\substack{i\in\mathsf Z\\j\in\mathsf D}}
   D_{\mathsf T}^{\mathrm{row}}((u_j)_{j\in\mathsf T}),\\
 \mathcal E_C(e_\Omega)
 &=\det\left[
   \frac{\sqrt{\tau_i\tau_j}}{\tau_i+\tau_j}
   \right]_{\substack{i\in\mathsf Z\\j\in\mathsf D}}
   C_{\mathsf T}^{\mathrm{row}}((u_j)_{j\in\mathsf T}),
 \end{aligned}
 \label{eq:s4-global-exterior-maps}
\end{equation}
with empty determinant equal to one.  At an index contributing one basis
vector, covariance is the one proved in Eq.~\eqref{eq:s4-primitive-deck-covariance};
at an index contributing both basis vectors it is trivial because
$\bigwedge^2h^{[2^j]}=\det(h)^{2^j}=1$.  Hence both maps are globally
$G_R$-equivariant.

We verify their chart values on every exterior coordinate.  Before the
fixed Vandermonde calibration, a basis wedge with double, missing, and
single sets $\mathsf D,\mathsf Z,\mathsf T$ contributes
\begin{equation}
 \left(\prod_{d\in\mathsf D}\tau_d\right)
 \Delta_{\mathsf D}^{4}
 \left(\prod_{\substack{t\in\mathsf T\\d\in\mathsf D}}
       (\tau_t+\tau_d)^2\right)
 \Delta_{\mathsf T}d_{\mathsf T},
 \label{eq:s4-exterior-calibration}
\end{equation}
and the second row family gives the same fixed factor with
$c_{\mathsf T}$ in place of $d_{\mathsf T}$.  The row and column
calibrations cancel all factors except
\begin{equation}
 \frac{
 (\prod_{i\in\mathsf Z}\sqrt{\tau_i})
 (\prod_{j\in\mathsf D}\sqrt{\tau_j})
 \Delta_{\mathsf Z}\Delta_{\mathsf D}}
 {\prod_{i\in\mathsf Z,j\in\mathsf D}(\tau_i+\tau_j)} .
 \label{eq:s4-exterior-cauchy-factor}
\end{equation}
which is exactly the Cauchy determinant appearing in
Eq.~\eqref{eq:s4-global-exterior-maps}.  Thus, if
$\mathcal Q_\oplus=\bigoplus_j\mathcal Q_j$ in the two-block basis
order and
$F_D=[\Pi_0,\Pi_1]\mathcal Q_\oplus$,
$F_C=[\Pi_1,\Pi_0]\mathcal Q_\oplus$, then on $AK\ne0$,
$\mathcal E_D=AK\,\operatorname{vol}(F_D)$ and
$\mathcal E_C=AK\,\operatorname{vol}(F_C)$, where $\operatorname{vol}(F)$ is the top
exterior form of the row frame $F$.

For a coordinate subset $B$ of the basis of $\mathcal V$, let
$B^\sharp$ exchange $e_{1,j}$ and $e_{2,j}$ for every index $j$, and write
$e_B$ for its ordered wedge.  If $I,J$ have the same cardinality $k$,
the tensor
\begin{equation}
 \Theta_{I,J}
 =\sum_{|B|=r-k}
 (e_{1,J}\wedge e_B)\otimes
 (e_{1,I}\wedge e_{B^\sharp})
 \label{eq:s4-invariant-exterior-tensor}
\end{equation}
is $U_R$-invariant: it is the inverse tensor of the induced alternating
pairing, wedged with the $U_R$-fixed vectors $e_{1,I}$ and $e_{1,J}$.
Apply $\mathcal E_D\otimes\mathcal E_C$ to this tensor and multiply
the two resulting sections.  For $|I|=|J|>0$, denote the resulting
zero-cocycle, of degree $8e_i$ in direction $i$, by $D_{I,J}$; prescribe
$D_{\varnothing,\varnothing}=(AK)^2$ independently.  Global
equivariance and the invariance of $\Theta_{I,J}$ show that every
$D_{I,J}$ descends to $X_R$.

We now prove its chart value, including the orientation of $I$ and $J$.
Let
$J_{\mathcal V}=\left(\begin{smallmatrix}0&I_r\\I_r&0\end{smallmatrix}\right)$,
and let $E_I,E_J$ select the first-coordinate columns indexed by $I,J$.
For arbitrary $r$-by-$2r$ matrices $F_1,F_2$, Cauchy--Binet gives
\begin{equation}
 (\operatorname{vol}(F_1)\otimes\operatorname{vol}(F_2))(\Theta_{I,J})
 =
 \det\begin{pmatrix}
 F_1J_{\mathcal V}F_2^T&F_1E_J\\
 E_I^TF_2^T&0
 \end{pmatrix}.
 \label{eq:s4-bordered-contraction}
\end{equation}
Indeed, factor the bordered matrix through
$\left(\begin{smallmatrix}F_1J_{\mathcal V}&F_1E_J\\
E_I^T&0\end{smallmatrix}\right)$ and
$\left(\begin{smallmatrix}F_2^T&0\\0&I\end{smallmatrix}\right)$.
A nonzero Cauchy--Binet term must contain the last $|I|$ indices and the
first-coordinate indices in $I$; writing the remaining indices as
$B^\sharp$ gives exactly the summands of $\Theta_{I,J}$.  Repeated
wedge vectors correspond to the vanishing minors, so the identity also
holds when $I$ and $J$ overlap.

Every $\mathcal Q_j$ has determinant one and therefore preserves the
direct-sum alternating form.  Since $\Pi_0,\Pi_1$ are symmetric complementary
projectors,
$F_DJ_{\mathcal V}F_C^T=\Pi_0^2+\Pi_1^2=I_r$.  With
$D_U=\operatorname{diag}(U_j)$ and
$D_V=\operatorname{diag}(V_j)$, the first-coordinate blocks of
$F_D,F_C$ are
$\mathsf U=\Pi_0D_U+\Pi_1D_V$ and
$\mathsf V=\Pi_1D_U+\Pi_0D_V$, and
$\mathsf V^T\mathsf U=D_U\Pi_1D_V+D_V\Pi_0D_U=\mathsf K$.
Equation~\eqref{eq:s4-bordered-contraction} therefore gives
\begin{equation}
 D_{I,J}=(AK)^2\det\mathsf K[I,J],
 \qquad
 \mathsf K_{jj}=U_jV_j,\qquad
 \mathsf K_{ij}=
 \frac{\sqrt{\tau_i\tau_j}}{\tau_i+\tau_j}
 (U_iV_j+V_iU_j)\quad(i\ne j).
 \label{eq:s4-global-minor-evaluation}
\end{equation}
This is a chart evaluation of a globally defined section, not a rational
extension across $AK=0$; the empty section was prescribed independently.  Coefficient Frobenius permutes the complete varying-place orbit and all
summands in the construction.  Trace-dual coefficient extraction
therefore descends the complete span from $\mathbb E\mathbb F_R$ to
$\mathbb E$ without changing its degrees or closedness.  This proves
the proposition.
\end{proof}

In every chart the normalized vectors in
Eq.~\eqref{eq:s4-actual-minor-functions} satisfy
\begin{equation}
 w_j^{[2]}=\mathcal R(\tau_{j+1})w_{j+1},
 \qquad
 \mathcal R(z)\in\operatorname{SL}_2(\mathbb E(z)).
 \label{eq:s4-cyclic-frobenius-system}
\end{equation}
Here $[2]$ means entrywise squaring.  The coefficients and poles of
$\mathcal R$ depend on the six fixed places and on the chosen member of
the finite chart cover, but their number and degrees are bounded
independently of $r$.

Every product
\begin{equation}
 A^{d_*-8u-a-3d}
 \prod_{\nu=1}^{u}D_{I_\nu,J_\nu}
 \prod_{\nu=1}^{a}f_{p_\nu}
 \prod_{\nu=1}^{d}t_{q_\nu},
 \qquad
 u\le7,
 \quad
 a+d\le28,
 \label{eq:s4-generator-products}
\end{equation}
is therefore a zero-cocycle of degree $d_*e_i$ in direction $i$.  Its
exponent of $A$ is nonnegative by the choice of $d_*$.  On the chart
$AK\ne0$, division of its coordinate expression by $A^{d_*}$ gives
\begin{equation}
 \left(\frac K{A^3}\right)^{2u+d}
 \prod_{\nu=1}^{u}\det\mathsf K[I_\nu,J_\nu]
 \prod_{\nu=1}^{a}U_{p_\nu}
 \prod_{\nu=1}^{d}V_{q_\nu}.
 \label{eq:s4-normalized-generator-product}
\end{equation}
This normalization is used only for the algebraic calculation.  All
cochains are the polynomial expressions in
Eq.~\eqref{eq:s4-generator-products}.  Products of $s_*$ such
zero-cocycles have degree
$d_*s_*e_i=d_{\mathrm{prod}}e_i$ in direction $i$ and therefore lie in
$Z^0(X,\mathcal P)$.

\paragraph{Trace realization.}
The common determinant-one frame in
Eq.~\eqref{eq:s4-common-DD-frame} has a rational connection whose
degrees depend only on the fixed-place data.  Define
\begin{equation}
 \mathcal R(z)=\sigma(\mathcal Q)(z)\mathcal Q(z)^{-1},
 \label{eq:s4-connection-from-old-rows}
\end{equation}
where $\sigma$ squares coefficients and fixes $z$.  Then
$\det\mathcal R(z)=1$.

To see the cancellation in this expression, put
$H(z)=\gcd(P_{\mathrm{ref}}(z),\sigma(P_{\mathrm{ref}})(z))$.  Before the common determinant-one lower shear, the fixed-place row core of
the numerator is
\begin{equation}
 \sigma(\mathsf M)(z^2)
 \begin{pmatrix}0&1\\z^2&0\end{pmatrix}
 \operatorname{adj}(\mathsf M)(z^2).
 \label{eq:s4-connection-raw-numerator}
\end{equation}
The shear in Eq.~\eqref{eq:s4-common-DD-frame} is regular at every
allowed node and does not change the common cancellation or the uniform
degree bound below.  At a shared root $\beta$ of $P_{\mathrm{ref}}$ and
$\sigma(P_{\mathrm{ref}})$, write $w=z+\beta$.  The kernels of
$\mathsf M(z^2)\mathsf S(z)$ and its coefficientwise square are,
modulo $w^2$, generated respectively by
$p+(w/\beta)Jp$ and
$(1+w/\beta)(p+(w/\beta)Jp)$, where
$J=\left(\begin{smallmatrix}\omega&1\\1&\omega^2\end{smallmatrix}\right)$.
Thus the matrix in Eq.~\eqref{eq:s4-connection-raw-numerator} vanishes
modulo $w^2$.  Hence $H^2$ divides every one of its entries.

The six lower Frobenius halves have
$b_0=\sum_i\mu_i/(2h)$ boundary pairs.  Consequently
$B(z)=P_{\mathrm{ref}}(z)\sigma(P_{\mathrm{ref}})(z)/H(z)^2$ is monic
of degree $2b_0$ and satisfies $B(0)\ne0$.  The determinant-one shear
used to define the common rows $w_j$ preserves the following parity and
degree form:
\begin{equation}
 \mathcal R(z)=\frac{N(z)}{D(z)},
 \qquad
 N(z)=
 \begin{pmatrix}
 z a(z^2)&z^2b(z^2)\\
 c(z^2)&z d(z^2)
 \end{pmatrix},
 \qquad
 D(z)=z^2B(z),
 \qquad
 \det\mathcal R(z)=1.
 \label{eq:s4-connection-normal-form}
\end{equation}
Here $\deg a,\deg d\le b_0$, while $b$ and $c$ are monic of degrees
$b_0$ and $b_0+1$, respectively.  Thus
$d_{\mathrm{conn}}=2b_0+2$ bounds the degrees of $D$ and all four
entries of the numerator $N(z)$.

Treat the non-leading coefficients of $a,b,c,d$ as variables, subject
only to the determinant identity.  With
$F(t)=\left(\begin{smallmatrix}a(t)&b(t)\\c(t)&d(t)\end{smallmatrix}\right)$,
the coefficient ring is
\begin{equation}
 \mathscr R=
 \frac{\mathbb E[\text{non-leading coefficients of }F]}
 {([t^j](\det F-S):0\le j\le2b_0)},
 \label{eq:s4-determinant-ring}
\end{equation}
where $S(t)$ is characterized by $S(z^2)=z^2B(z)^2$.
The coefficient of $t^{2b_0+1}$ is fixed by the two monic entries, so
the displayed relations are exactly equivalent to
$\det F=S$.  In the cyclic system we localize at
$\prod_jD(\tau_j)$; the later choice of $\zeta$ avoids its zero set.
For $f\in\mathscr R$, let $\deg_{\mathscr R}(f)$ be the least total
degree of a polynomial representative in these coefficients.

\begin{lemma}
\label{lem:s4-exact-frobenius-degree}
The associated graded ring of $\mathscr R$ is geometrically reduced.
Consequently
\begin{equation}
 \deg_{\mathscr R}(f^2)=2\deg_{\mathscr R}(f)
 \qquad(0\ne f\in\mathscr R).
 \label{eq:s4-exact-frobenius-degree}
\end{equation}
\end{lemma}

\begin{proof}
Write $a,d$ with degree at most $b_0$, write $b$ as a monic
polynomial of degree $b_0$, and write $c$ as a monic polynomial of
degree $b_0+1$, as in Eq.~\eqref{eq:s4-connection-normal-form}.  The quadratic leading forms of
the determinant relations are the coefficients of
\begin{equation}
 a(t)d(t)+b_{\mathrm{low}}(t)c_{\mathrm{low}}(t)=0,
 \label{eq:s4-leading-determinant-cone}
\end{equation}
where $b_{\mathrm{low}},c_{\mathrm{low}}$ denote the variable
lower-degree parts.  A nonzero point of
this cone is a rank-one polynomial matrix.  Since
$\mathbb E[t]$ is a principal ideal domain, it has a factorization
$uv^T$ with $u$ primitive, unique up to a nonzero scalar.  Separating the
cases $\deg u_1\ge\deg u_2$, $\deg u_2>\deg u_1$, and the two vanishing
endpoint cases gives finitely many irreducible parameter spaces, each of
dimension $2b_0+2$.  Their union is the whole cone.  Since there are
$2b_0+1$ coefficient relations, the leading ideal has height equal to its
number of generators and is an unmixed complete intersection.

On the two families of components, each indexed by $0\le j<b_0$,
and on the two endpoint components, choose the matrices
\begin{equation}
 \begin{pmatrix}t^{b_0}&t^j\\t^{b_0-j}&1\end{pmatrix},
 \qquad
 \begin{pmatrix}1&t^{b_0-j-1}\\t^{j+1}&t^{b_0}\end{pmatrix},
 \qquad
 \begin{pmatrix}0&0\\1&t^{b_0}\end{pmatrix},
 \qquad
 \begin{pmatrix}t^{b_0}&0\\1&0\end{pmatrix}.
 \label{eq:s4-smooth-points-determinant-cone}
\end{equation}
At each displayed point the differential of the coefficient determinant
map is surjective onto all polynomials of the required degree.  Hence every
component is generically reduced.  An unmixed complete intersection whose
minimal components are generically reduced is reduced.  The same
factorization and Jacobian argument works after every scalar extension.

Homogenize the inhomogeneous determinant equations with one variable
$z_0$.  Modulo $z_0$ they form the regular sequence just considered, so
$z_0$ is a nonzerodivisor after permuting the regular sequence.  The
homogeneous ideal is therefore saturated with respect to $z_0$, and its
fibre at $z_0=0$ is the complete associated graded ideal.  The associated
graded ring is consequently geometrically reduced.  The initial class of
$f$ is nonzero, so its square is nonzero and has twice the degree.  This
proves Eq.~\eqref{eq:s4-exact-frobenius-degree}.
\end{proof}

Introduce the finite free algebra
\begin{equation}
 \mathscr A=
 \frac{\mathscr R[X_j,Y_j:j\in\mathbb Z/r\mathbb Z]}
 {(w_j^{[2]}-\mathcal R(\tau_{j+1})w_{j+1})_j}.
 \label{eq:s4-artin-schreier-algebra}
\end{equation}
The leading monomials of the defining equations are the individual squares,
so the squarefree monomials form a basis of size $R^2$.  Let $\mathsf L$
be the $2r$-by-$2r$ coefficient matrix of their linear terms.  It is a
cyclic block matrix whose nonzero blocks are
$\mathcal R(\tau_j)$, and therefore $\det\mathsf L=1$.
Consequently $\mathscr A$ is finite \'{e}tale over $\mathscr R$.  Write
$\operatorname{NF}(T)$ for the unique squarefree normal form of a
polynomial $T$ in the variables $w_j$.  Its multiplication trace is
\begin{equation}
 \operatorname{Tr}_{\mathscr A/\mathscr R}(T)
 =
 \left[\prod_jX_jY_j\right]\operatorname{NF}(T).
 \label{eq:s4-universal-trace}
\end{equation}
To verify the formula, pass first to a generic splitting field.  For each
solution $u$ of the quadratic system, the polynomial
\begin{equation}
 \det\bigl(\operatorname{diag}(w+u)+\mathsf L\bigr)
 \label{eq:s4-point-indicator-polynomial}
\end{equation}
vanishes on every other solution and has value
$\det\mathsf L=1$ at $u$; its top
squarefree coefficient is one.  Expanding a function in these point
indicators proves Eq.~\eqref{eq:s4-universal-trace}.  Since both sides are
polynomial in the generic matrix entries, the identity descends to
$\mathscr R$.

The matrix $\mathsf K$ changes under one Frobenius step by a matrix of
rank at most $2d_{\mathrm{conn}}$.  More precisely, there are a matrix
$V$ linear in the $w_j$ and a symmetric
$2d_{\mathrm{conn}}$-by-$2d_{\mathrm{conn}}$ matrix
$C_{\mathcal R}$ of coefficient degree at most two such that
\begin{equation}
 \mathsf K((\mathcal R(\tau_j)w_j)_j)
 =\mathsf K((w_j)_j)+VC_{\mathcal R}V^T.
 \label{eq:s4-finite-rank-update}
\end{equation}
Write
$J_{\mathrm{sp}}=\left(\begin{smallmatrix}0&1\\1&0\end{smallmatrix}\right)$
for the alternating matrix in characteristic two.  For off-diagonal
entries, Eq.~\eqref{eq:s4-finite-rank-update} is the B\'{e}zout identity
obtained by dividing
\begin{equation}
 N(x)^TJ_{\mathrm{sp}}N(y)+D(x)D(y)J_{\mathrm{sp}}
 \label{eq:s4-bezout-numerator}
\end{equation}
by $x+y$.  On the diagonal it follows by differentiating at $x=y$; the
parity of the four entries in Eq.~\eqref{eq:s4-connection-normal-form}
gives
\begin{equation}
 U V+\bigl(\mathcal R(z)w\bigr)_1
       \bigl(\mathcal R(z)w\bigr)_2
 =z\,w^T\mathcal R(z)^TJ_{\mathrm{sp}}\mathcal R'(z)w.
 \label{eq:s4-diagonal-bezout-identity}
\end{equation}
Thus Eq.~\eqref{eq:s4-finite-rank-update} includes the diagonal entries.

For an index set $I$, write $I+1=\{i+1:i\in I\}$ with cyclic indices.
Squaring a minor and applying Cauchy--Binet to the rank-$2d_{\mathrm{conn}}$ update gives
\begin{equation}
 \det\mathsf K[I,J]^2
 =
 \sum_{\substack{A\subseteq I+1,\ B\subseteq J+1\\|A|=|B|}}
 \det\mathsf K[(I+1)\setminus A,(J+1)\setminus B]
 \det(VC_{\mathcal R}V^T)[A,B].
 \label{eq:s4-minor-frobenius-reduction}
\end{equation}
Every nonzero correction uses at most
$2d_{\mathrm{conn}}$ rows of $V$, and its additional coordinate and
coefficient degrees are at most $4d_{\mathrm{conn}}$.  To keep track of
iteration, let $\mathcal W_{u,a}$ be the span of products of at most $u$
minors with a coordinate polynomial of degree at most $a$, and let
$M_a$ be the largest value of
$\deg_{\mathscr R}\operatorname{Tr}_{\mathscr A/\mathscr R}(T)$ for
$T\in\mathcal W_{u,a}$.  Direct normal-form reduction gives
$M_a\le2ur+a$.  Squaring the whole word, rather than its factors
separately, and applying Eq.~\eqref{eq:s4-minor-frobenius-reduction}
gives terms in $\mathcal W_{u,a+4u d_{\mathrm{conn}}}$ with coefficient
degree at most $a+4u d_{\mathrm{conn}}$.  The exact Frobenius degree in
Eq.~\eqref{eq:s4-exact-frobenius-degree} therefore gives
\begin{equation}
 2M_a\le
 M_{a+4u d_{\mathrm{conn}}}+a+4u d_{\mathrm{conn}},
 \qquad
 M_a\le a+8u d_{\mathrm{conn}}.
 \label{eq:s4-uniform-fourth-moment-bound}
\end{equation}
The second inequality follows by iterating the first and using the
preliminary bound; the final error is
$2^{-n}(2ur-8u d_{\mathrm{conn}})$ and tends to zero.  The bound is
independent of $r$, the sizes and overlaps of the minors, and the number
of summands.  It also holds after deleting the zero solution, whose
correction is $T(0)$ and has coefficient degree zero.

\paragraph{The common dual factor.}
The trace calculation has so far been carried out in normalized chart
coordinates.  We now reconstruct those coordinates from global
zero-cocycles before choosing an evaluation point.

Recall that $C=\widehat C_{\varnothing}$ is the second determinant
variant with no auxiliary columns.  For each adjacent pair
$(\sqrt{\upsilon_\kappa},\upsilon_\kappa)$ in the fixed auxiliary orbit,
set
$K_\kappa=\langle
\widehat D_{\{\sqrt{\upsilon_\kappa}\}}^{[2]},
\widehat D_{\{\upsilon_\kappa\}}\rangle$ and
$J_\kappa=\langle
\widehat D_{\{\sqrt{\upsilon_\kappa}\}}^{[2]},
\widehat C_{\{\upsilon_\kappa\}}\rangle$.
These are zero-cocycles of degree $3e_i$ in direction $i$.  The index $\kappa$ records the
finite choice of adjacent pair.  On $AK_\kappa\ne0$, put
$\eta=C/A$, $\beta_\kappa=\upsilon_\kappa J_\kappa/K_\kappa$, and
$v_\kappa=(\eta,\beta_\kappa)^T$.  These are chart coordinates whose
numerators and denominators are global zero-cocycles.

The common fixed-place row interpolation gives a connection
\begin{equation}
 \mathcal R_\kappa(z)
 =
 \frac{1}{z^2B(z)}
 \begin{pmatrix}
 zF_{\kappa,00}(z^2)&z^2F_{\kappa,01}(z^2)\\
 F_{\kappa,10}(z^2)&zF_{\kappa,11}(z^2)
 \end{pmatrix}.
 \label{eq:s4-chart-connection}
\end{equation}
The polynomial $B$ is independent of $\kappa$, has degree $2b_0$, and
satisfies $B(0)\ne0$.  We now recover every coefficient of the numerator
from genuine scalar sections.  For adjacent sample nodes $z,z^2$ in the
same fixed auxiliary orbit, let $f_z,t_z$ be the corresponding covariant
rows constructed in Proposition~\ref{prop:s4-global-minor-sections}, and
define
\begin{equation}
 W_\kappa(z)=
 \begin{pmatrix}
 K_\kappa^3\langle f_z^{[2]},f_{z^2}\rangle&
 A K_\kappa^2\langle f_z^{[2]},t_{z^2}\rangle\\
 A^2K_\kappa\langle t_z^{[2]},f_{z^2}\rangle&
 A^3\langle t_z^{[2]},t_{z^2}\rangle
 \end{pmatrix}J_{\mathrm{sp}} .
 \label{eq:s4-sample-pairing-matrix}
\end{equation}
The four unscaled pairings have weighted degrees $3,5,7,9$, respectively, so
every entry of $W_\kappa(z)$ has weighted degree $12$.  The determinant-one
Gram identity gives
$W_\kappa(z)=A^3K_\kappa^3\mathcal R_\kappa(z)$ at every sample node.
Since the fixed orbit contains more than $2b_0+2$ nodes, interpolation of
$z^2B(z)W_\kappa(z)$, followed by removal of the explicit powers of $z$
in Eq.~\eqref{eq:s4-chart-connection}, gives zero-cocycles of degree $12e_i$ in direction $i$
$W_{\kappa,\mu}$ such that
\begin{equation}
 F_{\kappa,\mu}
 =\frac{W_{\kappa,\mu}}{A^3K_\kappa^3}
 \qquad\text{on }AK_\kappa\ne0,
 \label{eq:s4-actual-connection-coefficients}
\end{equation}
where $\mu$ records a matrix entry and a polynomial coefficient.

The upper-right numerator in Eq.~\eqref{eq:s4-chart-connection} is
monic.  Its values at the adjacent auxiliary nodes are fixed nonzero
multiples of $K_\kappa/A^3$.  Extracting its leading coefficient by
Lagrange interpolation gives fixed $c_\kappa\in\mathbb E$ satisfying
\begin{equation}
 A\left(\sum_\kappa c_\kappa K_\kappa-A^3\right)=0.
 \label{eq:s4-finite-chart-cover}
\end{equation}
Thus every evaluation with $A=1$ lies in at least one chart
$K_\kappa\ne0$.  The prefactor $A$ makes this an equality of global
sections, including the zero set of $A$.

Let $Z_\kappa(t)=A^3K_\kappa^3F_\kappa(t)$ and
$B_0=B(0)$.  The constant terms of the same row calculation give
the origin relation
\begin{equation}
 F_\kappa(0)=B_0v_\kappa^{[2]}v_\kappa^TJ_{\mathrm{sp}},
 \qquad
 (v_\kappa^{[2]})^TJ_{\mathrm{sp}}F_\kappa'(0)v_\kappa=B_0.
 \label{eq:s4-origin-relations}
\end{equation}
Here the derivative is with respect to $t=z^2$.  These rational
relations are imposed on the parameter space through the following
polynomial identities:
\begin{equation}
 \begin{aligned}
 Z_{\kappa,00}(0)&=B_0\upsilon_\kappa A C^2K_\kappa^2J_\kappa,&
 Z_{\kappa,01}(0)&=B_0C^3K_\kappa^3,\\
 Z_{\kappa,10}(0)&=B_0\upsilon_\kappa^3A^3J_\kappa^3,&
 Z_{\kappa,11}(0)&=B_0\upsilon_\kappa^2A^2CK_\kappa J_\kappa^2 .
 \end{aligned}
 \label{eq:s4-cleared-origin-relations}
\end{equation}
Each side has weighted degree $12$.  Multiplying each difference by
$AK_\kappa$ gives a global relation of weighted degree $16$; the determinant
relations have weighted degree $28$, and the parity and monicity
relations have weighted degree $16$.  Let $d_{\mathrm{alg}}$ be a fixed
common upper bound for these weighted degrees, to be enlarged below if
necessary.  Consequently the apolar point chosen below
satisfies the determinant scheme, the chart-cover identity, and
Eq.~\eqref{eq:s4-origin-relations} after division by the selected
nonzero values of $A$ and $K_\kappa$.

At such a point, $v_\kappa\ne0$: otherwise $F_\kappa(0)=0$, whereas the
coefficient of $t$ in $\det F_\kappa(t)=S(t)$ is $B_0^2\ne0$.  Choose
$L_0\in\operatorname{SL}_2(\mathbb E)$ with $L_0v_\kappa=e_1$ and put
\begin{equation}
 H_{L_0}(z)=
 \begin{pmatrix}1&0\\0&z^{-1}\end{pmatrix}
 L_0
 \begin{pmatrix}1&0\\0&z\end{pmatrix},
 \qquad
 \mathcal R^\circ(z)=
 H_{L_0^{[2]}}(z)\mathcal R_\kappa(z)H_{L_0}(z)^{-1}.
 \label{eq:s4-origin-gauge}
\end{equation}
Put
$F_\kappa^\circ(t)=L_0^{[2]}F_\kappa(t)L_0^{-1}$.  Since
$L_0v_\kappa=e_1$ and every determinant-one matrix preserves
$J_{\mathrm{sp}}$, the first relation in
Eq.~\eqref{eq:s4-origin-relations} gives
$F_\kappa^\circ(0)
=B_0\left(\begin{smallmatrix}0&1\\0&0\end{smallmatrix}\right)$.
The second relation, equivalently the coefficient of $t$ in
$\det F_\kappa^\circ(t)=tB(z)^2$, gives
$(F_{\kappa,10}^\circ)'(0)=B_0$.  Consequently the apparent divisions
by $z$ and $z^2$ in Eq.~\eqref{eq:s4-chart-connection} cancel after the
gauge.  The resulting connection is regular at zero, its numerator and
denominator degrees are still at most
$d_{\mathrm{conn}}=2b_0+2$, and
\begin{equation}
 \mathcal R^\circ(0)=J_{\mathrm{sp}}.
 \label{eq:s4-connection-at-zero}
\end{equation}
Since $\tau_{j+1}=\tau_j^2$,
$H_{L_0}(\tau_j)^{[2]}
=H_{L_0^{[2]}}(\tau_{j+1})$.  It follows directly from
Eq.~\eqref{eq:s4-origin-gauge} that
$(w_j)_j\mapsto(H_{L_0}(\tau_j)w_j)_j$ is a bijection between the two
cyclic solution sets.  Because $\tau_r=\tau_0$, it also preserves the
closing equation and the zero solution.

Moreover, if
$L_0=\left(\begin{smallmatrix}p&q\\r&s\end{smallmatrix}\right)$ and
$X$ has $j$-th row
$(U_j/\sqrt{\tau_j},\sqrt{\tau_j}V_j)$, direct expansion gives
\begin{equation}
 \mathsf K(H_{L_0}w)=
 \mathsf K(w)+
 X\begin{pmatrix}pr&qr\\qr&qs\end{pmatrix}X^T.
 \label{eq:s4-origin-gauge-rank-two}
\end{equation}
Thus every minor after the origin gauge is a sum of a minor before the gauge times
at most four coordinate factors.  The allowance of at most $28$ row
factors in Eq.~\eqref{eq:s4-generator-products} includes this expansion
for every product of at most seven minors.  From this point onward,
$\mathcal R$, $w_j$, and $\mathsf K$ denote the origin-normalized
objects.  The evaluation point is an algebraic point of the parameter
space, not a vertex of $X$.

Collect $A$ and the finitely many scalar zero-cocycles used in
Eqs.~\eqref{eq:s4-actual-connection-coefficients}--%
\eqref{eq:s4-origin-relations} as homogeneous variables
$A,Z_1,\ldots,Z_n$, with their degrees in the coefficient sheaves as
positive weights.  Their number $n_{\mathrm{alg}}$ is fixed independently
of $R$.  Let
$\mathscr S=\mathbb E[A,Z_1,\ldots,Z_n]$ with this weighted grading.

We use the following duality statement.

\begin{lemma}[Weighted apolarity]
\label{lem:s4-weighted-apolarity}
Let $\Lambda:\mathscr S_D\to\mathbb E$ be linear and suppose
$\Lambda(A^D)\ne0$.  If $D\ge d+(n+1)d^{n+1}$, then, after one fixed finite extension of $\mathbb E$, there are a point
$\alpha$ and $g_\alpha\in\mathscr S_{D-d}$ such that
\begin{equation}
 \Lambda(fg_\alpha)=f(1,\alpha)
 \qquad
 (f\in\mathscr S_d).
 \label{eq:s4-apolar-evaluation}
\end{equation}
\end{lemma}

\begin{proof}
For $0\le j\le D$, let $I_j$ consist of the elements
$f\in\mathscr S_j$ such that $\Lambda(fg)=0$ for every
$g\in\mathscr S_{D-j}$, and put $I_j=\mathscr S_j$ for $j>D$.
Then $I=\bigoplus_{j\ge0}I_j$ is a homogeneous ideal: if
$f\in I_j$ and $h\in\mathscr S_k$ with $j+k\le D$, then
$\Lambda((hf)g)=\Lambda(f(hg))=0$ for every
$g\in\mathscr S_{D-j-k}$, while the case $j+k>D$ follows from the
definition.  If
the affine zero set of $I_d$ in the chart $A=1$ were empty, the effective
Nullstellensatz would give a certificate for $1$ of degree at most
$(n+1)d^{n+1}$.  Embed $\mathscr S$ into an ordinary polynomial ring by replacing each
variable of weight $w$ by the $w$-th power of a new variable.  The larger
ring is a free $\mathscr S$-module with basis given by the monomials whose
exponents in those new variables are smaller than their weights.
Homogenize the certificate and apply the $\mathscr S$-linear projection
onto the basis vector $1$.  This gives
$A^{d+(n+1)d^{n+1}}\in I$.
Multiplication by the remaining power of $A$ contradicts
$\Lambda(A^D)\ne0$.  Thus $I_d$ has an affine zero $\alpha$.
Evaluation at $(1,\alpha)$ annihilates the kernel of the map
$\mathscr S_d\to\mathscr S_{D-d}^*$ given by
$f\mapsto(g\mapsto\Lambda(fg))$.  Finite-dimensional linear algebra
gives $g_\alpha$ satisfying
Eq.~\eqref{eq:s4-apolar-evaluation}.
\end{proof}

Enlarge $d_{\mathrm{alg}}$, if necessary, so that it also bounds the
polynomials obtained from four products of $s_*$ normalized expressions
in Eq.~\eqref{eq:s4-normalized-generator-product}.  Such a bound follows
from Eq.~\eqref{eq:s4-uniform-fourth-moment-bound} and is independent of
$R$.  Choose fixed clearing exponents $e_0,a_0$ and enlarge
$d_{\mathrm{alg}}$ once more, if necessary, so that
\begin{equation}
 d_{\mathrm{ch}}=4d_{\mathrm{prod}}+a_0+3e_0+4
 \le d_{\mathrm{alg}}.
 \label{eq:s4-chart-clearing-degrees}
\end{equation}
We next identify this algebraic sum with the varying-place part of the
cup product pairing.  The coset bijection in
Equation~\eqref{eq:s4-level-cosets} identifies the fibre of $X_R$
over a fixed vertex with $\mathbb F_R^2\setminus\{0\}$: the coset
$U_Rg$ corresponds to $v=g^{-1}e_1$, whose input in the factor indexed by $j$
is $v^{[2^j]}$.  Let $\mathcal Q_j$ be the origin-normalized common
determinant-one frame underlying $w_j$ at the node $\tau_j$, including
the common lower shear in the construction of $t_j$.  These frames are
the evaluations of one rational frame on the complete Frobenius orbit,
so
\begin{equation}
 \mathcal Q_j^{[2]}\mathcal Q_{j+1}^{-1}
 =\mathcal R(\tau_{j+1}),
 \qquad
 \mathcal Q_r=\mathcal Q_0.
 \label{eq:s4-cyclic-frame-closure}
\end{equation}
The second equality follows from $\tau_r=\tau_0$, and the origin gauge
preserves both identities by the covariance proved above.  Hence
\begin{equation}
 v\longmapsto
 \bigl(\mathcal Q_jv^{[2^j]}\bigr)_{j\in\mathbb Z/r\mathbb Z}
 \label{eq:s4-fibre-to-cyclic-solutions}
\end{equation}
maps $\mathbb F_R^2$ into the solutions of
Eq.~\eqref{eq:s4-cyclic-frobenius-system}.

Conversely, for a cyclic solution $(w_j)_j$, put
$u_j=\mathcal Q_j^{-1}w_j$.  Equations
\eqref{eq:s4-cyclic-frobenius-system} and
\eqref{eq:s4-cyclic-frame-closure} give
$u_{j+1}=u_j^{[2]}$.  The closing equation then gives
$u_0^{[2^r]}=u_0$, so $u_0\in\mathbb F_R^2$ and
$w_j=\mathcal Q_ju_0^{[2^j]}$.  Thus
Eq.~\eqref{eq:s4-fibre-to-cyclic-solutions} is a bijection from
$\mathbb F_R^2$ onto all cyclic solutions.  Removing the zero vector
identifies the geometric fibre with the nonzero solution set $\Omega_R$.

Under this bijection, the normalized value of every generating product is
exactly the polynomial in
Eq.~\eqref{eq:s4-normalized-generator-product}.  Therefore the sum of a
product $T$ of four such expressions over one geometric fibre is
$\operatorname{Tr}_{\mathscr A/\mathscr R}(T)+T(0)$, where addition
equals subtraction in characteristic two; equivalently, it is the sum over
$\Omega_R$.

Let $\Sigma_R(P_1,P_2,P_3,P_4)$ denote this fibre sum.  In every chart
$AK_\kappa\ne0$ its normalized value is a polynomial in
$F_{\kappa,\mu}$ and $K_\kappa/A^3$.  Substituting
Eq.~\eqref{eq:s4-actual-connection-coefficients} and multiplying by the
indicated nonnegative powers gives a global homogeneous identity
\begin{equation}
 A^{a_0+1}K_\kappa^{e_0+1}
 \Sigma_R(P_1,P_2,P_3,P_4)
 =
 Q_{\kappa;P_1,P_2,P_3,P_4},
 \qquad
 \deg Q_{\kappa;P_1,P_2,P_3,P_4}=d_{\mathrm{ch}}.
 \label{eq:s4-global-trace-clearing}
\end{equation}
The equality is global: on the chart it is the substitution just
described, while off the chart the factor $AK_\kappa$ makes both sides
zero.  Thus no inverse of a global cochain is used.

All scalar zero-cocycles represented by variables of $\mathscr S$, as well as the three fixed
cocycles $\gamma_a$, descend from the fixed quotient before the
varying-place cover.  Pullback to $X_R$ multiplies the resulting
top-dimensional pairing by the covering degree $R^2-1$, which equals one
in characteristic two.  Put
$D_0=2^{h-1}-4+3d_{\mathrm{prod}}$.  The fixed value of $h$ is chosen so that
$D_0\ge d_{\mathrm{alg}}+
(n_{\mathrm{alg}}+1)d_{\mathrm{alg}}^{\,n_{\mathrm{alg}}+1}$.
Hence the following functional is independent of $r$ and may be fixed
before the growing place is chosen.  Define
$\Lambda:\mathscr S_{D_0}\to\mathbb E$ by
\begin{equation}
 \Lambda(f)=
 \left\langle
 \gamma_1\smile\gamma_2\smile\gamma_3,
 A^{3(\ell-d_{\mathrm{prod}})}
 f(A,Z_1,\ldots,Z_n)\frown[X]
 \right\rangle.
 \label{eq:s4-apolar-functional}
\end{equation}
The parameter identity
\begin{equation}
 3(\ell-d_{\mathrm{prod}})+D_0=2^h-6
 \label{eq:s4-apolar-old-degree}
\end{equation}
and Proposition~\ref{prop:s4-fixed-period} give
\begin{equation}
 \Lambda(A^{D_0})\ne0.
 \label{eq:s4-apolar-anchor}
\end{equation}
The choices of $D_0$, $d_{\mathrm{alg}}$, and $n_{\mathrm{alg}}$ permit
Lemma~\ref{lem:s4-weighted-apolarity} with
$D=D_0$ and $d=d_{\mathrm{alg}}$.  The resulting point $\alpha$ and
the coefficients of $g_\alpha$ lie in a finite extension; this is the
extension included in the fixed coefficient field $\mathbb E$ in the
first subsection.  Substitute the chosen scalar zero-cocycles for the
variables of $g_\alpha$ and use the same notation for the resulting
zero-cocycle of degree $(D_0-d_{\mathrm{alg}})e_i$ in direction $i$.
Equation~\eqref{eq:s4-finite-chart-cover} with
$A(\alpha)=1$ supplies a chart with $K_\kappa(\alpha)\ne0$; fix it and
write $K=K_\kappa$.  Define
\begin{equation}
 G=
 K(\alpha)^{-(e_0+1)}
 A^{a_0+1+d_{\mathrm{alg}}-d_{\mathrm{ch}}}
 K^{e_0+1}g_\alpha.
 \label{eq:s4-global-dual-cocycle}
\end{equation}
The inverse in this formula is a fixed nonzero field scalar.  The
exponent of $A$ is
\begin{equation}
 a_0+1+d_{\mathrm{alg}}-d_{\mathrm{ch}}
 =d_{\mathrm{alg}}-4d_{\mathrm{prod}}-3e_0-3>0.
 \label{eq:s4-dual-A-exponent}
\end{equation}
The exponent of $K$ is $e_0+1>0$, and the defining inequality for $D_0$ gives
$D_0-d_{\mathrm{alg}}\ge0$ for the apolar factor.  Thus every power of an
global zero-cocycle is nonnegative.  The degree of $G$ in direction $i$ is
\begin{equation}
 (d_{\mathrm{alg}}-4d_{\mathrm{prod}})+(D_0-d_{\mathrm{alg}})
 =D_0-4d_{\mathrm{prod}}
 =2^{h-1}-4-d_{\mathrm{prod}}.
 \label{eq:s4-global-dual-degree}
\end{equation}
Consequently, for every $P\in Z^0(X,\mathcal P)$ constructed from the
generating products above,
\begin{equation}
 PG\in
 Z^0\!\left(
 X,
 (\mathcal F_1\otimes\mathcal F_2\otimes\mathcal F_3)^\perp
 \right).
 \label{eq:s4-dual-membership-general}
\end{equation}
Indeed $P$ has degree $d_{\mathrm{prod}}e_i$ in direction $i$, so
Eq.~\eqref{eq:s4-global-dual-degree} gives the complementary coefficient
$2^{h-1}-4$ from Eq.~\eqref{eq:s4-dual-completion}.  Fix a
codimension-one face $\rho$ missing direction $i$ and use its designated
chart.  Every factor used to construct $PG$ is a global zero-cocycle
whose normalization is constant on the top cubes above $\rho$.
Consequently the restriction of $PG$ is a face-dependent scalar times
the standard chart evaluation of a polynomial of degree
$(2^{h-1}-4)e_i$.  A word of
$\mathcal F_1\otimes\mathcal F_2\otimes\mathcal F_3$ has degree
$(3\ell+7)e_i$.  In their coordinatewise product, the two chart powers
of $c_i(v,Q)$ add to $s_i-1$ and hence cancel, while Frobenius only
permutes the projective labels.  Their pairing is therefore a nonzero
constant times the projective sum of a homogeneous form of degree
$s_i-1$, which vanishes by Eq.~\eqref{eq:s4-prs-dual}.  Thus the
restriction of $PG$ belongs to $C_\rho^\perp$ for every $\rho$, proving
the coordinate-dual condition in
Eq.~\eqref{eq:s4-dual-membership-general}.

Let $W_R$ be the $\mathbb L_R$-space of products of exactly $s_*$
Moore expressions from Eq.~\eqref{eq:s4-generator-products}, padding
with powers of $A^{d_*}$ so that all products have the same degree.  The
same expression has two realizations: its global zero-cocycle on $X$ and
its specialized function on the cyclic solution set.  Thus there are
linear maps
\begin{equation}
 \rho_X:W_R
 \longrightarrow
 Z^0(X,\mathcal P)\otimes_{\mathbb E}\mathbb L_R,
 \qquad
 \rho_\alpha:W_R
 \longrightarrow
 \operatorname{Fun}(\Omega_R,\mathbb L_R).
 \label{eq:s4-formal-product-maps}
\end{equation}
The map $\rho_X$ substitutes the globally defined zero-cocycles in
Eq.~\eqref{eq:s4-generator-products}; $\rho_\alpha$ substitutes their
normalized expressions in
Eq.~\eqref{eq:s4-normalized-generator-product} on the cyclic solution set.
No relation between the kernels of these two maps is assumed.

For $w_1,w_2,w_3,w_4\in W_R$, Eq.~\eqref{eq:s4-global-trace-clearing},
the fibre identification above, and the apolar identity give
\begin{equation}
 \begin{aligned}
 &\left\langle
 u_1(\rho_X(w_1))\smile
 u_2(\rho_X(w_2))\smile
 u_3(\rho_X(w_3)),
 \rho_X(w_4)G\frown[X]
 \right\rangle\\
 &\qquad=
 \sum_{x\in\Omega_R}
 \prod_{i=1}^4\rho_\alpha(w_i)(x).
 \end{aligned}
 \label{eq:s4-fourth-product-sum}
\end{equation}
This identity holds for every four formal linear combinations, including
mixed linear combinations.  The same point, chart, and zero-cocycle $G$
work for every sufficiently large $r$, because all auxiliary quantities in
$\mathscr S$ and all identities of degree at most $d_{\mathrm{alg}}$
were fixed before the growing field was chosen.

\paragraph{Linear dimension.}
We prove that the evaluations of the degree-$d_*e_i$ zero-cocycles have
linear dimension.  From this point on, $\mathcal R$ denotes the
connection specialized at the single point $\alpha$, in the chart chosen
above, and transformed by the origin gauge in
Eq.~\eqref{eq:s4-origin-gauge}.  No further specialization of
the fixed auxiliary data is made when $r$ grows; the element
$\zeta\in\mathbb F_R$ will be chosen below.

Put
$m_0=\lfloor(r-13)/3\rfloor$ and $M=2^{m_0}$.  Starting from
\begin{equation}
 S_0=\begin{pmatrix}1&\omega\\1&\omega^2\end{pmatrix},
 \qquad \omega^2+\omega+1=0,
 \label{eq:s4-reference-frame}
\end{equation}
define determinant-one matrices
\begin{equation}
 Q_0=S_0,
 \qquad
 Q_{j+1}=\mathcal R(\zeta^{2^{j+2}})^{-1}Q_j^{[2]}.
 \label{eq:s4-finite-frame-recursion}
\end{equation}
These recursively defined frames are used only in the rank argument.
Unlike the cyclic frames $\mathcal Q_j$ in
Eq.~\eqref{eq:s4-cyclic-frame-closure}, they need not satisfy
$Q_r=Q_0$.  Because of Eq.~\eqref{eq:s4-connection-at-zero},
\begin{equation}
 Q_j=S_0+O(\zeta^{2^{j+1}}).
 \label{eq:s4-frame-approximation}
\end{equation}
Put $v_j=(x^{2^j},y^{2^j})^T$ and $w_j=Q_jv_j$.  For the reference
matrices $Q_j=S_0$, write
\begin{equation}
 \mathsf Q=x^2+xy+y^2,
 \qquad
 \mathsf F=x^2y+xy^2.
 \label{eq:s4-reference-invariants}
\end{equation}
If $I_n\subseteq\{0,\ldots,r-1\}$ is the binary support of $n$, put
$I_n+1=\{j+1:j\in I_n\}$ and define the minors of the matrix
$\mathsf K$ by
\begin{equation}
 P_n=\det\mathsf K[I_n,I_n],
 \qquad
 G_n=\det\mathsf K[I_n,I_n+1].
 \label{eq:s4-principal-shifted-minors}
\end{equation}
Write $P_n^{(0)}$ and $G_n^{(0)}$ for the same minors after every
$Q_j$ has been replaced by the reference matrix $S_0$.  Let $\nu$ be
the coefficientwise $\zeta$-adic valuation.  Direct determinant
expansion in the reference frame gives
\begin{equation}
 P_n^{(0)}
 =
 \sum_{u\ge0}p_{n,u}\mathsf Q^{\,n-3u}\mathsf F^{\,2u},
 \qquad
 p_{n,0}=1,
 \qquad
 \nu(p_{n,u})>0\quad(u>0),
 \label{eq:s4-principal-expansion}
\end{equation}
and, for $n=2t+1$,
\begin{equation}
 \frac{G_{2t+1}^{(0)}}{\mathsf F}
 =\sum_{u=0}^{t}g_{t,u}
 \mathsf Q^{3(t-u)}\mathsf F^{2u}.
 \label{eq:s4-shifted-expansion}
\end{equation}
The Cauchy determinant formula applied separately to the even and odd
binary supports gives
\begin{equation}
 \nu(g_{t,u})\ge2t+1,
 \qquad
 \nu(g_{t,t})=2t+1,
 \label{eq:s4-shifted-leading-term}
\end{equation}
and the normalized leading coefficient of $g_{t,t}$ is one.

Let $b=3M-3$ and set
\begin{equation}
 B_t=P_{b-3t}^{(0)}\frac{G_{2t+1}^{(0)}}{\mathsf F}
 =\sum_{k=0}^{M-1}a_{t,k}
 \mathsf Q^{b-3k}\mathsf F^{2k},
 \qquad 0\le t<M.
 \label{eq:s4-triangular-polynomials}
\end{equation}
Equations~\eqref{eq:s4-principal-expansion}--%
\eqref{eq:s4-shifted-leading-term} imply
\begin{equation}
 \begin{cases}
 \nu(a_{t,k})\ge2t+1,&k<t,\\
 \nu(a_{t,t})=2t+1\text{ with leading coefficient }1,&k=t,\\
 \nu(a_{t,k})>2t+1,&k>t.
 \end{cases}
 \label{eq:s4-triangular-valuations}
\end{equation}
After dividing row $t$ by $\zeta^{2t+1}$, the coefficient matrix
$(a_{t,k})$ reduces to a lower triangular matrix with diagonal one.

This reference calculation controls whole polynomials in the recursively defined
finite frame.  Define $E_t=P_{2b-6t}G_{4t+2}$ for $0\le t<M$.
Every $E_t$ is homogeneous of degree $12M-6$.  In the reference frame,
shifting both indices squares every entry of $\mathsf K$, and hence
\begin{equation}
 P_{2b-6t}^{(0)}G_{4t+2}^{(0)}
 =\mathsf F^2B_t^2
 =\sum_{k=0}^{M-1}a_{t,k}^2
   \mathsf Q^{2b-6k}\mathsf F^{4k+2}.
 \label{eq:s4-reference-whole-words}
\end{equation}
The forms
$\mathsf Q^{2b-6k}\mathsf F^{4k+2}$ are independent: their lowest
$y$-degrees are the distinct integers $4k+2$.  Moreover,
Eq.~\eqref{eq:s4-frame-approximation} and determinant multilinearity give
$\nu(E_t-\mathsf F^2B_t^2)>4t+2$.  Therefore the same fixed coefficient
projection has, on the $E_t$, a
matrix whose determinant has valuation
$\sum_{t=0}^{M-1}(4t+2)=2M^2$.  In particular, the $E_t$ are independent
whenever that determinant and the frame denominators are nonzero.

Use three index intervals beginning at $0$, $m_0+4$, and $2m_0+8$.
The binary supports of $2b-6t$ and $4t+2$ use indices at most
$m_0+2$.  After the last starting point, the largest index used is
therefore
$3m_0+10\le r-3$, so none of these words wraps around the cyclic system.  We first bound the exceptional
values for the interval beginning at zero.  Put
$L_{\mathrm{int}}=m_0+3$.  In an uncancelled presentation
$Q_j=N_j/D_j$, let $h_j$ bound the degrees in $\zeta$ of
$D_j$ and of every entry of $N_j$.  The finite recursion in
Eq.~\eqref{eq:s4-finite-frame-recursion} gives
\begin{equation}
 h_0=0,
 \qquad
 h_j\le d_{\mathrm{conn}}j2^{j+1},
 \qquad
 \sum_{j<L_{\mathrm{int}}}h_j
 =O(d_{\mathrm{conn}}L_{\mathrm{int}}M).
 \label{eq:s4-frame-height}
\end{equation}
A determinant term uses the variables at each index at most twice.  Hence the polynomial
\begin{equation}
 \mathscr D=
 \prod_{j<L_{\mathrm{int}}}D_j^2
 \prod_{0\le i<j<L_{\mathrm{int}}}(\tau_i+\tau_j)^4
 \label{eq:s4-common-clearing-factor}
\end{equation}
clears every coefficient of every minor on these indices.  Its degree,
and the degrees of the frame and Cauchy numerators after clearing, are
$O((d_{\mathrm{conn}}+1)L_{\mathrm{int}}M)$.  Multiplication by
$\mathscr D^2$ therefore clears each two-minor product $E_t$.  After
applying the fixed coefficient projection to its $M$ rows, the resulting
nonzero determinant has degree
$O((d_{\mathrm{conn}}+1)L_{\mathrm{int}}M^2)$.

It remains to see that a high starting index does not raise this bound by
its Frobenius power.  Let $\sigma$ square coefficients and put
$\mathcal R_s=\sigma^{-s}(\mathcal R)$.  Then
\begin{equation}
 \mathcal R\!\left(\zeta^{2^{s+j+2}}\right)
 =
 \left(
 \mathcal R_s\!\left(\zeta^{2^{j+2}}\right)
 \right)^{[2^s]}.
 \label{eq:s4-shifted-connection}
\end{equation}
The pulled-back connection has the same degree bound and the same value
at zero.  Starting its recursion at $S_0$ and the raw recursion at
$S_0^{[2^s]}$ shows that the latter is the $2^s$-power of the former.
The Cauchy coefficients obey the same identity.  A different
initial frame is absorbed by one invertible change of the two starting
variables.  Thus the coefficient determinant for a block beginning at
$s$ vanishes exactly when its bounded-degree pullback vanishes.

For the three intervals, the number of roots excluded by the determinant
polynomials, their common denominators, and the poles in the cyclic
recursion is
$O((d_{\mathrm{conn}}+1)L_{\mathrm{int}}M^2+r d_{\mathrm{conn}})$.
A nonzero polynomial
over an algebraic extension of $\mathbb F_2$ has at most its degree many
roots in $\mathbb F_R$, while fewer than $2\sqrt R$ elements of
$\mathbb F_R$ lie in proper subfields.  The Frobenius saturation of all
excluded values still has size $o(R)$: it enlarges the preceding bound by
at most a factor of $r$, while $L_{\mathrm{int}}=O(r)$ and
$M^2=O(R^{2/3})$.  We may
therefore choose $\zeta\in\mathbb F_R$ of exact binary degree $r$ such
that all three blocks are defined and independent, together with every
Frobenius conjugate required by the construction.  Let $\mathfrak p_r$ be the minimal polynomial of $\zeta$ over
$\mathbb F_2$, and use the distinguished root $\zeta$ to identify
$\mathbb F_2[T]/(\mathfrak p_r)$ with $\mathbb F_R$.  Since $\zeta$
does not lie in a proper subfield, $\mathfrak p_r$ is irreducible of
degree $r$ and serves as the varying-place prime.  Once $r$ exceeds the
degrees of the fixed arithmetic exclusions, an exact-degree element
automatically avoids them.  Hence the construction is available for every
sufficiently large even $r$, rather than only along a subsequence.

After the three Frobenius shifts, the blocks consist of forms of degrees
$12M-6$, $(16M)(12M-6)$, and $(16M)^2(12M-6)$.  Coefficient Frobenius
preserves their dimensions because the coefficient field is perfect.
The three invertible starting-variable changes may be different; each
block is nevertheless a subspace of the corresponding full homogeneous
space.  Multiplication on these three full spaces is injective: the
$x$-exponent of a product monomial is
$a_0+(16M)a_1+(16M)^2a_2$, where
$0\le a_0,a_1,a_2\le12M-6<16M$, and hence has a unique base-$16M$
expansion.  The resulting product space
has dimension $M^3$, and all its elements have the same degree.  Since
\begin{equation}
 \frac{R}{2^{15}}\le M^3\le\frac{R}{2^{13}},
 \qquad
 (12M-6)(1+16M+256M^2)<\frac{3R}{8},
 \label{eq:s4-capacity-product-degree}
\end{equation}
that degree is less than $R/2$.

For every even-cardinality subset
$E\subseteq\{0,\ldots,r-3\}$, put
$\pi_E(T)=\prod_{i\in E}(T+\tau_i)$ and
$e(E)=\sum_{i\in E}2^i$.
On polynomials of degree less than $|E|$, the Euler operator
$f(T)\mapsto Tf'(T)$ is idempotent.  Let $P_E$ be obtained from its
evaluation matrix in the Lagrange basis at $(\tau_i)_{i\in E}$ by left
multiplication by $S_E^{-1}$ and right multiplication by $S_E$, where
\begin{equation}
 \begin{gathered}
 S_E=\operatorname{diag}
 \bigl(\sqrt{\tau_i}\,\pi_E'(\tau_i)\bigr)_{i\in E},
 \qquad
 Q_E=I+P_E,\\
 D_U=\operatorname{diag}(U_i)_{i\in E},
 \qquad
 D_V=\operatorname{diag}(V_i)_{i\in E}.
 \end{gathered}
 \label{eq:s4-euler-projectors}
\end{equation}
The square roots are unique in the finite field.  Lagrange interpolation
gives
\begin{equation}
 (P_E)_{ij}=(Q_E)_{ij}
 =\frac{\sqrt{\tau_i\tau_j}}{\tau_i+\tau_j}
 \qquad(i\ne j),
 \label{eq:s4-euler-off-diagonal}
\end{equation}
so $P_E,Q_E$ are symmetric complementary idempotents.  Consequently
\begin{equation}
 \begin{aligned}
 &(Q_ED_U+P_ED_V)^T(P_ED_U+Q_ED_V)\\
 &\qquad=D_UQ_ED_V+D_VP_ED_U
 =\mathsf K[E,E],\\
 \det\mathsf K[E,E]
 &=\det(P_ED_U+Q_ED_V)\,
   \det(Q_ED_U+P_ED_V).
 \end{aligned}
 \label{eq:s4-principal-factorization}
\end{equation}
Each determinant factor is multilinear in the two variables indexed by every
$i\in E$ and equals one when $U_i=V_i=1$.  It is therefore
nonzero.  The maps $Q_i$ are invertible, so their tensor
action preserves nonzero multilinear forms.  After this invertible
tensor-product action is removed, the remaining Frobenius substitution sends a
basis monomial choosing the second coordinate on $T\subseteq E$ to
$x^{e(E)-e(T)}y^{e(T)}$.  Distinct $T$ have distinct binary sums, so this
substitution is injective on the multilinear tensor space.  Hence
$\det\mathsf K[E,E]\ne0$ and
$\deg\det\mathsf K[E,E]=2e(E)$.

There are exactly $2^{r-3}=R/8$ even subsets of
$\{0,\ldots,r-3\}$.  Their binary sums $e(E)$ are distinct and less than
$R/4$.  Multiplication by each principal minor is injective on the
preceding $M^3$-dimensional product space.  The degree bounds above show
that the resulting spaces have distinct homogeneous degrees below $R$,
so, for a fixed constant $c_0>0$, their direct sum has dimension
$\frac R8M^3\ge c_0R^2$.
Every element in this direct sum is a product of seven minors.  Up to the
same fixed nonzero normalization factor $K(\alpha)^{14}$, it is therefore
one of the normalized zero-cocycles of degree $d_*e_i$ in direction $i$
from Eq.~\eqref{eq:s4-normalized-generator-product}.  The two closing equations
of Eq.~\eqref{eq:s4-cyclic-frobenius-system} have leading monomials
$x^R,y^R$.  Since every word has positive total degree less than $R$,
reduction by those equations cannot create a relation, and deleting the
zero solution does not lower the rank.

Let $V_R$ be the $\mathbb L_R$-space spanned by the normalized
evaluations on $\Omega_R$ of the individual degree-$d_*e_i$
expressions in Eq.~\eqref{eq:s4-generator-products}.  Thus
$V_R^{s_*}=\rho_\alpha(W_R)$.
The direct sum above lies in $V_R$, so $\dim V_R\ge c_0R^2$.  The space
$V_R$ also contains $1$ and the coordinate functions, and hence separates
the nonzero solutions.

\paragraph{Coordinate indicators and descent.}
For a field $F$ and a subspace $V\subseteq F^\Omega$ containing $1$ and separating the
$n=|\Omega|$ points, let $V^s$ denote the span of products of $s$ elements
of $V$.  The coordinatewise Kneser inequality
\begin{equation}
 \dim(UW)\ge\dim U+\dim W-\dim\operatorname{St}(UW),
 \qquad
 \operatorname{St}(Y)=\{g:gY\subseteq Y\},
 \label{eq:s4-coordinate-kneser}
\end{equation}
holds for subspaces of $F^\Omega$.  We need the following consequence.

\begin{lemma}
\label{lem:s4-point-indicators}
If $d=\dim V$ and $s\ge2$, then $V^s$ contains at least
\begin{equation}
 d-\frac{2(n-d)}s
 \label{eq:s4-point-indicator-count}
\end{equation}
coordinate indicators.
\end{lemma}

\begin{proof}
Let $W=V^s$ and $H=\operatorname{St}(W)$.  Since $1\in W$, $H$ is a
unital subalgebra of $W$.  Every unital subalgebra of $F^\Omega$ is the
space of functions constant on the blocks of a partition of $\Omega$:
products of normalized functions separating one block from all others give
the block indicators.  A coordinate indicator belongs to $W$ exactly when
its block is a singleton.

On a nonsingleton block $B$, write $b_B=|B|$ and
$d_B=\dim(V|_B)$.  Point separation gives $d_B\ge2$.  The stabilizer of
$W|_B$ consists of the constants, since a nonconstant stabilizer would
split $B$.  Stabilizers of the successive powers of $V|_B$ are increasing,
so the coordinatewise Kneser inequality gives
\begin{equation}
 \dim(W|_B)\ge s(d_B-1)+1.
\end{equation}
The left side is smaller than $b_B$.  Since $d_B\ge2$, this first
forces $b_B\ge s+2$; then
\begin{equation}
 d_B\le 1+\frac{b_B-2}{s}
 \le\frac{2b_B}{s+2}.
\end{equation}
If $a_s$ is the number of singleton blocks,
restriction to all blocks is injective on $V$, and therefore
\begin{equation}
 d\le a_s+\frac{2(n-a_s)}{s+2}.
\end{equation}
Rearranging proves Eq.~\eqref{eq:s4-point-indicator-count}.
\end{proof}

Fix $s_*$ large enough that
$c_1:=c_0-2(1-c_0)/s_*>0$.  Apply the lemma to $V_R$, with
$n=R^2-1$, $d\ge c_0R^2$, and $s=s_*$.  The evaluation image
$\rho_\alpha(W_R)=V_R^{s_*}$
therefore contains the indicators of a set $I\subseteq\Omega_R$ with
\begin{equation}
 |I|\ge c_1R^2
 \label{eq:s4-indicator-dimension}
\end{equation}
Take $I$ to be the full set of points whose indicators occur in this
image.

Put $\Gamma_R=\operatorname{Gal}(\mathbb L_R/\mathbb E)$.  Write
$|\mathbb E|=2^f$.  Raising coefficients to the $2^f$-th power
and shifting every cyclic index by $f$ preserves the cyclic equations,
the space of formal products, and both maps in
Eq.~\eqref{eq:s4-formal-product-maps}.  Indeed it sends $f_j,t_j$, and
$D_{I,J}$ to the corresponding generators with shifted indices and
permutes the complete auxiliary interpolation sum in
Proposition~\ref{prop:s4-global-minor-sections}.  On geometric
solutions it acts by the inverse cyclic shift followed by the
$2^f$-th power.  Hence the full set $I$ is Galois stable, and the span
of its indicators descends to a finite \'{e}tale
$\mathbb E$-algebra $\mathscr B$ of dimension $|I|$.

For completeness, this descent does not divide by the extension degree.
On an orbit of length $a$, the invariant vectors are parametrized by
the degree-$a$ subfield of $\mathbb L_R$ and span the orbit after scalar
extension.  Equivalently, if $(\theta_j)$ is an $\mathbb E$-basis of
$\mathbb L_R$ and $(\theta_j^*)$ its trace-dual basis, then
\begin{equation}
 v_j=
 \sum_{\sigma\in\operatorname{Gal}(\mathbb L_R/\mathbb E)}
 \sigma(\theta_j^*v)
 \label{eq:s4-trace-descent}
\end{equation}
belongs to every invariant subspace containing $v$, and
$v=\sum_j\theta_jv_j$.  Thus invariant subspaces, their kernels, and
their quotient maps descend exactly.

Let $W_R^{\Gamma_R}$ be the invariant subspace of
$W_R$.  The
restricted map
$\rho_\alpha:
\rho_\alpha^{-1}(\mathscr B)\to\mathscr B$
is surjective.  Choose an $\mathbb E$-linear section
$s:\mathscr B\to\rho_\alpha^{-1}(\mathscr B)
\subseteq W_R^{\Gamma_R}$ and define
\begin{equation}
 j=\rho_X\circ s:
 \mathscr B\longrightarrow Z^0(X,\mathcal P).
 \label{eq:s4-linear-lift}
\end{equation}
No multiplicativity of $j$ is required.

We next verify that relations among global cochains do not identify two
selected functions.  If $w\in W_R$ satisfies
$\rho_X(w)=0$ and $x\in I$, choose
$u_x\in W_R$ with
$\rho_\alpha(u_x)=\delta_x$.  Inserting
$(w,u_x,u_x,u_x)$ into Eq.~\eqref{eq:s4-fourth-product-sum} gives
$\rho_\alpha(w)(x)=0$.  Consequently $j$ is injective.
Applying Eq.~\eqref{eq:s4-fourth-product-sum} again gives
\begin{equation}
 \begin{aligned}
 &\left\langle
 A^{\ell-d_{\mathrm{prod}}}j(b_1)\gamma_1\smile
 A^{\ell-d_{\mathrm{prod}}}j(b_2)\gamma_2\smile
 A^{\ell-d_{\mathrm{prod}}}j(b_3)\gamma_3,
 j(b_4)G\frown[X]
 \right\rangle\\
 &\qquad=
 \operatorname{Tr}_{\mathscr B/\mathbb E}(b_1b_2b_3b_4).
 \end{aligned}
 \label{eq:s4-four-linear-trace}
\end{equation}
The trace pairing on a finite \'{e}tale algebra is nondegenerate.

\paragraph{Independent pairings.}
Let $q_{\mathbb E}=|\mathbb E|$.  The standard
multiplication-friendly embedding construction~\cite{Golowich_Guruswami2025A},
applied to each field factor of the finite \'{e}tale algebra $\mathscr B$,
gives $\mathbb E$-linear maps
\begin{equation}
 \phi:\mathbb E^m\longrightarrow\mathscr B,
 \qquad
 \psi:\mathscr B\longrightarrow\mathbb E^m
 \label{eq:s4-rmfe-maps}
\end{equation}
such that
\begin{equation}
 \psi(\phi(x)\phi(y)\phi(z))=x*y*z,
 \qquad
 m\ge
 \frac{\dim_{\mathbb E}\mathscr B}
 {18q_{\mathbb E}(q_{\mathbb E}+1)}.
 \label{eq:s4-etale-rmfe}
\end{equation}

Let $e_1,\ldots,e_m$ be the standard basis of $\mathbb E^m$.  For every
$t$, trace duality gives a unique $p_t\in\mathscr B$ such that
\begin{equation}
 \psi_t(b)=\operatorname{Tr}_{\mathscr B/\mathbb E}(p_tb).
 \label{eq:s4-trace-dual}
\end{equation}
Define
\begin{equation}
 P_t=j(\phi(e_t)),
 \qquad
 H_t=j(p_t)G.
 \label{eq:s4-PH-definition}
\end{equation}
Equations~\eqref{eq:s4-four-linear-trace},
\eqref{eq:s4-etale-rmfe}, and \eqref{eq:s4-trace-dual} give
Eq.~\eqref{eq:s4-selector-identity}.  Finally,
Eqs.~\eqref{eq:s4-indicator-dimension} and
\eqref{eq:s4-etale-rmfe} give Eq.~\eqref{eq:s4-linear-count}.
This proves Theorem~\ref{thm:s4-diagonal-pairing}.

\subsubsection{Poincar\'e duality and logical addressability}
\label{subsec:s4-addressability}

Define
\begin{equation}
 u_{a,t}=u_a(P_t)
 \in Z^2(X,\mathcal F_a),
 \qquad
 \xi_t=H_t\frown[X].
 \label{eq:s4-logical-cocycles-and-cycles}
\end{equation}
Since $H_t$ is a zero-cocycle of the coordinate-dual sheaf,
Poincar\'e duality in Eq.~\eqref{eq:s4-poincare}, or directly
Lemma~\ref{lem:s4-cap-cycle}, gives
$\xi_t\in Z_6(X,\mathcal F_1\otimes\mathcal F_2\otimes\mathcal F_3)$.
Equation~\eqref{eq:s4-selector-identity} now reads
\begin{equation}
 \left\langle
 u_{1,i}\smile u_{2,j}\smile u_{3,k},\xi_t
 \right\rangle
 =\delta_{it}\delta_{jt}\delta_{kt}.
 \label{eq:s4-diagonal-pairing}
\end{equation}

This identity first proves the rate.  If
$\sum_i\eta_i[u_{1,i}]=0$ in $H^2(X,\mathcal F_1)$, pairing with
$[u_{2,t}]$, $[u_{3,t}]$, and $\xi_t$ gives $\eta_t=0$.  The same
argument applies in the other two positions.  Hence
\begin{equation}
 \dim_{\mathbb E}H^2(X,\mathcal F_a)\ge m
 \qquad(a=1,2,3).
 \label{eq:s4-three-rates}
\end{equation}
The constant-rate conclusion is therefore a consequence of the explicit
cup product pairing.  It is not obtained from an independent dimension
count for any of the three codes.

For $S\subseteq[m]$, set $\xi_S=\sum_{t\in S}\xi_t$.  For arbitrary $a_i,b_i,c_i\in\mathbb E$,
Eq.~\eqref{eq:s4-diagonal-pairing} gives
\begin{equation}
 \begin{aligned}
 &\left\langle
 \left(\sum_i a_i u_{1,i}\right)
 \smile
 \left(\sum_j b_j u_{2,j}\right)
 \smile
 \left(\sum_k c_k u_{3,k}\right),
 \xi_S
 \right\rangle
 =\sum_{t\in S}a_tb_tc_t.
 \end{aligned}
 \label{eq:s4-subset-pairing}
\end{equation}
Thus the three code blocks and the selected logical coordinates are fixed,
while changing $S$ changes only the Poincar'e-dual cycle.  This is the
cohomological form of logical addressability.

Fix a binary basis of $\mathbb E$ beginning with $1$, use the binary
$1$-component of each selected $\mathbb E$-coordinate, and initialize the
remaining binary components and the complementary logical coordinates to
computational zero.  Choose $\theta\in\mathbb E$ with
$\operatorname{Tr}_{\mathbb E/\mathbb F_2}(\theta)=1$.  For
computational-basis labels $y_a\in C^2(X,\mathcal F_a)$, define
\begin{equation}
 U_{\xi_S}|y_1,y_2,y_3\rangle
 =
 (-1)^{
 \operatorname{Tr}_{\mathbb E/\mathbb F_2}
 \left(\theta\langle y_1\smile y_2\smile y_3,\xi_S\rangle\right)}
 |y_1,y_2,y_3\rangle.
 \label{eq:s4-physical-gate}
\end{equation}
Proposition~\ref{prop:qltc-cup-invariant} shows that this unitary preserves
the complete three-block code space.  On the selected logical qubits,
\begin{equation}
 U_{\xi_S}=\prod_{t\in S}\CCZ_t.
 \label{eq:s4-addressed-ccz}
\end{equation}
Taking $S=[m]$ implements the product of all $m$ logical $\CCZ$
gates on pairwise distinct triples of the selected logical coordinates.
Equivalently, Eq.~\eqref{eq:s4-subset-pairing} restricts the trilinear
logical form to
$(a,b,c)\mapsto\sum_{t=1}^m a_tb_tc_t$, so its tensor subrank is at
least $m$.  This proves logical parallelizability; the standard
bounded-incidence repetition reduction used for a physically transversal
implementation is invoked in the next subsection.

Theorem~\ref{thm:s4-diagonal-pairing} gives $m=\Omega(R^2)$.
Together with $N_a=\Theta(R^2)$ below, the same pairing proves linear
addressability, linear parallelizability, and constant rate for all three
code blocks.

\subsection{Code parameters}
\label{subsec:s4-code-parameters}

\paragraph{Blocklength.}
For $a\in\{1,2,3\}$, let
$N_a=(\log_2|\mathbb E|)\dim_{\mathbb E}C^2(X,\mathcal F_a)$ be the
binary blocklength.  Since the local fields, the coefficient field, and
the local stalk dimensions are fixed, the cell count in
Eq.~\eqref{eq:s4-arithmetic-cell-count} gives $N_a=\Theta(R^2)$ for all
three code blocks.

\paragraph{Rate.}
Theorem~\ref{thm:s4-diagonal-pairing} gives
$m=\Omega(R^2)$ classes in each code block.  The pairing in
Eq.~\eqref{eq:s4-diagonal-pairing} isolates every one of these classes,
so Eq.~\eqref{eq:s4-three-rates} gives
$\dim_{\mathbb E}H^2(X,\mathcal F_a)\ge m$ for
$a\in\{1,2,3\}$.  Together with $N_a=\Theta(R^2)$, this proves constant
rate for all three codes.  Thus the rate follows from the explicit
cup-product pairing, rather than from a separate dimension count like Ref~\cite{GJ2026,BLN2026,CHLT2026}.

\paragraph{Distance and soundness.}
The local degrees in Eq.~\eqref{eq:s4-code-local-codes} give primal and
coordinate-dual PRS codes whose rates are bounded away from zero and one.
These codes satisfy the required uniform product-expansion conditions by
the projective Reed--Solomon theorem of Ref.~\cite{GJ2026}; Frobenius
permutations and nonzero coordinate scalings preserve those conditions.
The fixed arithmetic data can be chosen to satisfy simultaneously the
nonvanishing conditions of Proposition~\ref{prop:s4-fixed-period} and
the Ramanujan threshold recorded in
Eq.~\eqref{eq:s4-arithmetic-spectrum}.  The distance-and-soundness theorem
for the good-qLTC construction therefore applies directly
\cite[Lemma~5.4]{BLN2026}\cite[Sections~6--8]{GJ2026}.  After restriction
of scalars to $\mathbb F_2$, each code has bounded check weight and qubit
incidence, linear $X$- and $Z$-distance, and constant soundness; all
changes are by constants depending only on the fixed field $\mathbb E$.

\paragraph{Transversal implementation.}
Equation~\eqref{eq:s4-subset-pairing} and
$m=\Omega(R^2)=\Omega(N_a)$ give linearly many addressable and
parallelizable logical $\CCZ$ gates.  In fixed binary bases, every physical
coordinate occurs in only boundedly many cubic monomials.  The standard
bounded-incidence repetition construction therefore gives a fixed
constant-overhead encoding on which every addressed product
$U_{\xi_S}$ is implemented by one layer of pairwise disjoint physical
$\CCZ$ gates~\cite{Nguyen2025CCZ,Golowich_Lin2024}.  The repetition
encoding changes all qLDPC and qLTC parameters by constant factors only.

Combining these conclusions gives the good qLDPC and qLTC families with
the rate, distance, soundness, addressability, parallelizability, and
transversal logical gates claimed in Theorem~\ref{thm:main}.

\subsection{Applications to magic state distillation}

Finally, these codes can be applied to $\CCZ$ magic-state distillation in
the standard ideal-Clifford setting.  Theorem~\ref{thm:s4-diagonal-pairing}
and the transversal implementation above provide
$k_{\CCZ}=\Omega(N)$ independent logical $\CCZ$ gates, implemented
using $M=O(N)$ physical $\CCZ$ gates, while the minimum distance of the
three code blocks is $d=\Theta(N)$.  The standard postselected
distillation protocol therefore has average overhead
exponent~\cite{BravyiHaah2012,HastingsHaah2018}
\begin{equation}
 \gamma
 =\frac{\log(M/k_{\CCZ})}{\log d}
 =o(1),
 \qquad N\to\infty.
 \label{eq:s4-msd-overhead-exponent}
\end{equation}
Thus, for any $\eta>0$, choosing a sufficiently large code in the family
and concatenating yields magic states with error at most $\epsilon$ using
$O(\log^{\eta}(1/\epsilon))$ raw magic states per output, for
sufficiently small input noise.  This sublogarithmic overhead is achieved
while the underlying code blocks simultaneously retain bounded-weight
checks, constant rate, linear distance, and constant soundness.

\section*{Acknowledgments} 
This work is supported in part by NSFC under Grant No.~12475023, Dushi Program, and a startup funding from YMSC.

\section*{AI disclosure}
The authors formulated the problem and overall strategies. Generative AI tools substantially assisted in developing and completing the proof details and drafting the manuscript under extensive instructions from the authors. The authors have reviewed and verified all arguments and take responsibility for the results.
\endgroup
\printbibliography[heading=bibintoc,title=References]

@article{Gorodetsky2020,
  author        = {Gorodetsky, Ofir},
  title         = {Mean values of arithmetic functions in short intervals
                   and in arithmetic progressions in the large-degree limit},
  journal       = {Mathematika},
  year          = {2020},
  volume        = {66},
  number        = {2},
  pages         = {373--394},
  doi           = {10.1112/mtk.12032},
  eprint        = {1810.00483},
  archivePrefix = {arXiv},
  primaryClass  = {math.NT}
}

@misc{Panteleev2024,
      title={Maximally Extendable Sheaf Codes}, 
      author={Pavel Panteleev and Gleb Kalachev},
      year={2024},
      eprint={2403.03651},
      archivePrefix={arXiv},
      primaryClass={cs.IT},
      url={https://arxiv.org/abs/2403.03651}, 
}

@misc{KP2025Extendable,
	title={Maximally Extendable Product Codes are Good Coboundary Expanders}, 
	author={Gleb Kalachev and Pavel Panteleev},
	year={2025},
	eprint={2501.01411},
	archivePrefix={arXiv},
	primaryClass={cs.IT},
	url={https://arxiv.org/abs/2501.01411}, 
}

@INPROCEEDINGS{QuantumTanner2022,
	author={Leverrier, Anthony and Zémor, Gilles},
	booktitle={2022 IEEE 63rd Annual Symposium on Foundations of Computer Science (FOCS)}, 
	title={Quantum Tanner codes}, 
	year={2022},
	volume={},
	number={},
	pages={872-883},
	doi={10.1109/FOCS54457.2022.00117}
}

@ARTICLE{Burton2022,
	author={Burton, Simon and Browne, Dan},
	journal={IEEE Transactions on Information Theory}, 
	title={Limitations on Transversal Gates for Hypergraph Product Codes}, 
	year={2022},
	volume={68},
	number={3},
	pages={1772-1781},
	doi={10.1109/TIT.2021.3131043}}

@inproceedings{Golowich_Lin2024,
	author = {Golowich, Louis and Lin, Ting-Chun},
	title = {Quantum {LDPC} Codes with Transversal Non-Clifford Gates via Products of Algebraic Codes},
	year = {2025},
	booktitle = {Proceedings of the 57th Annual ACM Symposium on Theory of Computing},
	isbn = {9798400715105},
	publisher = {Association for Computing Machinery},
	url = {https://doi.org/10.1145/3717823.3718139},
	doi = {10.1145/3717823.3718139},
	pages = {689–696},
	numpages = {8},
	location = {Prague, Czechia},
	series = {STOC '25}
}

@inproceedings{Golowich_Guruswami2025A,
	author = {Golowich, Louis and Guruswami, Venkatesan},
	title = {Asymptotically Good Quantum Codes with Transversal Non-Clifford Gates},
	year = {2025},
	isbn = {9798400715105},
	publisher = {Association for Computing Machinery},
	url = {https://doi.org/10.1145/3717823.3718234},
	doi = {10.1145/3717823.3718234},
	booktitle = {Proceedings of the 57th Annual ACM Symposium on Theory of Computing},
	pages = {707–717},
	numpages = {11},
	location = {Prague, Czechia},
	series = {STOC '25}
}

@INPROCEEDINGS{Golowich_Guruswami2025B,
	author={Golowich, Louis and Guruswami, Venkatesan},
	booktitle={2025 IEEE 66th Annual Symposium on Foundations of Computer Science (FOCS)}, 
	title={Near-Asymptotically-Good Quantum Codes with Transversal CCZ Gates and Sublinear-Weight Parity-Checks}, 
	year={2025},
	volume={},
	number={},
	pages={1561-1569},
	doi={10.1109/FOCS63196.2025.00082}
}

@misc{Fu2025nogo,
	title={No-go theorems for logical gates on product quantum codes}, 
	author={Esther Xiaozhen Fu and Han Zheng and Zimu Li and Zi-Wen Liu},
	year={2025},
	eprint={2507.16797},
	archivePrefix={arXiv},
	primaryClass={quant-ph},
	url={https://arxiv.org/abs/2507.16797}, 
}

@inproceedings{PK2022Good,
	author = {Panteleev, Pavel and Kalachev, Gleb},
	title = {Asymptotically good Quantum and locally testable classical {LDPC} codes},
	year = {2022},
	isbn = {9781450392648},
	publisher = {Association for Computing Machinery},
	url = {https://doi.org/10.1145/3519935.3520017},
	doi = {10.1145/3519935.3520017},
	booktitle = {Proceedings of the 54th Annual ACM SIGACT Symposium on Theory of Computing},
	pages = {375–388},
	numpages = {14},
	location = {Rome, Italy},
	series = {STOC 2022}
}

@inproceedings{DHLV2022,
	author = {Dinur, Irit and Hsieh, Min-Hsiu and Lin, Ting-Chun and Vidick, Thomas},
	title = {Good Quantum {LDPC} Codes with Linear Time Decoders},
	year = {2023},
	isbn = {9781450399135},
	publisher = {Association for Computing Machinery},
	url = {https://doi.org/10.1145/3564246.3585101},
	doi = {10.1145/3564246.3585101},
	booktitle = {Proceedings of the 55th Annual ACM Symposium on Theory of Computing},
	pages = {905–918},
	numpages = {14},
	location = {Orlando, FL, USA},
	series = {STOC 2023}
}

@INPROCEEDINGS{Dinur2024sheaf,
	author={Dinur, Irit and Lin, Ting-Chun and Vidick, Thomas},
	booktitle={2024 IEEE 65th Annual Symposium on Foundations of Computer Science (FOCS)}, 
	title={Expansion of High-Dimensional Cubical Complexes: with Application to Quantum Locally Testable Codes}, 
	year={2024},
	volume={},
	number={},
	pages={379-385},
	doi={10.1109/FOCS61266.2024.00031}
}

@misc{Lin2024transversal,
	title={Transversal non-Clifford gates for quantum {LDPC} codes on sheaves}, 
	author={Ting-Chun Lin},
	year={2024},
	eprint={2410.14631},
	archivePrefix={arXiv},
	primaryClass={quant-ph},
	url={https://arxiv.org/abs/2410.14631}, 
}

@misc{Li2025Poincare,
	title={Poincar\'e Duality and Multiplicative Structures on Quantum Codes}, 
	author={Yiming Li and Zimu Li and Zi-Wen Liu and Quynh T. Nguyen},
	year={2025},
	eprint={2512.21922},
	archivePrefix={arXiv},
	primaryClass={quant-ph},
	url={https://arxiv.org/abs/2512.21922}, 
}

@misc{LSWLL2026Theory,
      title={Theory of (Co)homological Invariants on Quantum {LDPC} Codes}, 
      author={Zimu Li and Yuguo Shao and Fuchuan Wei and Yiming Li and Zi-Wen Liu},
      year={2026},
      eprint={2603.25831},
      archivePrefix={arXiv},
      primaryClass={quant-ph},
      url={https://arxiv.org/abs/2603.25831}, 
}

@misc{LLL2026nontrivial,
      title={Transversal non-Clifford gates on almost-good quantum {LDPC} and quantum locally testable codes}, 
      author={Yiming Li and Zimu Li and Zi-Wen Liu},
      year={2026},
      eprint={2604.01874},
      version={2},
      archivePrefix={arXiv},
      primaryClass={quant-ph},
      note={Version 2, revised 2 July 2026},
      url={https://arxiv.org/abs/2604.01874v2}, 
}

@Article{10.21468/SciPostPhys.14.4.065,
	title={{Codimension-2 defects and higher symmetries in (3+1)D topological phases}},
	author={Maissam Barkeshli and Yu-An Chen and Sheng-Jie Huang and Ryohei Kobayashi and Nathanan Tantivasadakarn and Guanyu Zhu},
	journal={SciPost Phys.},
	volume={14},
	pages={065},
	year={2023},
	publisher={SciPost},
	doi={10.21468/SciPostPhys.14.4.065},
	url={https://scipost.org/10.21468/SciPostPhys.14.4.065},
}

@article{Wang_2024,
	title={Efficient fault-tolerant implementations of non-Clifford gates with reconfigurable atom arrays},
	volume={10},
	ISSN={2056-6387},
	url={http://dx.doi.org/10.1038/s41534-024-00945-3},
	DOI={10.1038/s41534-024-00945-3},
	number={1},
	journal={npj Quantum Information},
	publisher={Springer Science and Business Media LLC},
	author={Wang, Yifei and Wang, Yixu and Chen, Yu-An and Zhang, Wenjun and Zhang, Tao and Hu, Jiazhong and Chen, Wenlan and Gu, Yingfei and Liu, Zi-Wen},
	year={2024},
	month=dec 
}

@article{Breuckmann2024Cups,
  author = {Breuckmann, Nikolas P. and Davydova, Margarita and Eberhardt, Jens N. and Tantivasadakarn, Nathanan},
  title = {Cups and Gates {I}: Cohomology Invariants and Logical Quantum Operations},
  journal = {Communications in Mathematical Physics},
  volume = {407},
  number = {5},
  eid = {86},
  year = {2026},
  month = may,
  doi = {10.1007/s00220-026-05570-z},
  eprint = {2410.16250},
  archivePrefix = {arXiv},
  primaryClass = {quant-ph},
  url = {https://link.springer.com/article/10.1007/s00220-026-05570-z}
}

@misc{RainbowCode,
	title={Quantum Rainbow Codes: Achieving Linear Rate, Growing Distance and Transversal Non-Clifford Gates with Generalised Colour Codes}, 
	author={Thomas R. Scruby and Arthur Pesah and Mark Webster},
	year={2025},
	eprint={2408.13130},
	archivePrefix={arXiv},
	primaryClass={quant-ph},
	url={https://arxiv.org/abs/2408.13130}, 
}

@article{Zhu2023,
	title = {Non-Clifford and Parallelizable Fault-Tolerant Logical Gates on Constant and Almost-Constant Rate Homological Quantum Low-Density Parity-Check Codes via Higher Symmetries},
	author = {Zhu, Guanyu and Sikander, Shehryar and Portnoy, Elia and Cross, Andrew W. and Brown, Benjamin J.},
	journal = {PRX Quantum},
	volume = {6},
	issue = {4},
	pages = {040361},
	numpages = {43},
	year = {2025},
	month = 12,
	publisher = {American Physical Society},
	doi = {10.1103/wcxs-w69t},
	url = {https://link.aps.org/doi/10.1103/wcxs-w69t}
}

@misc{Zhu2025A,
	title={A topological theory for {qLDPC}: non-Clifford gates and magic state fountain on homological product codes with constant rate and beyond the $N^{1/3}$ distance barrier}, 
	author={Guanyu Zhu},
	year={2025},
	eprint={2501.19375},
	archivePrefix={arXiv},
	primaryClass={quant-ph},
	url={https://arxiv.org/abs/2501.19375}, 
}

@misc{Zhu2025B,
	title={Transversal non-Clifford gates on {qLDPC} codes breaking the $\sqrt{N}$ distance barrier and quantum-inspired geometry with $\mathbb{Z}_2$ systolic freedom}, 
	author={Guanyu Zhu},
	year={2025},
	eprint={2507.15056},
	archivePrefix={arXiv},
	primaryClass={quant-ph},
	url={https://arxiv.org/abs/2507.15056}, 
}

@article{Bombin_2007,
	title={Topological Computation without Braiding},
	volume={98},
	ISSN={1079-7114},
	url={http://dx.doi.org/10.1103/PhysRevLett.98.160502},
	DOI={10.1103/physrevlett.98.160502},
	number={16},
	journal={Physical Review Letters},
	publisher={American Physical Society (APS)},
	author={Bombin, H. and Martin-Delgado, M. A.},
	year={2007},
	month=apr 
}

@misc{Bombin_2013,
	title={Gauge Color Codes: Optimal Transversal Gates and Gauge Fixing in Topological Stabilizer Codes}, 
	author={Bombín, Héctor},
	year={2015},
	eprint={1311.0879},
	archivePrefix={arXiv},
	primaryClass={quant-ph},
	url={https://arxiv.org/abs/1311.0879}, 
}

@article{Kubica2015,
	title={Universal transversal gates with color codes: A simplified approach},
	volume={91},
	ISSN={1094-1622},
	url={http://dx.doi.org/10.1103/PhysRevA.91.032330},
	DOI={10.1103/physreva.91.032330},
	number={3},
	journal={Physical Review A},
	publisher={American Physical Society (APS)},
	author={Kubica, Aleksander and Beverland, Michael E.},
	year={2015},
	month=mar 
}

@article{Gottesman2013,
	author = "Gottesman, Daniel",
	title = "{Fault-Tolerant Quantum Computation with Constant Overhead}",
	eprint = "1310.2984",
	archivePrefix = "arXiv",
	primaryClass = "quant-ph",
	month = "10",
	year = "2013"
}

@article{Yamasaki2024QusiPloylog,
	title={Time-Efficient Constant-Space-Overhead Fault-Tolerant Quantum Computation},
	volume={20},
	ISSN={1745-2481},
	url={http://dx.doi.org/10.1038/s41567-023-02325-8},
	DOI={10.1038/s41567-023-02325-8},
	number={2},
	journal={Nature Physics},
	publisher={Springer Science and Business Media LLC},
	author={Yamasaki, Hayata and Koashi, Masato},
	year={2024},
	month=jan, 
	pages={247–253} 
}

@article{Tamiya2024PolylogTime,
	title={Fault-tolerant quantum computation with polylogarithmic time and constant space overheads},
	author={Shiro Tamiya and Masato Koashi and Hayata Yamasaki},
        volume={22},
	ISSN={1745-2481},
	url={ https://doi.org/10.1038/s41567-025-03102-5},
	DOI={10.1038/s41567-025-03102-5},
	number={1},
	journal={Nature Physics},
	publisher={Springer Science and Business Media LLC},
	year={2026},
	month=jan, 
	pages={27–32}  
}

@inproceedings{Nguyen2025FT,
	author = {Nguyen, Quynh T. and Pattison, Christopher A.},
	title = {Quantum Fault Tolerance with Constant-Space and Logarithmic-Time Overheads},
	year = {2025},
	isbn = {9798400715105},
	publisher = {Association for Computing Machinery},
	url = {https://doi.org/10.1145/3717823.3718318},
	doi = {10.1145/3717823.3718318},
	booktitle = {Proceedings of the 57th Annual ACM Symposium on Theory of Computing},
	pages = {730–737},
	numpages = {8},
	location = {Prague, Czechia},
	series = {STOC '25}
}

@article{Wills2024magic,
  author = {Wills, Adam and Hsieh, Min-Hsiu and Yamasaki, Hayata},
  title = {Constant-overhead magic state distillation},
  journal = {Nature Physics},
  volume = {21},
  number = {11},
  pages = {1842--1846},
  year = {2025},
  doi = {10.1038/s41567-025-03026-0},
  url = {https://www.nature.com/articles/s41567-025-03026-0}
}

@inproceedings{Nguyen2025CCZ,
	author = {Nguyen, Quynh T.},
	title = {Good Binary Quantum Codes with Transversal CCZ Gate},
	year = {2025},
	isbn = {9798400715105},
	publisher = {Association for Computing Machinery},
	url = {https://doi.org/10.1145/3717823.3718186},
	doi = {10.1145/3717823.3718186},
	booktitle = {Proceedings of the 57th Annual ACM Symposium on Theory of Computing},
	pages = {697–706},
	numpages = {10},
	location = {Prague, Czechia},
	series = {STOC '25}
}

@misc{He2025addressable,
	title={Quantum Codes with Addressable and Transversal Non-Clifford Gates}, 
	author={Zhiyang He and Vinod Vaikuntanathan and Adam Wills and Rachel Yun Zhang},
	year={2025},
	eprint={2502.01864},
	archivePrefix={arXiv},
	primaryClass={quant-ph},
	url={https://arxiv.org/abs/2502.01864}, 
}

@misc{Tiew2026copycup,
      title={Copy-cup Gates in Tensor Products of Group Algebra Codes}, 
      author={Ryan Tiew and Nikolas P. Breuckmann},
      year={2026},
      eprint={2602.23307},
      archivePrefix={arXiv},
      primaryClass={quant-ph},
      url={https://arxiv.org/abs/2602.23307}, 
}

@misc{VirgileCCZ,
      title={Good quantum codes with addressable and parallelizable transversal non-Clifford gates}, 
      author={Virgile Guémard},
      year={2025},
      eprint={2510.19809},
      archivePrefix={arXiv},
      primaryClass={quant-ph},
      url={https://arxiv.org/abs/2510.19809}, 
}

@inproceedings{Eldar2016NLETS,
	title={Local Hamiltonians Whose Ground States Are Hard to Approximate},
	author={Lior Eldar and Aram W. Harrow},
	booktitle={2017 IEEE 58th Annual Symposium on Foundations of Computer Science (FOCS)},
	year={2017},
	pages={427--438},
	doi={10.1109/FOCS.2017.46},
	eprint={1510.02082},
	archivePrefix={arXiv},
	primaryClass={quant-ph},
	url={https://arxiv.org/abs/1510.02082},
}

@article{AharonovAradVidick2013_qPCPSurvey,
	author  = {Aharonov, Dorit and Arad, Itai and Vidick, Thomas},
	title   = {The Quantum {PCP} Conjecture},
	journal = {SIGACT News},
	year    = {2013},
	volume  = {44},
	number  = {2},
	pages   = {47--79},
	doi     = {10.1145/2491533.2491549},
	url     = {https://doi.org/10.1145/2491533.2491549}
}

@misc{GJ2026,
      title={Asymptotically Good Quantum Locally Testable Codes}, 
      author={William Gay and Fernando Granha Jeronimo},
      year={2026},
      eprint={2609.20780},
      archivePrefix={arXiv},
      primaryClass={quant-ph},
      url={https://arxiv.org/abs/2609.20780}, 
}

@misc{BLN2026,
      title={Good Quantum Locally Testable Codes from Product Expansion}, 
      author={Mitali Bafna and Anqi Li and Quynh T. Nguyen},
      year={2026},
      eprint={2609.26735},
      archivePrefix={arXiv},
      primaryClass={quant-ph},
      url={https://arxiv.org/abs/2609.26735}, 
}

@misc{CHLT2026,
      title={Cubical Sheaf Complexes with Constant Expansion with Applications to Asymptotically Good qLTCs}, 
      author={Yeyuan Chen and Miryam Mi-Ying Huang and Yinchen Liu and Er-Cheng Tang},
      year={2026},
      eprint={2609.28028},
      archivePrefix={arXiv},
      primaryClass={cs.IT},
      url={https://arxiv.org/abs/2609.28028}, 
}

@article{RSV2019,
  author = {Nithi Rungtanapirom and Jakob Stix and Alina Vdovina},
  title = {Infinite series of quaternionic 1-vertex cube complexes, the doubling construction, and explicit cubical Ramanujan complexes},
  journal = {International Journal of Algebra and Computation},
  volume = {29},
  year = {2019},
  eprint = {1808.03290},
  archivePrefix = {arXiv}
}

@misc{He2025GoodAddressable,
  author        = {He, Zhiyang and Vaikuntanathan, Vinod and Wills, Adam and Zhang, Rachel Yun},
  title         = {Asymptotically Good Quantum Codes with Addressable and Transversal Non-{Clifford} Gates},
  year          = {2025},
  eprint        = {2507.05392},
  archivePrefix = {arXiv},
  primaryClass  = {quant-ph},
  url           = {https://arxiv.org/abs/2507.05392}
}

@misc{GasnierGuemard2026,
  author        = {Gasnier, Jean and Gu{\'e}mard, Virgile},
  title         = {Quantum group codes for non-{Clifford} logic: enhanced decoding, addressability and parallelizability},
  year          = {2026},
  eprint        = {2606.27211},
  archivePrefix = {arXiv},
  primaryClass  = {quant-ph},
  url           = {https://arxiv.org/abs/2606.27211}
}

@misc{GolowichTamoZhu2026,
  author        = {Golowich, Louis and Tamo, Itzhak and Zhu, Guanyu},
  title         = {Improved Transversal Non-{Clifford} Gates from Cup Products},
  year          = {2026},
  eprint        = {2609.27801},
  archivePrefix = {arXiv},
  primaryClass  = {quant-ph},
  url           = {https://arxiv.org/abs/2609.27801}
}

@misc{Wills2026TransversalT,
  author        = {Wills, Adam},
  title         = {Improved Quantum Codes with Transversal {$T$} Gates},
  year          = {2026},
  eprint        = {2608.24000},
  archivePrefix = {arXiv},
  primaryClass  = {quant-ph},
  url           = {https://arxiv.org/abs/2608.24000}
}

@misc{SanJose2026Triorthogonal,
  author        = {San-Jos{\'e}, Rodrigo},
  title         = {Asymptotically good binary triorthogonal codes and higher-level transversal gates},
  year          = {2026},
  eprint        = {2609.08203},
  archivePrefix = {arXiv},
  primaryClass  = {cs.IT},
  url           = {https://arxiv.org/abs/2609.08203}
}

@misc{hsieh2025explicitlosslessvertexexpanders,
  author        = {Hsieh, Jun-Ting and Lubotzky, Alexander and Mohanty, Sidhanth and Reiner, Assaf and Zhang, Rachel Yun},
  title         = {Explicit Lossless Vertex Expanders},
  year          = {2025},
  eprint        = {2504.15087},
  archivePrefix = {arXiv},
  primaryClass  = {math.CO},
  url           = {https://arxiv.org/abs/2504.15087}
}

@article{AharonovEldar2015QLTC,
  author = {Aharonov, Dorit and Eldar, Lior},
  title = {Quantum Locally Testable Codes},
  journal = {SIAM Journal on Computing},
  volume = {44},
  number = {5},
  pages = {1230--1262},
  year = {2015},
  doi = {10.1137/140975498},
  eprint = {1310.5664},
  archivePrefix = {arXiv},
  primaryClass = {quant-ph},
  url = {https://arxiv.org/abs/1310.5664}
}

@article{BravyiHaah2012,
  author = {Bravyi, Sergey and Haah, Jeongwan},
  title = {Magic-state distillation with low overhead},
  journal = {Physical Review A},
  volume = {86},
  number = {5},
  pages = {052329},
  year = {2012},
  doi = {10.1103/PhysRevA.86.052329},
  eprint = {1209.2426},
  archivePrefix = {arXiv},
  primaryClass = {quant-ph},
  url = {https://arxiv.org/abs/1209.2426}
}

@article{HastingsHaah2018,
  author = {Hastings, Matthew B. and Haah, Jeongwan},
  title = {Distillation with Sublogarithmic Overhead},
  journal = {Physical Review Letters},
  volume = {120},
  number = {5},
  pages = {050504},
  year = {2018},
  doi = {10.1103/PhysRevLett.120.050504},
  eprint = {1709.03543},
  archivePrefix = {arXiv},
  primaryClass = {quant-ph},
  url = {https://arxiv.org/abs/1709.03543}
}

@misc{Cohen2026qLTC_lossless,
      title={Good Quantum Locally Testable Codes from Lossless Cubical Complexes}, 
      author={Itay Cohen and Itai Leigh and Assaf Reiner and Amnon Ta-Shma and Elad Tzalik},
      year={2026},
      eprint={2610.00525},
      archivePrefix={arXiv},
      primaryClass={quant-ph},
      url={https://arxiv.org/abs/2610.00525}, 
}

\end{document}